\documentclass[10pt]{article}
\usepackage[margin=.8in,left=.8in]{geometry}

\usepackage{amsthm}
\usepackage{amssymb}
\usepackage{amsmath}
\usepackage{mathtools}
\usepackage{adjustbox}

\usepackage[table]{xcolor}

\usepackage{stmaryrd}    

\usepackage{xfrac}

\usepackage{tensor}

\usepackage{enumitem}
\setlist{nosep}

\usepackage{slashed}

\usepackage[safe]{tipa} 

\usepackage{hhline}

\usepackage{tikz}
\usetikzlibrary{
  calc,
  backgrounds, 
  decorations.pathmorphing, 
  decorations.markings,
  decorations.text,
  snakes
}
\usetikzlibrary{cd}

\tikzset{
  snake left/.style={
    rounded corners,
    to path={
      let \p1 = (\tikztostart.east),
          \p2 = (\tikztotarget.west),
          \p3 = ($(\p1)!0.5!(\p2)$),
          \n1 = {8pt} 
      in
      (\p1)
      -- (\x1 + \n1, \y1)
      -- (\x1 + \n1, \y3)
      -- (\x2 - \n1, \y3) \tikztonodes
      -- (\x2 - \n1, \y2)
      -- (\p2)
    }
  }
}

\tikzset{
  uphordown/.style={
    rounded corners,
    to path={
      let \p1 = (\tikztostart.north),
          \p2 = (\tikztotarget.north),
          \n1 = {max(\y1,\y2) + 8pt}
      in
      (\p1)
      -- (\x1, \n1)
      -- (\x2, \n1) \tikztonodes 
      -- (\p2)
    }
  }
}

\tikzset{
  downhorup/.style={
    rounded corners,
    to path={
      let \p1 = (\tikztostart.south),
          \p2 = (\tikztotarget.south),
          \n1 = {min(\y1,\y2) - 8pt}
      in
      (\p1)
      -- (\x1, \n1)
      -- (\x2, \n1) \tikztonodes 
      -- (\p2)
    }
  }
}

\tikzset{
  rightvertleft/.style={
    rounded corners,
    to path={
      let \p1 = (\tikztostart.east),
          \p2 = (\tikztotarget.east),
          \n1 = {max(\x1,\x2) + 8pt}
      in
      (\p1)
      -- (\n1, \y1)
      -- (\n1, \y2) \tikztonodes 
      -- (\p2)
    }
  }
}

\tikzset{
  leftvertright/.style={
    rounded corners,
    to path={
      let \p1 = (\tikztostart.west),
          \p2 = (\tikztotarget.west),
          \n1 = {min(\x1,\x2) - 8pt}
      in
      (\p1)
      -- (\n1, \y1)
      -- (\n1, \y2) \tikztonodes 
      -- (\p2)
    }
  }
}

\usepackage{mathrsfs} 
\DeclareMathAlphabet{\mathpzc}{OT1}{pzc}{m}{it} 

\usepackage{hyperref}
\hypersetup{
  colorlinks = true,
  linkcolor= darkblue, citecolor=darkblue            
}

\hypersetup{
        colorlinks=true,
        linkcolor=darkblue,
        filecolor=darkblue,      
        urlcolor=black,
        citecolor=darkblue,
        linktoc=page
    }

\usepackage{cleveref}

\definecolor{darkblue}{rgb}{0.05,0.25,0.65}
\definecolor{darkgreen}{RGB}{20,140,10}
\definecolor{lightgray}{rgb}{0.9,0.9,0.9}
\definecolor{darkorange}{RGB}{200,100,5}
\definecolor{darkyellow}{rgb}{.91,.91,0}
\definecolor{orangeii}{RGB}{200,100,5}
\definecolor{lightblue}{RGB}{243, 250, 255}

\newtheorem{theorem}{Theorem}[section]

\newtheorem{lemma}[theorem]{Lemma}

\newtheorem{proposition}[theorem]{Proposition}
\newtheorem{corollary}[theorem]{Corollary}
\theoremstyle{definition}
\newtheorem{definition}[theorem]{Definition}

\newtheorem{example}[theorem]{Example}

\newtheorem{remark}[theorem]{Remark}

\usepackage{accents}
\newlength{\dhatheight}

\newcommand{\cpt}{\hspace{.8pt}{\adjustbox{scale={.5}{.77}}{$\cup$} \{\infty\}}}

\let\PLAINthebibliography\thebibliography
\renewcommand\thebibliography[1]{
  \PLAINthebibliography{#1}
  \setlength{\parskip}{0.5pt}
  \setlength{\itemsep}{0.5pt plus .3ex}
}

\newcommand{\defneq}{\equiv}

\newcommand\bos[1]{\mathstrut\mkern2.5mu#1\mkern-14mu\raise1.7ex%
  \hbox{$\scriptstyle\rightsquigarrow$}}

\newcommand\bosonic[1]{\mathstrut\mkern2.5mu#1\mkern-14mu\raise1.7ex%
  \hbox{$\scriptstyle\rightsquigarrow$}}

\usepackage[cmtip,all]{xy}
\newcommand{\longsquiggly}{\xymatrix{{}\ar@{~>}[r]&{}}}

\usepackage{bbm} 

\usepackage[new]{old-arrows}   

\newcommand{\interval}[3]{
\draw[gray, line width=3.3]
  (#1+.2,#2) -- (#1+#3-.2,#2);
\draw[draw=gray,line width=1.3,fill=black]
  (#1+#3,#2) circle (.2);
\draw[draw=gray,line width=1.3,fill opacity=.6,fill=white]
  (#1,#2) circle (.2);
}

\newcommand{\oppositeinterval}[3]{
\draw[gray, line width=3.3]
  (#1+.2,#2) -- (#1+#3-.2,#2);
\draw[draw=gray,line width=1.3,fill=black]
  (#1,#2) circle (.2);
\draw[draw=gray,line width=1.3,fill opacity=.6,fill=white]
  (#1+#3,#2) circle (.2);
}

\newcommand{\openinterval}[3]{
\draw[gray, line width=3.3]
  (#1+.2,#2) -- (#1+#3-.2,#2);
\draw[draw=gray,line width=1.3,fill opacity=.6,fill=white]
  (#1+#3,#2) circle (.2);
\draw[draw=gray,line width=1.3,fill opacity=.6,fill=white]
  (#1,#2) circle (.2);
}

\newcommand{\closedinterval}[3]{
\draw[gray, line width=3.3]
  (#1+.2,#2) -- (#1+#3-.2,#2);
\draw[draw=gray,line width=1.3,fill=black]
  (#1+#3,#2) circle (.2);
\draw[draw=gray,line width=1.3,fill=black]
  (#1,#2) circle (.2);
}

\newcommand{\halfcircle}{
\begin{scope}
\clip
  (-4,0) rectangle (4.1,-4.1);
\draw[
  gray,
  line width=30pt,
  draw opacity=.4
]
  (0,0) circle (3);
\end{scope}

\draw[draw=gray,line width=1.3,fill=black]
  (0.1,-3.5) circle (.2);
\draw[draw=gray,line width=1.3,fill opacity=.6,fill=white]
  (-0.1,-3.5) circle (.2);

\draw[gray, line width=3.3]
  (-1.7+.2,-3.1) -- (1.7-.2,-3.1);
\draw[draw=gray,line width=1.3,fill=black]
  (1.7,-3.1) circle (.2);
\draw[draw=gray,line width=1.3,fill opacity=.6,fill=white]
  (-1.7,-3.1) circle (.2);

\closedinterval{.1}{-2.5}{2.35};
\openinterval{-2.5}{-2.5}{2.4};

\closedinterval{1.7}{-1.83}{1.34};
\openinterval{-3}{-1.83}{1.34};

\closedinterval{2.2}{-1.12}{1.12};
\openinterval{-3.31}{-1.12}{1.12};

\closedinterval{2.45}{-.38}{1.05};
\openinterval{-3.5}{-.38}{1.05};
}

\newcommand{\halfcircleAdjusted}{
\begin{scope}
\clip
  (-4,0) rectangle (4.1,-4.1);
\draw[
  gray,
  line width=30pt,
  draw opacity=.4
]
  (0,0) circle (3);
\end{scope}

\draw[draw=gray,line width=1.3,fill=black]
  (0.1,-3.5) circle (.2);
\draw[draw=gray,line width=1.3,fill opacity=.6,fill=white]
  (-0.1,-3.5) circle (.2);
  
\draw[gray, line width=3.3]
  (-1.8+.2,-3.0) -- 
  (+1.8-.2,-3.0);
\draw[draw=gray,line width=1.3,fill=black]
  (1.8,-3) circle (.2);
\draw[draw=gray,line width=1.3,fill opacity=.6,fill=white]
  (-1.8,-3) circle (.2);

\closedinterval{.1}{-2.5}{2.35};
\openinterval{-2.5}{-2.5}{2.4};

\closedinterval{1.7}{-1.83}{1.34};
\openinterval{-3}{-1.83}{1.34};

\closedinterval{2.2}{-1.12}{1.12};
\openinterval{-3.31}{-1.12}{1.12};

\closedinterval{2.45}{-.38}{1.05};
\openinterval{-3.5}{-.38}{1.05};
}

\newcommand{\halfcircleover}[2]{
\begin{scope}
\clip
  (-4,0) rectangle (4.1,-4.1);
\draw[
  white,
  line width=43pt,
]
  (0,0) circle (3);
\end{scope}
  \halfcircle{#1}{#1}
}

\newcommand{\halfcircleoverAdjusted}{
\begin{scope}
\clip
  (-4,0) rectangle (4.1,-4.1);
\draw[
  white,
  line width=43pt,
]
  (0,0) circle (3);
\end{scope}
  \halfcircleAdjusted
}

\newcommand{\weakHomotopyEquivalence}{\sim}

\newcounter{FixedFigure}
\newcommand{\figurenumber}{%
    \refstepcounter{FixedFigure}%
    \theFixedFigure%
}

\begin{document}

\setlength{\abovedisplayskip}{2.2pt}
\setlength{\belowdisplayskip}{2.2pt}
\setlength{\abovedisplayshortskip}{-5pt}
\setlength{\belowdisplayshortskip}{2pt}

\title{
  Cohomotopy, 
  Framed Links, and 
  Abelian Anyons
}

\author{
  \def\arraystretch{.7}
  \begin{tabular}{c}
  Hisham Sati\rlap{${}^{
    \hyperlink{DoS}{\ast} 
    \hyperlink{Courant}{\dagger}
  }$}
  \\
  \footnotesize%
  \tt%
  hsati@nyu.edu
  \end{tabular}
  \;\;
  \mbox{and}
  \;\;
  \def\arraystretch{.7}
  \begin{tabular}{c}
  Urs Schreiber\rlap{${}^{\hyperlink{DoS}{\ast}}$}
  \\
  \footnotesize%
  \tt%
  us13@nyu.edu
  \end{tabular}
}

\maketitle

\begin{abstract}
  We establish a natural identification of cobordism classes of framed links with the fundamental group of the group-completed configuration space of points in the plane,
  by appeal to Okuyama's previously underappreciated interval configuration model for the latter.
  Under Segal's theorem, these classes are integers generated by the Hopf generator in the 2-cohomotopy of the 3-sphere, and we identify these knot-theoretically with the writhe sum of the linking and framing number of links. We observe that this link invariant is the properly regularized ``Wilson line observable'' of abelian Chern-Simons theory, and show, in consequence of the main theorem, that this arises equivalently as the expectation value of pure quantum states on the group algebra, under link sum, of cobordism classes of framed links. Observing that these quantum states regard framed links as worldlines of anyons, in that they assign a fixed complex phase factor to each crossing (braiding) of strands, we close with an outlook on implications for the identification of anyonic solitons in 2D electron gases as exotic flux quantized in Cohomotopy instead of in ordinary cohomology.
\end{abstract}

\vspace{.8cm}

\begin{center}
\begin{minipage}{10cm}
  \tableofcontents
\end{minipage}
\end{center}

\medskip

\vfill

\hrule
\vspace{5pt}

{
\hypertarget{DoS}{}
\footnotesize
\noindent
\def\arraystretch{1}
\tabcolsep=0pt
\begin{tabular}{ll}
${}^*$\,
&
Mathematics, Division of Science; and
\\
&
Center for Quantum and Topological Systems,
\\
&
NYUAD Research Institute,
\\
&
New York University Abu Dhabi, UAE.  
\end{tabular}
\hfill
\adjustbox{raise=-15pt}{
\href{https://ncatlab.org/nlab/show/Center+for+Quantum+and+Topological+Systems}{\includegraphics[width=3cm]{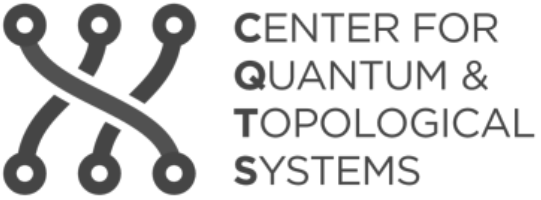}}
}

\vspace{1mm} 
\hypertarget{Courant}{}
\noindent 
${}^\dagger$The Courant Institute for Mathematical Sciences, NYU, New York, USA

\vspace{.2cm}

\noindent
The authors acknowledge the support by {\it Tamkeen UAE} under the 
{\it NYU Abu Dhabi Research Institute grant} {\tt CG008}.
}

\newpage

\section{Introduction and Overview}
\label{Introduction}

\noindent
{\bf Motivation and Background.}
In differential topology (cf. \cite{Milnor97}\cite{Benedetti21}), there are famous and profound relations between the submanifolds of a smooth manifold $X$ and continuous maps from $X$ to higher dimensional spheres $S^d$.
Historically, the first of these was {\it Pontrjagin's theorem} in its original unstable and framed form (\cite{Pontrjagin38}\cite{Pontrjagin55}, cf. \cite[\S 7]{Milnor97}\cite[\S II.16]{Bredon93}\cite[\S IX]{Kosinski93}\cite[\S 17]{Benedetti21},
before it became more widely known in the stable and structured guise of the {\it Pontrjagin-Thom theorem}, cf. \cite[\S 1.5]{Kochman96}): This identifies the homotopy classes of maps $X \xrightarrow{\;}  S^d$, hence the {\it $n$-Cohomotopy} \cite{Spanier49}\cite[\S VII]{STHu59}\cite[Ex. 2.7]{FSS23Char} of $X$, denoted
\begin{equation}
  \label{CohomotopyInIntroduction}
  \pi^d(X)
  \;:=\;
  \pi_0\, 
  \mathrm{Map}\big(
    X ,\, 
    S^d
  \big)
  \,,
\end{equation}
with cobordism classes $\mathrm{Cob}_{\mathrm{fr}}^d(X)$ 
(cf. \cite[\S 1.5]{Kochman96})
of normally framed closed submanifolds of co-dimension $d$ 
inside $X$:
\begin{equation}
  \label{PontrjaginTheoremInIntroduction}
  \pi^d(X)
  \underset{
    \mathclap{
      \adjustbox{
        scale=.6,
        raise=-2pt
      }{
        \color{gray}
        Pontrjagin
      }
    }
  }{
    \;\;\;
    \simeq
    \;\;\;
  }
  \mathrm{Cob}^d_{\mathrm{fr}}(X)
  \,.
\end{equation}
In the special case where $d = \mathrm{dim}(X)$, this relation is easy but instructive to describe: The normally framed closed submanifolds are {\it configurations of points} in $X$ and their normal framing is equivalently a choice of sign carried by each point in the configuration. 
In this simple situation, normally framed cobordism is readily seen to identify pairs of oppositely signed points with the empty subset of points, so that cobordism classes are found to be in bijection with the integer {\it net number} of signs in a configuration:
\begin{equation}
  \label{PontrjaginTheoremAndHopfDegree}
  \begin{tikzcd}[sep=small]
    \pi^{\mathrm{dim}(X)}(X)
    \ar[r, phantom, "{ \simeq }"]
    &
    \mathrm{Cob}^{\mathrm{dim}(X)}_{\mathrm{fr}}(X)
    \ar[r, phantom, "{ \simeq }"]
    &
    \mathbb{Z}
    \mathrlap{\,.}
  \end{tikzcd}
\end{equation}
The resulting identification of the $\mathrm{dim}(X)$-Cohomotopy of $X$ with the integers recovers the {\it Hopf degree theorem} \cite[Cor. 5.8]{Kosinski93}\cite[p. 62]{Milnor97}, which labels homotopy classes of maps $X \xrightarrow{\;} S^{\mathrm{dim}(X)}$ by their {\it winding number}.

\medskip

\noindent
{\bf Outline and Results.}
In \S\ref{ConfigurationSpace} we give an analogous (but much richer) formulation of the Pontrjagin theorem in terms of {\it geometric configurations}, now in the special case $X = S^3$ and $d = 2$, where the normally framed closed submanifolds are {\it framed links} (cf. \cite[pp. 15]{Ohtsuki01} and Def. \ref{FramedOrientedLinks} below). It does not previously seem to have found due attention that in this case  Pontrjagin's theorem on cobordism relates to {\it knot theory} (cf. \cite{CrowellFox63}\cite{Sossinsky23}). We identify the knot-theoretic mechanism which corresponds to the Pontrjagin theorem, in this case. In \S\ref{QuantumObservables} we explain how this result may be understood as a previously missing geometric justification for the traditional but {\it ad hoc} renormalization of {\it Wilson loop observables} in Chern-Simons theory, and in \S\ref{Conclusion} we indicate how it has implications in application to the identification of {\it anyonic solitons} in quantum materials \cite{SS25-FQHViaAlgTop}.

Concretely, since the {\it Hopf fibration} $S^3 \xrightarrow{h} S^2$ 
(cf. \cite{Lyons03})
freely generates the abelian group
\begin{equation}
 \label{HopfFibration}
 \pi^2(S^3)
 \;\simeq\;
 \pi_3(S^2)
 \;\simeq\;
 \mathbb{Z}
 \,,
\end{equation}
Pontrjagin's theorem \eqref{PontrjaginTheoremInIntroduction} implies that the cobordism classes of framed links must also be labeled by integers. Our main Thm. \ref{ChargedOpenStringLoopsClassifiedByCrossingNumber} shows that this characteristic integer is equivalently the {\it total crossing number} or {\it writhe} $\# L$ (Def. \ref{LinkingNumber}) of a give link diagram $L$, namely the sum of the {\it linking number} and the {\it framing number} of a framed link. This combined quantity appears, as we observe in \S\ref{QuantumObservables}, in the exponential for Wilson loop observables in Chern-Simons theory.

Equivalently, this means, as we observe in \S\ref{QuantumObservables},  that any {\it pure quantum state} $\vert \zeta \rangle$ (in the sense of quantum probability theory) on the group algebra of 
\begin{equation}
  \mathrm{Cob}_{\mathrm{fr}}^2\big(S^3\big)
  \;\simeq\;
  \mathbb{Z}
\end{equation}
is determined by a unimodular phase $\zeta \in \mathrm{U}(1) \subset \mathbb{C}$ and as such given by assigning to a framed link $L$ --- regarded as presenting a homogeneous element $[L]$ in the group algebra $\mathbb{C}\big[\mathrm{Cob}_{\mathrm{fr}}^2(S^3)\big]$ --- the value
\begin{equation}
  \langle \zeta \vert L \vert \zeta\rangle
  \;=\;
  \zeta^{\# L}
  \;\in\;
  \mathrm{U}(1)
  \,\subset\,
  \mathbb{C}
  \,,
\end{equation}
hence one factor of the phase $\zeta$ for each  {\it braiding} of one strand in the link diagram over another (with the inverse phase for braidings of the opposite orientation). Exactly such assignment of {\it braiding phases} to their {\it worldlines} is the characteristic property of physical particles ({\it solitons}, really) known as (abelian) {\it anyons}. We close by explaining how this is not a coincidence, but reflective of a more general result \cite{SS25-FQHViaAlgTop}
by which (abelian) anyons such as observed in ``fractional quantum Hall systems'' may be understood as solitons of exotic flux quantized in 2-Cohomotopy instead of in ordinary cohomology.

\medskip

\noindent
{\bf Approach and Methods.}
The method by which we prove these results is via refinement of Pontrjagin's theorem by {\it Segal's theorem} (\cite[Thm. 1]{Segal73}, recalled as Prop. \ref{SegalTheorem}, which is well known to experts) and then by invoking a more recent result by Okuyama (\cite{Okuyama05}, recalled as Prop. \ref{OkuyamaTheorem}, which may not have found due attention before) to obtain a more concrete handle on the situation:

Segal's theorem gives the full homotopy type of the mapping space \eqref{CohomotopyInIntroduction} -- instead of just its set of connected components -- in the special case that $X = S^d$ itself is a sphere of the same dimension as the coefficient sphere. Concretely, by observing the isomorphisms (cf. Rem. \ref{SingedIntervalsReflectingCohomotopy} below)
\begin{equation}
  \label{MakingMapS2S2AppearInIntroduction}
  \pi^3(S^2)
  \,\defneq\,
  \pi_0 
  \, 
  \mathrm{Map}\big(
    S^3
    ,\,
    S^2
  \big)
  \;\simeq\;
  \pi_0 
  \, 
  \mathrm{Map}^\ast\big(
    S^3
    ,\,
    S^2
  \big)
  \;\simeq\;
  \pi_1
  \,
  \mathrm{Map}^\ast\big(
    S^2
    ,\,
    S^2
  \big)  
  \mathrlap{\,,}
\end{equation}
we find that framed links may equivalently be understood as loops is the pointed mapping space from the 2-sphere to itself --- and Segal's theorem describes this latter mapping space as the {\it group completion} $\mathbb{G}\mathrm{Conf}(\mathbb{R}^2)$ (see Def. \ref{PlainConfigurationSpaceOfPoints} below) of the space of (unordered) configurations of points in the plane, known as the {\it configuration space} $\mathrm{Conf}(\mathbb{R}^2)$, hence:
\begin{equation}
  \pi^3(S^2)
  \;\simeq\;
  \pi_1\,
  \mathbb{G}\mathrm{Conf}(\mathbb{R}^2)
  \mathrlap{\,.}
\end{equation}

Here ``group completion'' $\mathbb{G}(-)$ refers to adjoining formal inverses to the topological partial {\it monoid} ({\it semi-group}) of configurations of points under disjoint union. Naively this may appear to be exhibited just by making the points in the configurations carry signs, with oppositely signed pairs of points allowed to coincide and annihilate --- and indeed we saw that this gives the correct description of the connected components, via Pontrjagin's theorem \eqref{PontrjaginTheoremAndHopfDegree}.

However, the effect of group completion on the full homotopy type of the mapping space is well-known to be more subtle \cite[p. 96]{McDuff75}\cite{CarusoWaner81} and a tractable configuration model was long missing.
This gap was filled by Okuyama's observation
\cite{Okuyama05} that the group-completed configuration spaces may be modeled by configurations not of signed points but of {\it intervals} with signed endpoints, see Figure \ref{RelationsOfSignedConfigurations} and Ex. \ref{RelationsBetweenChargedOpenStringWorldsheets} below.

With \eqref{MakingMapS2S2AppearInIntroduction} we are thus reduced to understanding the loops of such intervals, and the main steps in our proof in \S\ref{ConfigurationSpace} consist of identifying these with framed links (Prop. \ref{ChargedStringLoopsAsFramedLinksOnEquivalenceClasses}, using facts from functorial knot theory, Fig. \ref{FigureRibbonCategory}) and their continuous deformation with link cobordism. From this analysis the claims about quantum observables follow in \S\ref{QuantumObservables}.

\bigskip

\noindent
{\bf Acknowledgements.} 
We thank Sadok Kallel and Shingo Okuyama for comments on the material in \S\ref{ConfigurationSpace}, following the first preprint version of this article.
In particular, Shingo Okuyama kindly informed us of a talk he gave \cite{Okuyama18} where graphical representations of his interval relations, as in our Fig. \ref{RelationsOfSignedConfigurations}, were already made use of.
Last but not least, we fondly remember inspiring discussion with the late Jack Morava, after a first version of this manuscript was made available, and we are grateful to him for pointing out his work \cite{MoravaRolfsen23}, which has some tantalizing points of contact with our discussion, such as one highlighted in Rem. \ref{RelationToDiscriminant} below.

\medskip

\section{Configuration Loops and Links}
\label{ConfigurationSpace}

\noindent
{\bf Group-completed configuration space of points.}
We just need some minimum of notions concerning configuration spaces of points, for further background we refer to \cite{Williams20}\cite{Kallel24}\cite{FadellHusseini01}. 

\begin{definition}[{\bf Plain configuration space of points} {\cite[p. 215]{Segal73}}]
\label{PlainConfigurationSpaceOfPoints}
For $d \in \mathbb{N}$, we write $\mathrm{Conf}(\mathbb{R}^d)$ for the topological space of finite subsets of (i.e. configurations of plain points in) $\mathbb{R}^d$. This is a partial topological monoid under the partial operation
\begin{equation}
  \label{PartialUnionOfCOnfigurations}
  \begin{tikzcd}
  \mathrm{Conf}(\mathbb{R}^d)
  \times
  \mathrm{Conf}(\mathbb{R}^d)
  \ar[
    r, 
    harpoon,
    "{ \sqcup }"
  ]
  &
  \mathrm{Conf}(\mathbb{R}^d)
  \mathrlap{\,,}
  \end{tikzcd}
\end{equation}
which is defined when the pair of configurations is disjoint, in which case it is given by their union. We write
\begin{equation}
  \label{TheGroupCompletedConfigurationSpace}
  \mathbb{G}
  \mathrm{Conf}(\mathbb{R}^d)
  \;:=\;
  \Omega
  \big(
    B_{\sqcup}
    \mathrm{Conf}(\mathbb{R}^d)
  \big)
\end{equation}
for the topological group completion of this partial monoid, namely the based loop space of the topological realization of its simplicial nerve. 
\end{definition}
\begin{proposition}[{\bf Segal's theorem: Group-completed configurations as iterated loops} {\cite[Thm. 1]{Segal73}}]
\label{SegalTheorem}
The cohomotopy charge map {\rm (cf. \cite[Fig. II.13]{Bredon93}\cite[Fig. D]{SS23-Mf}, also: ``scanning map'')} constitutes a weak homotopy equivalence between the group completion of the configuration space of plain points in $\mathbb{R}^d$ {\rm (Def. \ref{PlainConfigurationSpaceOfPoints})} and the $d$-fold based loop space of the $d$-sphere:
\begin{equation}
  \label{GCOnfEquivalentToOmeganSn}
  \mathbb{G}\mathrm{Conf}(\mathbb{R}^d)
  \;\;
    \weakHomotopyEquivalence
  \;\;
  \Omega^d S^d
  .
\end{equation}
\end{proposition}

Not to overburden the discussion with minutiae, we state slightly informally the following Def. \ref{ConfigurationsOfChargedStrings} of the configuration space of signed intervals, relying on graphical appeal to Figure \ref{RelationsOfSignedConfigurations}, since it is an easy exercise to give a fully formal description of the situation indicated on the right of the figure, or else to look it up in \cite[Def. 3.1-2]{Okuyama05}. Conversely, our graphical rendering in Fig. \ref{RelationsOfSignedConfigurations} of the definition offered in \cite[Def. 3.1-2]{Okuyama05} shows what is ``really going on'' there \footnote{
  After the first version of this manuscript was made available, we were kindly informed that the same graphical description was once presented in a talk by Okuyama \cite{Okuyama18}. 
} 
and opens the door to the following analysis.

\begin{definition}[{\bf Configuration space of signed intervals} {\cite[Def. 3.1-2]{Okuyama05}}]
\label{ConfigurationsOfChargedStrings}
For $d \in \mathbb{N}_{\geq 1}$,
we write \footnote{
  The space that we denote $\mathrm{Conf}^I(\mathbb{R}^d)$ in Def. \ref{ConfigurationsOfChargedStrings} would be denoted  ``$I_{d-1}(S^0)_{\mathbb{R}}$'' in the notation of  \cite[Def. 3.3]{Okuyama05}.
} $\mathrm{Conf}^{I}(\mathbb{R}^d)$ for the quotient space by the equivalence relations indicated on the right of Figure \ref{RelationsOfSignedConfigurations} of the topological space of disjoint unions of bounded (half-)open/closed line segments in $\mathbb{R}^d$ {\it parallel to the first coordinate axis}, where (in Fig. \ref{RelationsOfSignedConfigurations}) a filled (black) circle indicates that the corresponding point is included in the interval, while an empty (white) circle indicates that it is not. 
\end{definition}

\begin{tabular}{p{7.9cm}}
\hypertarget{FigureConf}{}
\footnotesize
  {\bf Figure \figurenumber:
  \label{RelationsOfSignedConfigurations}
  Configurations of signed points and intervals.}
  Indicated in the left column is the equivalence relation (\cite[p. 94]{McDuff75}) controlling the configuration space of signed {\it points} in some $\mathbb{R}^d$, where configurations involving a positively and a negatively signed point are connected by a continuous path to the corresponding configuration where both of these points are absent (have mutually annihilated). This configuration space is close to but {\it not} (weak homotopy) equivalent (by \cite[p. 6]{McDuff75}) to the group-completed configuration space $\mathbb{G}\mathrm{Conf}(\mathbb{R}^d)$ \eqref{TheGroupCompletedConfigurationSpace}.

  \smallskip 
  Indicated in the right column are the analogous relations (from \cite[Def. 3.1-2]{Okuyama05}) in the configuration space $\mathrm{Conf}^I(\mathbb{R}^d)$ of signed intervals (Def. \ref{ConfigurationsOfChargedStrings}), where signed points are replaced by line segments of finite length whose endpoints are signed, parallel to a fixed coordinate axis. This configuration space {\it is} (weak homotopy) equivalent to the group-completed  configuration space $\mathbb{G}\mathrm{Conf}(\mathbb{R}^d)$ (by \cite[Thm. 1.1]{Okuyama05}, Prop. \ref{OkuyamaTheorem}).

  \smallskip 
  In both cases, the curvy lines indicate continuous paths in these configuration spaces, here realizing the pair-annihilation processes. Running along these paths in the opposite direction reflects the corresponding pair-creation processes.

  \smallskip 

  The bottom graphics highlights that the processes represented by these continuous paths in both cases are {\it grosso modo} the same --- in that pairs of opposite signs mutually annihilate when coincident---, the difference being that on the right the process is ``smoothed out'' in the same way in which string interactions in string theory smooth out point-interactions of particles in particle physics (cf. \cite[Fig. 2.3-2.4]{Veneziano12}): Where the latter trace out ``worldlines'' with point interactions ({\it Feynman diagrams}) the former trace out ``worldsheet'' surfaces. 

  \smallskip

  On the other hand, in contrast to usual strings of string theory, the intervals here are constrained to be straight and parallel to a fixed coordinate axis in the plane. Our key observation in the following is that, for $d = 2$, this makes loops of such intervals be equivalent to (world)lines equipped with a (normal) {\it framing}, hence to {\it framed links}. This is as expected from Pontrjagin's theorem \eqref{PontrjaginTheoremInIntroduction}, when with Segal's theorem (Prop. \ref{SegalTheorem}) we understand homotopy classes of such loops as forming the 2-cohomotopy of $S^3$ (Rem. \ref{SingedIntervalsReflectingCohomotopy}).

\end{tabular}
\hspace{5mm} 
{\footnotesize
\def\arraystretch{1.3}
\begin{tabular}{|c|c|}
\hline
\multicolumn{2}{|c|}{\bf Configurations of signed}
\\
\bf points & \bf intervals
\\
\hline
\hline
&
\\[-10pt]
\adjustbox{raise=2cm}
{
\begin{tikzpicture}[decoration=snake]

\draw[line width=1pt,draw=gray,fill=black]
  (+.7,0) circle (.2);
\draw[
  line width=1pt,draw=gray,fill=white,
  fill opacity=.5
]
  (-0.7,0) circle (.2);

\draw[
  decorate,
  ->
] (0,-.3) -- (0,-1);

\draw[line width=1pt,draw=gray,fill=black]
  (.2,-1.4) circle (.2);
\draw[
  line width=1pt,draw=gray,fill=white,
  fill opacity=.5
]
  (0,-1.4) circle (.2);

\begin{scope}[shift={(0,-1.5)}]
\draw[
  decorate,
  ->
] (0,-.3) -- (0,-1);
\end{scope}

\node at (0,-2.8)
  {$
    \varnothing
  $};

\end{tikzpicture}
}

&

\begin{tikzpicture}[decoration=snake]

\begin{scope}[shift={(2,-2)}]

\begin{scope}[shift={(0,0)}]
\draw[line width=2pt, gray]
  (0,0) -- (1,0);
\draw[line width=1pt,draw=gray,fill=white]
  (0,0) circle (.2);
\draw[line width=1pt,draw=gray,fill=white]
  (1,0) circle (.2);
\end{scope}

\begin{scope}[shift={(1.6,0)}]
\draw[line width=2pt, gray]
  (0,0) -- (1,0);
\draw[line width=1pt,draw=gray,fill=black]
  (0,0) circle (.2);
\draw[line width=1pt,draw=gray,fill=black]
  (1,0) circle (.2);
\end{scope}

\begin{scope}[shift={(1.6,-1.3)}]
\draw[line width=2pt, gray]
  (-.05,0) -- (1,0);
\draw[line width=1pt,draw=gray,fill=black]
  (-.2,0) circle (.2);
\draw[line width=1pt,draw=gray,fill=black]
  (1,0) circle (.2);
\end{scope}

\begin{scope}[shift={(0,-1.3)}]
\draw[line width=2pt, gray]
  (0,0) -- (1,0);
\draw[line width=1pt,draw=gray,fill=white]
  (0,0) circle (.2);
\draw[line width=1pt,draw=gray,fill=white, fill opacity=.5]
  (1.2,0) circle (.2);
\end{scope}

\draw[
  decorate,
  ->
] (1.3,-.3) -- (1.3,-1);

\draw[
  decorate,
  ->
] (1.3,-1.7) -- (1.3,-2.4);

\begin{scope}[shift={(0,-1.2)}]
\draw[
  decorate,
  ->
] (1.3,-1.7) -- (1.3,-2.4);
\end{scope}

\begin{scope}[shift={(0,-2.4)}]
\draw[
  decorate,
  ->
] (1.3,-1.7) -- (1.3,-2.4);
\end{scope}

\begin{scope}[shift={(0,-3.9)}]
\draw[
  decorate,
  ->
] (1.3,-1.7) -- (1.3,-2.4);
\end{scope}

\begin{scope}[shift={(0,-2.6)}]
\draw[line width=2pt, gray]
  (0,0) -- (2.6,0);
\draw[line width=1pt,draw=gray,fill=white]
  (0,0) circle (.2);
\draw[line width=1pt,draw=gray,fill=black]
  (2.6,0) circle (.2);
\end{scope}

\begin{scope}[shift={(0,-3.9)}]
\draw[line width=2pt, gray]
  (.6,0) -- (2,0);
\draw[line width=1pt,draw=gray,fill=white]
  (.6,0) circle (.2);
\draw[line width=1pt,draw=gray,fill=black]
  (2,0) circle (.2);
\end{scope}

\begin{scope}[shift={(.3,-5.2)}]
\draw[line width=1pt,draw=gray,fill=black]
  (1.07,0) circle (.2);
\draw[line width=1pt,draw=gray,fill=white, fill opacity=.5]
  (.93,0) circle (.2);
\end{scope}

\draw (1.3, -6.6) node {$\varnothing$};

\end{scope}
  
\end{tikzpicture}
\\
\hline
\multicolumn{2}{|c|}{
  \bf
  \hspace{-1.4cm}
  tracing out
}
\\
\bf
``worldlines'' & \bf ``worldsheets''
\\
\hline
&
\\[-8pt]
\adjustbox{
  raise=1.5cm
}{
\begin{tikzpicture}[decoration=snake]

\draw[line width=1pt,draw=gray,fill=black]
  (+.7,0) circle (.2);
\draw[
  line width=1pt,draw=gray,fill=white,
  fill opacity=.5
]
  (-0.7,0) circle (.2);

\begin{scope}[shift={(0,.5)}]
\draw[line width=1pt,draw=gray,fill=black]
  (.2,-1.4) circle (.2);
\draw[
  line width=1pt,draw=gray,fill=white,
  fill opacity=.5
]
  (0,-1.4) circle (.2);
\end{scope}

\node at (0.1,-2)
  {$
    \varnothing
  $};

\draw[gray]
  (-.7,0) -- (0.1,-1.2); 
\draw[gray]
  (+.71,0) -- (0.1,-1.2); 
\draw[
  gray,
  dashed
]
  (0.1,-1.2) -- (0.1,-1.8);

\end{tikzpicture}
}
&
\begin{tikzpicture}

\begin{scope}[
  shift={(2,-2)}
]

\begin{scope}[shift={(0,0)}]
\draw[line width=2pt, gray]
  (0,0) -- (1,0);
\draw[line width=1pt,draw=gray,fill=white]
  (0,0) circle (.2);
\draw[line width=1pt,draw=gray,fill=white]
  (1,0) circle (.2);
\end{scope}

\begin{scope}[shift={(2,0)}]
\draw[line width=2pt, gray]
  (0,0) -- (1,0);
\draw[line width=1pt,draw=gray,fill=black]
  (0,0) circle (.2);
\draw[line width=1pt,draw=gray,fill=black]
  (1,0) circle (.2);
\end{scope}

\begin{scope}[shift={(2,-.8)}]
\draw[line width=2pt, gray]
  (-.2,0) -- (1,0);
\draw[line width=1pt,draw=gray,fill=black]
  (-.4,0) circle (.2);
\draw[line width=1pt,draw=gray,fill=black]
  (1,0) circle (.2);
\end{scope}

\begin{scope}[shift={(0,-.8)}]
\draw[line width=2pt, gray]
  (0,0) -- (1.2,0);
\draw[line width=1pt,draw=gray,fill=white]
  (0,0) circle (.2);
\draw[line width=1pt,draw=gray,fill=white, fill opacity=.5]
  (1.4,0) circle (.2);
\end{scope}

\begin{scope}[shift={(0,-1.6)}]
\draw[line width=2pt, gray]
  (0,0) -- (2.8,0);
\draw[line width=1pt,draw=gray,fill=white]
  (0.2,0) circle (.2);
\draw[line width=1pt,draw=gray,fill=black]
  (2.8,0) circle (.2);
\end{scope}

\begin{scope}[shift={(.15,-2.4)}]
\draw[line width=2pt, gray]
  (.5,0) -- (2,0);
\draw[line width=1pt,draw=gray,fill=white]
  (.5,0) circle (.2);
\draw[line width=1pt,draw=gray,fill=black]
  (2.2,0) circle (.2);
\end{scope}

\begin{scope}[shift={(.5,-3.2)}]
\draw[line width=1pt,draw=gray,fill=black]
  (1.07,0) circle (.2);
\draw[line width=1pt,draw=gray,fill=white, fill opacity=.5]
  (.93,0) circle (.2);
\end{scope}

\draw (1.5, -4) node {$\varnothing$};

\end{scope}

\draw[
  gray,
  smooth,
  fill=gray,
  fill opacity=.3,
  draw opacity=.3
]
  plot 
  coordinates{
    (2,-1.06) 
    (2,-2.8)
    (2.2,-3.6)
    (2.65,-4.4)
    (3.5,-5.3)
    (4.35,-4.4)
    (4.8,-3.6)
    (5,-2.8)
    (5,-1.06)
  }
  -- (4.05,-1.06)
  plot 
  coordinates {
    (4.05,-1.06)
    (4,-2)
    (3.6, -2.8)
    (3.4, -2.8)
    (3,-2.1)
    (2.95,-1.06)
  }
  --(2,-1.06);

\draw[
  line width=2pt,
  white
]
  (1.9,-1.5) -- (5.1,-1.5);
\draw[
  line width=2pt,
  white
]
  (1.9,-1.34) -- (5.1,-1.34);
\draw[
  line width=2pt,
  white
]
  (1.9,-1.18) -- (5.1,-1.18);
  
\end{tikzpicture}
\\
\hline
\end{tabular}
}

\vspace{.2cm}

\begin{example}[\bf Connected components of configurations of signed intervals]
For $c \in \mathrm{Conf}^I(\mathbb{R}^d)$ a configuration of signed intervals (Def. \ref{ConfigurationsOfChargedStrings})
and with 
$n_{\mathrm{black}/\mathrm{white}}(c) \in \mathbb{N}$ 
denoting its total numbers of black/white endpoints, respectively, it is clear that
\begin{equation}
  n(c) 
    \,:=\, 
  \tfrac{1}{2}\big(
    n_{\mathrm{black}}(c)
    -
    n_{\mathrm{white}}(c)
  \big)
  \;\in\;
  \mathbb{Z}
\end{equation}
is an integer which is constant under the above interaction processes, hence which constitutes a continuous function
\begin{equation}
 \begin{tikzcd}
   n(-) 
     \;:\;
   \mathrm{Conf}^I(\mathbb{R}^d)
   \ar[r]
   &
   \mathbb{Z}
   \,,
 \end{tikzcd}
\end{equation}
and that this establishes an isomorphism on path-connected components (cf. also Rem. \ref{SingedIntervalsReflectingCohomotopy} below):
\begin{equation}
  \begin{tikzcd}
    \pi_0\big(n(-)\big)
    \;:\;
    \pi_0 \, 
    \mathrm{Conf}^I(\mathbb{R}^d)
    \ar[
      r,
      "{ \sim }"
    ]
    &
    \mathbb{Z}
    \mathrlap{\,.}
  \end{tikzcd}
\end{equation}
Hence for $n \in \mathbb{Z}$ we shall denote the corresponding connected component by
\begin{equation}
  \label{ConnectedComponentOfConfigurationSpaceOfSignedIntervals}
  \mathrm{Conf}^I_n(\mathbb{R}^d)
  \;\subset\;
  \mathrm{Conf}^I(\mathbb{R}^d)
  \mathrlap{\,.}
\end{equation}
\end{example}

\begin{proposition}[{\bf Okuyama's theorem: Interval configurations group-complete points} {\cite[Thm. 1]{Okuyama05}}]
\label{OkuyamaTheorem}
  For $n \in \mathbb{N}_{\geq 1}$
  there is a weak homotopy equivalence between the configuration space of signed intervals (Def. \ref{ConfigurationsOfChargedStrings}) and the group completion of the plain configuration space of plain points {\rm (Def. \ref{PlainConfigurationSpaceOfPoints})}:
  \begin{equation}
    \label{ChargedStringSpaceIsGroupCompletedConfigSpace}
    \mathrm{Conf}^I(\mathbb{R}^d)
    \;\weakHomotopyEquivalence\;
    \mathbb{G}\mathrm{Conf}(\mathbb{R}^d)
    \,.
  \end{equation}
\end{proposition}

\begin{remark}[\bf Signed interval configurations reflecting Cohomotopy moduli]
  \label{SingedIntervalsReflectingCohomotopy}
  In summary, this identifies the $d$-Cohomotopy moduli vanishing at infinity on $\mathbb{R}^d$ with the Okuyama configuration space of signed intervals in $\mathbb{R}^d$:
  \begin{equation}
    \label{StringyConfigurationsAsCohomotopy}
    \def\arraystretch{1.5}
    \def\arraycolsep{2pt}
    \begin{array}{ll}
    \scalebox{.7}{
      \color{darkblue}
      \bf
      \def\arraystretch{.9}
      \begin{tabular}{c}
        Configuration space of
        \\
        signed intervals
      \end{tabular}
    }
    \mathrm{Conf}^I(\mathbb{R}^d)
    &
    \underset{
      \scalebox{.7}{
        \eqref{ChargedStringSpaceIsGroupCompletedConfigSpace}
      }
    }{
      \weakHomotopyEquivalence
    }
    \;
    \mathbb{G}\mathrm{Conf}(\mathbb{R}^d)
    \\
    &
    \underset{
      \scalebox{.7}{
        \eqref{GCOnfEquivalentToOmeganSn}
      }
    }{
      \weakHomotopyEquivalence
    }
    \;    
    \Omega^d S^d
    \;\,\defneq\;
    \mathrm{Map}^{\ast}\big(
      S^d
      ,\,
      S^d
    \big)
    \;\,\simeq\;
    \mathrm{Map}^{\ast}\big(
      \mathbb{R}^d_{\cup \{\infty\}}
      ,\,
      S^d
    \big)
    \!\!\!\scalebox{.7}{
      \color{darkblue}
      \bf
      \def\arraystretch{.9}
      \begin{tabular}{c}
        Cohomotopy moduli 
        \\
        vanishing at infinity
        \color{black},
      \end{tabular}
    }
    \end{array}
  \end{equation}
  where 
  $(-)_{\cpt}$ denotes the {\it one-point compactification} obtained by adjoining the {\it point at infinity},  and $\mathrm{Map}^\ast(-,-) \,\subset\, \mathrm{Map}(-,-)$ denotes the pointed mapping space (cf. \cite[\S 3]{James84}). This perspective on $\mathrm{Conf}^I(\mathbb{R}^2)$ shows that it also witnesses configurations flux quantized in 2-Cohomotopy, explained in \S\ref{Conclusion}. 
\end{remark}
This implies:
\begin{proposition}[\bf Fundamental group of signed interval configurations in the plane]
\label{TheFundamenralGroupOfStringConfigurationSpace}
The fundamental group of the configuration space of signed intervals in the plane ($d=2$, Def. \ref{ConfigurationsOfChargedStrings}) is the group of integers:
\begin{equation}
  \label{FundamentalGroupOfStringConfigurationSpace}
  \begin{array}{ll}
    \pi_1
    \big(
      \mathrm{Conf}^I_0(\mathbb{R}^2)
    \big)
    &
    \defneq
    \;\;
    \pi_0
    \big(
      \Omega \, 
      \mathrm{Conf}^I_0(\mathbb{R}^2)
    \big)    
    \;\;
    \underset{
      \mathclap{
        \adjustbox{
          scale=.7,
          raise=-3pt
        }{          \eqref{StringyConfigurationsAsCohomotopy}
        }
      }
    }{
      \simeq
    }
    \;\;
    \pi_0
    \big(
      \Omega^3 S^2
    \big)    
    \;\;
    \defneq
    \;\;
    \pi_3(S^2)
    \underset{
      \mathclap{
        \adjustbox{
          scale=.7,
          raise=-3pt
        }{
          \eqref{HopfFibration}
        }
      }
    }{
    \;\;
    \simeq
    \;\;
    }
    \mathbb{Z}
    \,.
  \end{array}
\end{equation}
\end{proposition}
The generator on the right of \eqref{FundamentalGroupOfStringConfigurationSpace} is well-known to be represented by the complex Hopf fibration \eqref{HopfFibration}. Our goal is to understand the corresponding generator on the left, i.e., the signed interval loop whose composites and their reverses are deformation-equivalent to general loops of signed intervals. 
A key observation for this identification is the following:

\begin{example}[\bf Continuous deformations of paths of signed interval configurations]
\label{RelationsBetweenChargedOpenStringWorldsheets}
  Continuous deformations of paths of signed intervals, i.e., continuous maps of the form
  $
      {[0,1]}^2 
      \longrightarrow 
      \mathrm{Conf}^I(\mathbb{R}^2)
  $,
  evidently subsume the following ``moves''  (and their images under the exchange of positive with negative endpoint signs):
\begin{equation}
\scalebox{0.8}{$
\adjustbox{
  scale=.6,
  raise=-1.4cm
}{
\begin{tikzpicture}
\draw[
  gray,
  line width=30pt,
  draw opacity=.4,
]
 (0,2) 
 --
 (0,-2);

\closedinterval{-.5}{1.6}{1}
\closedinterval{-.5}{.8}{1}
\closedinterval{-.5}{0}{1}
\closedinterval{-.5}{-.8}{1}
\closedinterval{-.5}{-1.6}{1}

\begin{scope}[
  xshift=3.3cm,
  xscale=-1
]
\draw[
  gray,
  line width=30pt,
  draw opacity=.4
]
 (0,2) 
 --
 (0,-2);

\openinterval{-.5}{1.6}{1}
\openinterval{-.5}{.8}{1}
\openinterval{-.5}{0}{1}
\openinterval{-.5}{-.8}{1}
\openinterval{-.5}{-1.6}{1}

\end{scope}
\end{tikzpicture}
}
\;\;\;\;\;
\begin{tikzpicture}[decoration=snake]
\draw[decorate,->, line width=1]
  (0,0) -- (+0.55,0);
\draw[decorate,->, line width=1]
  (0,0) -- (-0.55,0);
\end{tikzpicture}
\;\;\;\;\;
\adjustbox{
  scale=.6,
  raise=-1.4cm
}{
\begin{tikzpicture}

\draw[
  gray,
  line width=30pt,
  draw opacity=.4
]
 (0,2) 
  .. controls (0,1) and (1,0) ..
  (1,0)
  .. controls (1,0) and (0,-1) ..
 (0,-2);

\closedinterval{-.45}{1.6}{1.02};
\closedinterval{-.2}{.8}{1.2};
\closedinterval{.3}{0}{1.3};
\closedinterval{-.2}{-.8}{1.2};
\closedinterval{-.45}{-1.6}{1.02};

\begin{scope}[
  xshift=3.5cm,
  xscale=-1
]
\draw[
  gray,
  line width=30pt,
  draw opacity=.4
]
 (0,2) 
  .. controls (0,1) and (1,0) ..
  (1,0)
  .. controls (1,0) and (0,-1) ..
 (0,-2);

\openinterval{-.45}{1.6}{1.02};
\openinterval{-.2}{.8}{1.2};
\openinterval{.3}{0}{1.3};
\openinterval{-.2}{-.8}{1.2};
\openinterval{-.45}{-1.6}{1.02};
\end{scope}

\end{tikzpicture}
}
\;\;\;\;\;
\begin{tikzpicture}[decoration=snake]
\draw[decorate,->, line width=1]
  (0,0) -- (+0.55,0);
\draw[decorate,->, line width=1]
  (0,0) -- (-0.55,0);
\end{tikzpicture}
\;\;\;\;\;
\adjustbox{
  scale=.6,
  raise=-1.1cm
}{
\begin{tikzpicture}

\begin{scope}
\clip
  (-.6,0) rectangle (5,2);
\draw[
  gray,
  line width=28,
  draw opacity=.4
]
  (1.75,2.6) circle (1.8);
\end{scope}
\closedinterval{-.4}{1.75}{1.2};
\closedinterval{-.1}{1.3}{1.75};
\begin{scope}[
  xshift=3.5cm,
  xscale=-1
]
\openinterval{-.4}{1.75}{1.2};
\openinterval{-.1}{1.3}{1.75};
\end{scope}
\oppositeinterval{.32}{.8}{2.85};
\oppositeinterval{1.65}{.33}{.2};

\begin{scope}[yscale=-1]
\begin{scope}
\clip
  (-.6,0) rectangle (5,2);
\draw[
  gray,
  line width=28,
  draw opacity=.4
]
  (1.75,2.6) circle (1.8);
\end{scope}
\closedinterval{-.4}{1.75}{1.2};
\closedinterval{-.1}{1.3}{1.75};
\begin{scope}[
  xshift=3.5cm,
  xscale=-1
]
\openinterval{-.4}{1.75}{1.2};
\openinterval{-.1}{1.3}{1.75};
\end{scope}
\oppositeinterval{.32}{.8}{2.85};
\oppositeinterval{1.65}{.33}{.2};
\end{scope}
  
\end{tikzpicture}
}
$}
\end{equation}
\vspace{-.5cm}
\begin{equation}
\label{StringyZigZagMove}
\adjustbox{scale=.9}{
\hspace{-.8cm}
\adjustbox{
  scale=.3,
  raise=-1.5cm
}{
\begin{tikzpicture}

  \begin{scope}[
    scale=-1
  ]
    \halfcircle;  
  \end{scope}

  \begin{scope}[
    shift={(6,0)}
  ]
    \halfcircle;  
  \end{scope}

\draw[
  gray,
  draw opacity=.4,
  line width=30,
]
  (-3,0) -- 
  (-3,-4.8);

\foreach \k in {1,...,6} {
  \closedinterval{-3.5}{+.4-.8*\k}{1};
};

\begin{scope}[
  shift={(5.9,0)},
  scale=-1
]
\foreach \k in {1,...,6} {
  \closedinterval{-3.5}{+.4-.8*\k}{1};
};

\draw[
  gray,
  draw opacity=.4,
  line width=30,
]
  (-3,0) -- 
  (-3,-4.8);

\end{scope}
  
\end{tikzpicture}
}
\adjustbox{
  scale=.5
}{
\begin{tikzpicture}[decoration=snake]
\draw[decorate,->, line width=1]
  (0,0) -- (+0.55,0);
\draw[decorate,->, line width=1]
  (0,0) -- (-0.55,0);
\end{tikzpicture}
}
\adjustbox{
  scale=.3,
  raise=-1.8cm
}{
\begin{tikzpicture}

\closedinterval{.15}{-4.8}{1};
\closedinterval{.15}{-4}{1};
\closedinterval{.15}{-3.2}{1.25};
\closedinterval{.15}{-2.4}{1.5};
\closedinterval{.15}{-1.6}{1.75};
\closedinterval{.15}{-.8}{2};
\closedinterval{.15}{0}{2.4};
\closedinterval{.3}{.8}{2.6};
\closedinterval{.6}{1.6}{2.4};

\closedinterval{9-1.3}{.8}{1.4+1.3};

\openinterval{3.2}{1.6}{3.1};

\begin{scope}[
  shift={(3,5.5)},
  xscale=-1
]
\draw[gray, line width=3.3]
  (-1.7+.2,-3.1) -- (1.7-.2,-3.1);
\draw[draw=gray,line width=1.3,fill=black]
  (1.7,-3.1) circle (.2);
\draw[draw=gray,line width=1.3,fill opacity=.6,fill=white]
  (-1.7,-3.1) circle (.2);
\end{scope}

\begin{scope}[
  shift={(3,6.3)},
  xscale=-1
]
\draw[draw=gray,line width=1.3,fill=black]
  (.1,-3.1) circle (.2);
\draw[draw=gray,line width=1.3,fill opacity=.6,fill=white]
  (-.1,-3.1) circle (.2);
\end{scope}

\openinterval{3.15}{.8}{4.3};

\closedinterval{7.7}{0}{2.4};
\openinterval{5.1-.4}{0}{2.4+.4};

\closedinterval{9.5-1.2}{1.6}{1+1.2};
\closedinterval{9.5-.75}{2.4}{1+.75};
\closedinterval{9.5-.5}{3.2}{1+.5};
\closedinterval{9.5-.25}{4}{1+.25};
\closedinterval{9.5}{4.8}{1};
\closedinterval{9.5}{5.6}{1};

\begin{scope}[
  shift={(10.5,1.5)},
  scale=-1
]
\begin{scope}[
  shift={(3,5.5)},
  xscale=-1
]
\draw[gray, line width=3.3]
  (-1.7+.2,-3.1) -- (1.7-.2,-3.1);
\draw[draw=gray,line width=1.3,fill=black]
  (1.7,-3.1) circle (.2);
\draw[draw=gray,line width=1.3,fill opacity=.6,fill=white]
  (-1.7,-3.1) circle (.2);
\end{scope}

\begin{scope}[
  shift={(3,6.3)},
  xscale=-1
]
\draw[draw=gray,line width=1.3,fill=black]
  (.1,-3.1) circle (.2);
\draw[draw=gray,line width=1.3,fill opacity=.6,fill=white]
  (-.1,-3.1) circle (.2);
\end{scope}
\end{scope}

\draw[
  draw=gray,
  draw opacity=.4,
  fill=gray,
  fill opacity=.4,
]
  plot 
  [smooth cycle]
  coordinates {
    (.1,-5.5)
    (1.1,-5.5)
    (1.2,-4)
    (1.4,-3.2)
    (1.65,-2.4)
    (1.85,-1.8)
    (2.1,-1)
    (2.6,-0)
    (3,.8)
    (4.7,0)
    (5.8,-.9)
    (7.4,-1.7)
    (9.3,-.9)
    (10.1,-0)
    (10.4,+.8)
    (10.5,+1.6)
    (10.5,+6)
    (9.6,+6)
    (9.5,4.9)
    (9.3,4.2)
    (9.05,3.2)
    (8.75,2.4)
    (8.3,1.6)
    (7.6,.8)
    (6.3,1.6)
    (4.7,2.4)
    (3.0,3.2)
    (1.3,2.4)
    (.6,1.6)
    (.3,.8)
    (.2,0)
    (.2,-5)
  };

\draw[
  white,
  line width=30
]
  (-.1,-5.65) -- (1.3,-5.65);

\draw[
  white,
  line width=30
]
  (9,6.5) -- (10.8,6.5);
  
\end{tikzpicture}
}
\adjustbox{
  scale=.5
}{
\begin{tikzpicture}[decoration=snake]
\draw[decorate,->, line width=1]
  (0,0) -- (+0.55,0);
\draw[decorate,->, line width=1]
  (0,0) -- (-0.55,0);
\end{tikzpicture}
}
\adjustbox{
  scale=.3,
  raise=-2.2cm
}{
\begin{tikzpicture}

\draw[
  draw=gray,
  draw opacity=.4,
  fill=gray,
  fill opacity=.4
]
  plot 
  [smooth cycle]
  coordinates {
    (.1,-5.9)
    (1.2,-5.9)
    (1.2,-4.0)
    (2.6,-3.2)
    (4.5,-2.4)
    (7.4,-1.6)
    (9.3,-.8)
    (10.2,0)
    (10.45,.8)
    (10.55,1.6)
    (10.9,2.4)
    (10.9,9)
    (9.8,9)
    (9.8,7.2)
    (8.5,6.4)
    (6.5,5.6)
    (3.7,4.8)
    (1.8,4)
    (.9,3.2)
    (.6,2.4)
    (.5,1.6)
    (.1,.8)
    (.1,-5)
  };

\draw[
  white,
  line width=30
]
  (-.2,-5.7) -- 
  (1.5,-5.7);

\draw[
  white,
  line width=33
]
  (9.5,8.9) -- 
  (11.1,8.9);

\closedinterval{.1}{-4.8}{1.1};
\closedinterval{.1}{-4}{1.1};
\closedinterval{.1}{-4+.8}{2.4};
\closedinterval{.1}{-4+1.6}{4.4};
\closedinterval{.1}{-4+2.4}{7.2};
\closedinterval{.1}{-4+3.2}{9.1};
\closedinterval{.1}{-4+4}{10};
\closedinterval{.1}{-4+4.8}{10.3};
\closedinterval{.5}{-4+5.6}{10};

\begin{scope}[
  shift={(11,3.2)},
  scale=-1
]

\closedinterval{.1}{-4.8}{1.1};
\closedinterval{.1}{-4}{1.1};
\closedinterval{.1}{-4+.8}{2.4};
\closedinterval{.1}{-4+1.6}{4.4};
\closedinterval{.1}{-4+2.4}{7.2};
\closedinterval{.1}{-4+3.2}{9.1};
\closedinterval{.1}{-4+4}{10};
\closedinterval{.1}{-4+4.8}{10.3};
\closedinterval{.5}{-4+5.6}{10};

\end{scope}

\end{tikzpicture}
}
\adjustbox{
  scale=.5
}{
\begin{tikzpicture}[decoration=snake]
\draw[decorate,->, line width=1]
  (0,0) -- (+0.55,0);
\draw[decorate,->, line width=1]
  (0,0) -- (-0.55,0);
\end{tikzpicture}
}
}
\;\;\;
\adjustbox{
  scale=.28,
  raise=-2cm
}{
\begin{tikzpicture}
\draw[
  gray,
  draw opacity=.4,
  line width=30,
]
  (0,0) -- 
  (0,-13.6);

\foreach \k in {1,...,17} {
  \closedinterval{-.5}{+.4-.8*\k}{1};
};
\end{tikzpicture}
}
\end{equation}
\begin{equation}
\label{VacuumLoopVanishes}
\scalebox{0.8}{$   
\adjustbox{}{
\adjustbox{
  scale=.5,
  raise=-1.8cm
}{
\begin{tikzpicture}
 \begin{scope}
   \halfcircle{0,0};
 \end{scope}
 \begin{scope}[yscale=-1]
   \halfcircle{0,0};
 \end{scope}
\end{tikzpicture}
}
\;\;
\begin{tikzpicture}[decoration=snake]
\draw[decorate,->, line width=1]
  (0,0) -- (+0.55,0);
\draw[decorate,->, line width=1]
  (0,0) -- (-0.55,0);
\end{tikzpicture}
\;\;
\adjustbox{
  scale=.5,
  raise=-1.8cm
}{
\begin{tikzpicture}
\begin{scope}
\begin{scope}
\clip
  (-4,0) rectangle (4.1,-4.1);
\draw[
  gray,
  line width=30pt,
  draw opacity=.4,
]
  (0,0) circle (3);
\draw[
  draw opacity=0,
  fill=gray,
  line width=0pt,
  fill opacity=.4,
]
  (0,0) circle (2.48);
\end{scope}

\draw[draw=gray,line width=1.3,fill=black]
  (0.1,-3.5) circle (.2);
\draw[draw=gray,line width=1.3,fill opacity=.6,fill=white]
  (-0.1,-3.5) circle (.2);

\draw[gray, line width=3.3]
  (-1.7+.2,-3.1) -- (1.7-.2,-3.1);
\draw[draw=gray,line width=1.3,fill=black]
  (1.7,-3.1) circle (.2);
\draw[draw=gray,line width=1.3,fill opacity=.6,fill=white]
  (-1.7,-3.1) circle (.2);

\interval{-2.5}{-2.5}{5};

\interval{-3}{-1.83}{5.98};

\interval{-3.31}{-1.12}{6.6};

\interval{-3.5}{-.38}{6.96};
\end{scope}

\begin{scope}[
  yscale=-1
]
\begin{scope}
\clip
  (-4,0) rectangle (4.1,-4.1);
\draw[
  gray,
  line width=30pt,
  draw opacity=.4,
]
  (0,0) circle (3);
\draw[
  draw opacity=0,
  fill=gray,
  line width=0pt,
  fill opacity=.4,
]
  (0,0) circle (2.48);
\end{scope}

\draw[draw=gray,line width=1.3,fill=black]
  (0.1,-3.5) circle (.2);
\draw[draw=gray,line width=1.3,fill opacity=.6,fill=white]
  (-0.1,-3.5) circle (.2);

\draw[gray, line width=3.3]
  (-1.7+.2,-3.1) -- (1.7-.2,-3.1);
\draw[draw=gray,line width=1.3,fill=black]
  (1.7,-3.1) circle (.2);
\draw[draw=gray,line width=1.3,fill opacity=.6,fill=white]
  (-1.7,-3.1) circle (.2);

\interval{-2.5}{-2.5}{5};

\interval{-3}{-1.83}{5.98};

\interval{-3.31}{-1.12}{6.6};

\interval{-3.5}{-.38}{6.96};
\end{scope}
 
\end{tikzpicture}    
}
\;\;
\adjustbox{
  scale=1
}{
\begin{tikzpicture}[decoration=snake]
\draw[decorate,->, line width=1]
  (0,0) -- (+0.55,0);
\draw[decorate,->, line width=1]
  (0,0) -- (-0.55,0);
\end{tikzpicture}
}
}
$}
\;\;
\scalebox{1.2}{$
  \varnothing$}
  \mathrlap{\,.}
\end{equation}
Here and in the following:
\begin{itemize}[
  leftmargin=.7cm
]
\item[(i)] the plane $\mathbb{R}^2$ in which the intervals are embedded shares its horizontal axis with the page and has its other axis perpendicular to it,

\item[(ii)] the parameter $t \in [0,1]$ of path of interval configurations $[0,1] \to \mathrm{Conf}^I(\mathbb{R}^2)$ runs vertically along the page.
\end{itemize}
In partciular therefore, the third move \eqref{VacuumLoopVanishes} is a path of based loops in $\mathrm{Conf}^I(\mathbb{R}^2)$, which thus witnesses that the class of the annulus 
worldsheet in the fundamental group of the configuration space vanishes:
\begin{equation}
\scalebox{0.75}{$
  \left[
\adjustbox{
  scale=.5,
  raise=-2cm
}{
\begin{tikzpicture}
 \begin{scope}
   \halfcircle{0,0};
 \end{scope}
 \begin{scope}[yscale=-1]
   \halfcircle{0,0};
 \end{scope}
\end{tikzpicture}
}
  \right]
$}
  \;\;
  =
  \;\;
  \mathrm{e}
  \;\;
  \in
  \;\;
  \pi_1\big(
    \mathrm{Conf}^I_0(\mathbb{R}^2)
  \big)
  \,.
\end{equation}
\end{example}

\smallskip
Hence the annulus is not the generator of $\pi_1\big(\mathrm{Conf}^I_0(\mathbb{R}^2)\big)$ that we are after, and we need to look further:

\medskip

\noindent
{\bf Loops of signed intervals as framed oriented links.}
Our first observation now is that based loops in the configuration space of signed intervals (Def. \ref{ConfigurationsOfChargedStrings}) may be identified
(as indicated in Figure \ref{FigureStringLoopsAsLinks})
with {\it framed oriented links}. 
For examples of framed link diagrams see \eqref{HopfLinktoUnknot}, \eqref{TrefoilLinkToUnknot} and \eqref{FigureEightNotToUnknot};
for background on framed links in knot theory and quantum topology compare \cite[p. 15]{Ohtsuki01}\cite{EHI20}. We recall the relevant basics:

\begin{definition}[\bf Framed oriented links]
\label{FramedOrientedLinks}
$\,$

\noindent {\bf (i)} A {\it framed oriented link diagram} is an immersion of $k$ oriented circles $\big(S^1\big)^{\smash{\sqcup^k}}$, for $k \in \mathbb{N}$, into the plane $\mathbb{R}^2$ with isolated crossings at Euclidean distance $> 1$ from each other, at each of which one segment is labeled as crossing {\it over} the other.
Here we demand in addition (and without essential restriction of generality) that no strictly horizontal segments appear, hence that the restriction of a link diagram to any $\mathbb{R}^1 \hookrightarrow \mathbb{R}^2$ parallel to $\mathbb{R} \times \{0\}$ consists of finitely many points -- this is used in \eqref{MappingFramedLinksToStringLoops} below.

\noindent {\bf (ii)}  A pair of oriented link diagrams are regarded as equivalent if they may be transformed into each other by a sequence of isotopies (continuous paths in the space of framed link diagrams) and the three {\it Reidemeister moves} shown in Figure \ref{FigureReidemeister}.

\noindent {\bf (iii)}  The {\it framed oriented links} are the corresponding equivalence classes of framed oriented link diagrams.
\end{definition}

\def\tabcolsep{20pt}
\def\arraystretch{1.25}
\begin{minipage}{5cm}
\footnotesize 
{\bf Figure \figurenumber: 
\label{FigureReidemeister}
Reidemeister moves}  for framed link diagrams (cf. \cite[Thm. 1.8]{Ohtsuki01}).
For oriented framed links there are the evident oriented versions of each of these moves \cite{Polyak10}. To obtain a fully combinatorial description of (framed oriented) link diagram equivalence it is sufficient to include also the zig-zag yanking move, see Figure \ref{FigureRibbonCategory}.
\end{minipage}
\;\;\;\;
\begin{tabular}{|c||c|}
  \hline
  &
  \\[-12pt]
  {\small 1st}
  &
  \adjustbox{
    rotate=-90,
    raise=+2.7cm,
    scale=.3
  }{
\begin{tikzpicture}

\draw[
  line width=3,
]
  (-1.3,-1.4) --
  (0,-1.4)
  .. controls 
     (.5,-1.4) and (1,-.5) ..
  (1,0)
  .. controls 
     (1,.5) and (.5,1) ..
  (0,1)  
  .. controls
    (-.5,1) and (-1,.5)
  .. (-1,0)
  .. controls
    (-1,-.4) and (-.5, -.8) ..
  (0,-.8);
\draw[
  line width=3,
]
 (2.5,-.8)
 .. controls 
   (3,-.8) and (3.5,-.4) ..
 (3.5,0)
 .. controls
   (3.5,.5) and (3,1) ..
 (2.5,1)
  .. controls
    (2,1) and (1.5,.5) ..
  (1.5,0)
  .. controls
    (1.5,-.5) and (2,-1.4) ..
  (2.5,-1.4)
  --
  (3.8,-1.4);
\draw[
  line width=12,
  white
]
 (0,-.8) -- (2.5,-.8);
\draw[
  line width=3,
]
 (0,-.8) -- (2.5,-.8);

\end{tikzpicture}  
  }
 \hspace{.5cm}
 \adjustbox{scale=.4}{
 \begin{tikzpicture}[decoration=snake]
   \draw[decorate, ->]
    (-.01,0) -- (0.54,0);
   \draw[decorate, ->]
    (0.01,0) -- (-0.54,0);
 \end{tikzpicture}
 }
 \hspace{.5cm}
 \adjustbox{
   rotate=-90, 
   raise=2.7cm,
   scale=.3
 }
 {
\begin{tikzpicture}
  \draw[
    line width=3,
  ]
  (-1.9,0) -- (3.2,0);
\end{tikzpicture}
 }
 \hspace{.5cm}
 \adjustbox{
   scale=.4
 }{
 \begin{tikzpicture}[decoration=snake]
   \draw[decorate, ->]
    (-.01,0) -- (0.54,0);
   \draw[decorate, ->]
    (0.01,0) -- (-0.54,0);
 \end{tikzpicture}
 }
 \hspace{.5cm}
 \adjustbox{
   rotate=-90, 
   raise=2.7cm,
   scale=.3
 }
 {
\begin{tikzpicture}

\draw[
  line width=3,
]
 (0,-.8) -- (2.5,-.8);

\draw[
  line width=12,
  white
]
  (-1.3,-1.4) --
  (0,-1.4)
  .. controls 
     (.5,-1.4) and (1,-.5) ..
  (1,0)
  .. controls 
     (1,.5) and (.5,1.4) ..
  (0,1.4)  
  .. controls
    (-.5,1.4) and (-1,.5)
  .. (-1,0)
  .. controls
    (-1,-.4) and (-.5, -.8) ..
  (0,-.8);

\draw[
  line width=3,
]
  (-1.3,-1.4) --
  (0,-1.4)
  .. controls 
     (.5,-1.4) and (1,-.5) ..
  (1,0)
  .. controls 
     (1,.5) and (.5,1) ..
  (0,1)  
  .. controls
    (-.5,1) and (-1,.5)
  .. (-1,0)
  .. controls
    (-1,-.4) and (-.5, -.8) ..
  (0,-.8);

\draw[
  line width=12,
  white
]
 (2.5,-.8)
 .. controls 
   (3,-.8) and (3.5,-.4) ..
 (3.5,0)
 .. controls
   (3.5,.5) and (3,1) ..
 (2.5,1)
  .. controls
    (2,1) and (1.5,.5) ..
  (1.5,0)
  .. controls
    (1.5,-.5) and (2,-1.4) ..
  (2.5,-1.4)
  --
  (3.5,-1.4);
\draw[
  line width=3,
]
 (2.5,-.8)
 .. controls 
   (3,-.8) and (3.5,-.4) ..
 (3.5,0)
 .. controls
   (3.5,.5) and (3,1) ..
 (2.5,1)
  .. controls
    (2,1) and (1.5,.5) ..
  (1.5,0)
  .. controls
    (1.5,-.5) and (2,-1.4) ..
  (2.5,-1.4)
  --
  (3.8,-1.4);

\end{tikzpicture} 
 }
 \\[-12pt]
 &
 \\
 \hline
 &
 \\[-12pt]
{\small  2nd}
 &
 \hspace{-.2cm}
 \adjustbox{
   scale=.3,
   raise=-.6cm
  }{
 \begin{tikzpicture}
   \draw[
     line width=3
   ]
   (-1.5,2)
     .. controls (-1.5,1.3) and (0.2,1.2) ..
     (0.2,0)
     .. controls 
     (0.2,-1.2) and (-1.5,-1.3) ..
   (-1.5,-2);
\begin{scope}[
  shift={(-.5,0)},
  xscale=-1,
]
   \draw[
     line width=12,
     white
   ]
   (-1.5,2)
     .. controls (-1.5,1.3) and (0.2,1.2) ..
     (0.2,0)
     .. controls 
     (0.2,-1.2) and (-1.5,-1.3) ..
   (-1.5,-2);
\end{scope}
\begin{scope}[
  shift={(-.5,0)},
  xscale=-1,
]
   \draw[
     line width=3
   ]
   (-1.5,2)
     .. controls (-1.5,1.3) and (0.2,1.2) ..
     (0.2,0)
     .. controls 
     (0.2,-1.2) and (-1.5,-1.3) ..
   (-1.5,-2);
\end{scope}
 \end{tikzpicture}
 }
 \hspace{+.3cm}
 \adjustbox{scale=.4}{
 \begin{tikzpicture}[decoration=snake]
   \draw[decorate, ->]
    (-.01,-.5) -- (.54,-.5);
   \draw[decorate, ->]
    (0.01,-.5) -- (-0.54,-.5);
 \end{tikzpicture}
 }
 \hspace{.25cm}
 \adjustbox{
   scale=.3,
   raise=-.6cm
 }{
 \begin{tikzpicture}
   \draw[
     line width=3
   ]
   (-.7,2) -- (-.7,-2);
   \draw[
     line width=3
   ]
   (+.7,2) -- (+.7,-2);
 \end{tikzpicture}
 }
 \hspace{.3cm}
 \adjustbox{scale=.4}{
 \begin{tikzpicture}[decoration=snake]
   \draw[decorate, ->]
    (-.01,0) -- (0.54,0);
   \draw[decorate, ->]
    (0.01,0) -- (-0.54,0);
 \end{tikzpicture}
 }
 \hspace{+.3cm}
 \adjustbox{
   scale=.3,
   raise=-.7cm,
   rotate=180,
 }{
 \begin{tikzpicture}
   \draw[
     line width=3
   ]
   (-1.5,2)
     .. controls (-1.5,1.3) and (0.2,1.2) ..
     (0.2,0)
     .. controls 
     (0.2,-1.2) and (-1.5,-1.3) ..
   (-1.5,-2);
\begin{scope}[
  shift={(-.5,0)},
  xscale=-1,
]
   \draw[
     line width=12,
     white
   ]
   (-1.5,2)
     .. controls (-1.5,1.3) and (0.2,1.2) ..
     (0.2,0)
     .. controls 
     (0.2,-1.2) and (-1.5,-1.3) ..
   (-1.5,-2);
\end{scope}
\begin{scope}[
  shift={(-.5,0)},
  xscale=-1,
]
   \draw[
     line width=3
   ]
   (-1.5,2)
     .. controls (-1.5,1.3) and (0.2,1.2) ..
     (0.2,0)
     .. controls 
     (0.2,-1.2) and (-1.5,-1.3) ..
   (-1.5,-2);
\end{scope}
 \end{tikzpicture}
 }
 \\[-12pt]
 &
 \\
 \hline
 &
 \\[-12pt]
 \adjustbox{
   raise=.3cm
 }{\small 
   3rd 
 }
 &
 \adjustbox{
   raise=-.8cm,
   scale=.3
 }{
 \begin{tikzpicture}
  \draw[
    line width=3, 
    shift={(-3,0)},
    xscale=-1
  ]
  (.5,2) 
  .. controls (.5,0) and (-3.5,0)
  ..
  (-3.5,-2);
   \draw[
     line width=12, white,
     shift={(-3,0)},
     xscale=-1,
   ]
   (-1.5,2)
     .. controls (-1.5,1.3) and (0.2,1.2) ..
     (0.2,0)
     .. controls 
     (0.2,-1.2) and (-1.5,-1.3) ..
   (-1.5,-2);
   \draw[
     line width=3,
     shift={(-3,0)},
     xscale=-1
   ]
   (-1.5,2)
     .. controls (-1.5,1.3) and (0.2,1.2) ..
     (0.2,0)
     .. controls 
     (0.2,-1.2) and (-1.5,-1.3) ..
   (-1.5,-2);
  \draw[line width=12,white]
  (.5,2) 
  .. controls (.5,0) and (-3.5,0)
  ..
  (-3.5,-2);
  \draw[line width=3]
  (.5,2) 
  .. controls (.5,0) and (-3.5,0)
  ..
  (-3.5,-2);
  \end{tikzpicture}
  }
  \hspace{.3cm}
  \adjustbox{
    scale=.3,
    raise=.4cm
  }{
 \begin{tikzpicture}[
   decoration=snake
  ]
   \draw[decorate, ->]
    (-.01,0) -- (0.54,0);
   \draw[decorate, ->]
    (0.01,0) -- (-0.54,0);
 \end{tikzpicture}
 }
  \hspace{.3cm}
 \adjustbox{
   raise=-.8cm,
   scale=.3
  }{
 \begin{tikzpicture}
  \draw[
    line width=3, 
    shift={(-3,0)},
    xscale=-1
  ]
  (.5,2) 
  .. controls (.5,0) and (-3.5,0)
  ..
  (-3.5,-2);
   \draw[
     line width=12, 
     white
   ]
   (-1.5,2)
     .. controls (-1.5,1.3) and (0.2,1.2) ..
     (0.2,0)
     .. controls 
     (0.2,-1.2) and (-1.5,-1.3) ..
   (-1.5,-2);
   \draw[
     line width=3
   ]
   (-1.5,2)
     .. controls (-1.5,1.3) and (0.2,1.2) ..
     (0.2,0)
     .. controls 
     (0.2,-1.2) and (-1.5,-1.3) ..
   (-1.5,-2);
  \draw[
    line width=12,
    white
  ]
  (.5,2) 
  .. controls (.5,0) and (-3.5,0)
  ..
  (-3.5,-2);
  \draw[line width=3]
  (.5,2) 
  .. controls (.5,0) and (-3.5,0)
  ..
  (-3.5,-2);
  \end{tikzpicture}
  }
  \\[-12pt]
  &
 \\
 \hline
\end{tabular}

\smallskip 
\begin{definition}[\bf Crossing-, Linking-, Framing-numbers and Writhe]
\label{LinkingNumber}
$\,$
\begin{itemize}
\item[{\bf (i)}] Any crossing in a framed oriented link diagram $L$ (Def. \ref{FramedOrientedLinks}) locally is either of the following (up to local orientation-preserving diffeomorphism), which we assign the {\it crossing number} $\pm 1$, respectively, as shown:

\vspace{-.1cm}
\begin{equation}
\label{TotalLinkingNumber}
\#\left(\!
\adjustbox{raise=-.63cm, scale=.5}{
\begin{tikzpicture}

\draw[
  line width=1.2,
  -Latex
]
  (-.7,-.7) -- (.7,.7);
\draw[
  line width=7,
  white
]
  (+.7,-.7) -- (-.7,.7);
\draw[
  line width=1.2,
  -Latex
]
  (+.7,-.7) -- (-.7,.7);
 
\end{tikzpicture}
}
\!\right)
\;:=\;
+1
\,,
\hspace{1cm}
\#\left(\!
\adjustbox{raise=-.63cm, scale=0.5}{
\begin{tikzpicture}[xscale=-1]

\draw[
  line width=1.2,
  -Latex
]
  (-.7,-.7) -- (.7,.7);
\draw[
  line width=7,
  white
]
  (+.7,-.7) -- (-.7,.7);
\draw[
  line width=1.2,
  -Latex
]
  (+.7,-.7) -- (-.7,.7);
 
\end{tikzpicture}
}
\!\right)
\;\;
:=
\;\;
-1
\,.
\end{equation}

\item[{(\bf ii)}] For $(L_i)_{i =1}^N$ the connected components of $L$, the {\it linking number} $\mathrm{lnk}(L_i, L_j)$, defined  for $i \neq j$, is {\it half} the sum of crossing numbers between $L_i$ and $L_j$ (cf. \cite[p. 7]{Ohtsuki01}).

\item[{(\bf iii)}] The {\it framing number} $\mathrm{fr}(L_i)$ is the sum of crossing numbers of $L_i$ with itself.

\item[{\bf (iv)}] The sum $\#L$ of the crossing numbers of all crossings of $L$ 
-- hence the {\it total crossing number}, also called the {\it writhe} \cite[p. 152]{Adams94}\cite[p. 523]{Ohtsuki01} \footnote{
  To beware that some authors, especially in the physics literature, use the term ``writhe'' only in reference to connected (components of) links, where it is the framing number. 
  Hence a more unambiguous term for ``writhe'' in our context is ``total crossing number''.
} --
is hence the sum of all the framing and linking numbers:

\begin{equation}
  \label{TotalCrossingNumber}
  \mathllap{
    \scalebox{.7}{
      \color{darkblue}
      \bf
      \def\arraystretch{.9}
      \begin{tabular}{c}
        Total crossing number
        \\
        {\color{black} / } writhe
      \end{tabular}
    }
    \;\;
  }
  \#(L)
  \;\;:=\;\;
  \underset{
    \mathclap{
     \substack{
        c \in 
        \\
        \mathrm{crssngs}(L)
      }
    }
  }{\sum}
  \;
  \#(c)
  \;\;=\;\;
  \sum_{i}
  \mathrm{frm}(L_i)
  \,+\,
  \sum_{i \neq j} \mathrm{lnk}(L_i, L_j)
  \,.
\end{equation}
\end{itemize}
\end{definition}

\newpage 
\begin{example}[\bf Link invariants]
The framing and linking numbers (Def. \ref{LinkingNumber}) are individually invariants of framed links, in that they depend only on the equivalence class of 
a framed oriented link diagram. The following moves show how successive self-crossings of opposite crossing number 
cancel out under the Reidemeister moves, (Fig. \ref{FigureReidemeister}):

\vspace{-2mm}
\begin{equation}
\adjustbox{
  scale=.95
}{
  \def\tabcolsep{0pt}
  \begin{tabular}{cccccccccc}
\adjustbox{
  scale=.7,
  raise=-1.8cm
}{
\begin{tikzpicture}[
  rotate=90
]

\draw[
  line width=1.2,
  -Latex
]
  (-1.3,-1) --
  (0,-1)
  .. controls 
     (.5,-1) and (1,-.5) ..
  (1,0)
  .. controls 
     (1,.5) and (.5,1) ..
  (0,1)  
  .. controls
    (-.5,1) and (-1,.5)
  .. (-1,0)
  .. controls
    (-1,-.4) and (-.5, -.8) ..
  (0,-.8);
\draw[
  line width=4.5,
  white
]
 (.2,-.8) -- (1,-.8);

\begin{scope}[
  shift={(0,-1.6)},
  yscale=-1
]
\draw[
  line width=1.2,
]
 (0,-.8) -- (2.5,-.8);

\draw[
  line width=7,
  white
]
 (2.5,-.8)
 .. controls 
   (3,-.8) and (3.5,-.4) ..
 (3.5,0)
 .. controls
   (3.5,.5) and (3,1) ..
 (2.5,1)
  .. controls
    (2,1) and (1.5,.5) ..
  (1.5,0)
  .. controls
    (1.5,-.5) and (2,-1) ..
  (2.5,-1)
  --
  (3.8,-1);
\draw[
  -Latex,
  line width=1.2,
]
 (2.5,-.8)
 .. controls 
   (3,-.8) and (3.5,-.4) ..
 (3.5,0)
 .. controls
   (3.5,.5) and (3,1) ..
 (2.5,1)
  .. controls
    (2,1) and (1.5,.5) ..
  (1.5,0)
  .. controls
    (1.5,-.5) and (2,-1) ..
  (2.5,-1)
  --
  (3.8,-1);
\end{scope}

\node
  at (.35,-.48) {
    \scalebox{.7}{$-$}
  };
\node
  at (2.1,-1.05) {
    \scalebox{.7}{$+$}
  };

\end{tikzpicture}
}
  &
  \hspace{-5pt}
\adjustbox{scale=.8}{
 \begin{tikzpicture}[
   decoration=snake 
 ]
   \draw[decorate, ->]
    (-.01,0) -- (0.54,0);
   \draw[decorate, ->]
    (0.01,0) -- (-0.54,0);
 \end{tikzpicture}
}
&
\adjustbox{
  scale=.7,
  raise=-1.8cm
}{
\begin{tikzpicture}[
   rotate=90
]
  \draw[
    line width=1.2,
    -Latex
 ]
    (-2.9,0) -- (2.2,0);
\end{tikzpicture}
}
&
\adjustbox{scale=.8}{
 \begin{tikzpicture}[
   decoration=snake 
 ]
   \draw[decorate, ->]
    (-.01,0) -- (0.54,0);
   \draw[decorate, ->]
    (0.01,0) -- (-0.54,0);
 \end{tikzpicture}
}
&
\adjustbox{
  scale=.7,
  raise=-1.8cm
}{
\begin{tikzpicture}[
  rotate=90
]

\draw[
  line width=1.2,
  -Latex
]
  (-1.3,-1) --
  (0,-1)
  .. controls 
     (.5,-1) and (1,-.5) ..
  (1,0)
  .. controls 
     (1,.5) and (.5,1) ..
  (0,1)  
  .. controls
    (-.5,1) and (-1,.5)
  .. (-1,0)
  .. controls
    (-1,-.4) and (-.5, -.8) ..
  (0,-.8);
\draw[
  -Latex,
  line width=1.2,
]
 (2.5,-.8)
 .. controls 
   (3,-.8) and (3.5,-.4) ..
 (3.5,0)
 .. controls
   (3.5,.5) and (3,1) ..
 (2.5,1)
  .. controls
    (2,1) and (1.5,.5) ..
  (1.5,0)
  .. controls
    (1.5,-.5) and (2,-1) ..
  (2.5,-1)
  --
  (3.8,-1);
\draw[
  line width=4.5,
  white
]
 (0,-.8) -- (2.5,-.8);
\draw[
  line width=1.2,
]
 (0,-.8) -- (2.5,-.8);

\node 
  at (.4,-.64) {
    \scalebox{.7}{$-$}
  };
\node 
  at (2.1,-.62) {
    \scalebox{.7}{$+$}
  };

\end{tikzpicture}
}
\hspace{25pt}
&
\hspace{15pt}
\adjustbox{
  scale=.7,
  raise=-1.8cm
}{
\begin{tikzpicture}[
  yscale=-1,
  rotate=-90
]

\draw[
  line width=1.2,
  -Latex
]
  (-1.3,-1) --
  (0,-1)
  .. controls 
     (.5,-1) and (1,-.5) ..
  (1,0)
  .. controls 
     (1,.5) and (.5,1) ..
  (0,1)  
  .. controls
    (-.5,1) and (-1,.5)
  .. (-1,0)
  .. controls
    (-1,-.4) and (-.5, -.8) ..
  (0,-.8);
\draw[
  line width=4.5,
  white
]
 (.2,-.8) -- (1,-.8);

\begin{scope}[
  shift={(0,-1.6)},
  yscale=-1
]
\draw[
  line width=1.2,
]
 (0,-.8) -- (2.5,-.8);

\draw[
  line width=7,
  white
]
 (2.5,-.8)
 .. controls 
   (3,-.8) and (3.5,-.4) ..
 (3.5,0)
 .. controls
   (3.5,.5) and (3,1) ..
 (2.5,1)
  .. controls
    (2,1) and (1.5,.5) ..
  (1.5,0)
  .. controls
    (1.5,-.5) and (2,-1) ..
  (2.5,-1)
  --
  (3.8,-1);
\draw[
  -Latex,
  line width=1.2,
]
 (2.5,-.8)
 .. controls 
   (3,-.8) and (3.5,-.4) ..
 (3.5,0)
 .. controls
   (3.5,.5) and (3,1) ..
 (2.5,1)
  .. controls
    (2,1) and (1.5,.5) ..
  (1.5,0)
  .. controls
    (1.5,-.5) and (2,-1) ..
  (2.5,-1)
  --
  (3.8,-1);
\end{scope}

\node
  at (.35,-.6) {
    \scalebox{.7}{$+$}
  };
\node
  at (2.1,-1.05) {
    \scalebox{.7}{$-$}
  };

\end{tikzpicture}
}
\hspace{-10pt}
&
\adjustbox{scale=.8}{
 \begin{tikzpicture}[
   decoration=snake 
 ]
   \draw[decorate, ->]
    (-.01,0) -- (0.54,0);
   \draw[decorate, ->]
    (0.01,0) -- (-0.54,0);
 \end{tikzpicture}
}
&
\adjustbox{
  scale=.7,
  raise=-1.8cm
}{
\begin{tikzpicture}[
   rotate=90
]
  \draw[
    line width=1.2,
    -Latex
 ]
    (-2.9,0) -- (2.2,0);
\end{tikzpicture}
}
&
\adjustbox{scale=.8}{
 \begin{tikzpicture}[
   decoration=snake 
 ]
   \draw[decorate, ->]
    (-.01,0) -- (0.54,0);
   \draw[decorate, ->]
    (0.01,0) -- (-0.54,0);
 \end{tikzpicture}
}
&
\adjustbox{
  scale=.7,
  raise=-1.8cm
}{
\begin{tikzpicture}[
  rotate=90
]

\draw[
  line width=1.2,
]
 (0,-.8) -- (2.5,-.8);

\draw[
  line width=4.5,
  white
]
  (-1.3,-1) --
  (0,-1)
  .. controls 
     (.5,-1) and (1,-.5) ..
  (1,0)
  .. controls 
     (1,.5) and (.5,1) ..
  (0,1)  
  .. controls
    (-.5,1) and (-1,.5)
  .. (-1,0)
  .. controls
    (-1,-.4) and (-.5, -.8) ..
  (0,-.8);

\draw[
  line width=1.2,
  -Latex
]
  (-1.3,-1) --
  (0,-1)
  .. controls 
     (.5,-1) and (1,-.5) ..
  (1,0)
  .. controls 
     (1,.5) and (.5,1) ..
  (0,1)  
  .. controls
    (-.5,1) and (-1,.5)
  .. (-1,0)
  .. controls
    (-1,-.4) and (-.5, -.8) ..
  (0,-.8);

\draw[
  line width=4.5,
  white
]
 (2.5,-.8)
 .. controls 
   (3,-.8) and (3.5,-.4) ..
 (3.5,0)
 .. controls
   (3.5,.5) and (3,1) ..
 (2.5,1)
  .. controls
    (2,1) and (1.5,.5) ..
  (1.5,0)
  .. controls
    (1.5,-.5) and (2,-1) ..
  (2.5,-1)
  --
  (3.5,-1);
\draw[
  -Latex,
  line width=1.2,
]
 (2.5,-.8)
 .. controls 
   (3,-.8) and (3.5,-.4) ..
 (3.5,0)
 .. controls
   (3.5,.5) and (3,1) ..
 (2.5,1)
  .. controls
    (2,1) and (1.5,.5) ..
  (1.5,0)
  .. controls
    (1.5,-.5) and (2,-1) ..
  (2.5,-1)
  --
  (3.8,-1);

\node 
  at (.4,-.62) {
    \scalebox{.7}{$+$}
  };
\node 
  at (2.1,-.6) {
    \scalebox{.7}{$-$}
  };

\end{tikzpicture}
}
  \end{tabular}
}
\end{equation}
\end{example}

\vspace{.2cm}

\begin{definition}[\bf Signed interval loops as framed oriented links]
\label{ChargedStringLoopsAsFramedOrientedLinks}
  From a framed oriented link diagram (Def. \ref{FramedOrientedLinks}), we obtain a based loop in the configuration space of signed intervals in $\mathbb{R}^2$ (Def. \ref{ConfigurationsOfChargedStrings}) by thickening the underlying link to a string worldsheet as illustrated in Figure \ref{FigureStringLoopsAsLinks} below:
  \begin{equation}
    \label{MappingFramedLinksToStringLoops}
    \begin{tikzcd}
      \mathrm{FrmdOrntdLnkDgrm}
      \ar[
        r,
      ]
      &
      \Omega
      \,
      \mathrm{Conf}^I_0(\mathbb{R}^2)
      \,.
    \end{tikzcd}
  \end{equation}
To note that this is well-defined due to our condition in Def. \ref{FramedOrientedLinks} that link diagrams have well-separated crossings and no straight horizontal segments:
These conditions imply that the intersection of the link diagram with any horizontal line $\mathbb{R}^1 \hookrightarrow \mathbb{R}^2$ is a finite set of points, and that as we move the horizontal line vertically, these points (i) move, (ii) cross, (iii) (e)merge over well-separated intervals, thus translating to the corresponding string worldsheets, where the orientation of the link determines the charges on the endpoints of these strings.
\end{definition}

\hspace{0cm}
\def\tabcolsep{5pt}
\begin{tabular}{p{5.3cm}}

\vspace{.3cm}
\footnotesize
{\bf Figure \figurenumber: 
\label{FigureStringLoopsAsLinks}
Framed oriented links as loops of interval configurations.} Here it is the stringy nature of the loops of configurations on the right (via Def. \ref{ConfigurationsOfChargedStrings}) that reflects the ``blackboard framing'' of the link diagrams on the left. This framing would be absent for configurations of charged {\it points} as on the left of Fig. \ref{RelationsOfSignedConfigurations}.
\end{tabular}
\def\tabcolsep{2pt}
\adjustbox{
  raise=-.5cm,
  margin=-5pt,
  scale=.6,
  fbox
}{
\begin{tabular}{ccc}

\adjustbox{raise=-1.5cm}{
\begin{tikzpicture}
  \draw[
    line width=1.4,
    -Latex
  ]
    (0:1.6) arc (0:180:1.6);
  \draw[
    line width=1.4,
    -Latex
  ]
    (180:1.6) arc (180:360:1.6);
\end{tikzpicture}
}
&
\scalebox{1.6}{
$\mapsto$
}
&
\adjustbox{
  scale=.5,
  raise=-2cm
}{
\begin{tikzpicture}
 \begin{scope}
   \halfcircle{0,0};
 \end{scope}
 \begin{scope}[yscale=-1]
   \halfcircle{0,0};
 \end{scope} 
\end{tikzpicture}
}
\\
\adjustbox{raise=-1.5cm}{
\begin{tikzpicture}
  \draw[
    line width=1.4,
    -Latex
  ]
    (0:1.6) arc (0:180:1.6);

\begin{scope}[shift={(2,0)}]
  \draw[
    line width=10,
    white
  ]
    (180:1.6) arc (180:0:1.6);
  \draw[
    line width=1.4,
    -Latex,
  ]
    (180:1.6) arc (180:0:1.6);
  \draw[
    line width=1.4,
    -Latex,
  ]
    (0:1.6) arc (0:-180:1.6);
\end{scope}
  
  \draw[
    line width=10,
    white
  ]
    (180:1.6) arc (180:360:1.6);
  \draw[
    line width=1.4,
    -Latex
  ]
    (180:1.6) arc (180:360:1.6);

\end{tikzpicture}
}
&
\scalebox{1.6}{
$\mapsto$
}
&
\adjustbox{
  scale=.5,
  raise=-1.7cm
}{
\begin{tikzpicture}
\begin{scope}[yscale=-1]
  \halfcircle{0,0};
\end{scope}
\begin{scope}[
  shift={(3.4,.2)},
  xscale=-1
]
\halfcircle{0}{0};
\end{scope}
\begin{scope}[shift={(0,0)}]
\halfcircleover{0}{0};
\end{scope}
\begin{scope}[
  yscale=-1,
  shift={(3.4,-.2)},
  xscale=-1
]
\halfcircleover{0}{0};
\end{scope}
  
\end{tikzpicture}
}
\\
\adjustbox{raise=-1.5cm}{
\begin{tikzpicture}
  \draw[
    line width=1.4,
    -Latex
  ]
    (0:1.6) arc (0:180:1.6);

\begin{scope}[shift={(3.2,0)}]
  \draw[
    line width=12,
    white
  ]
    (180:1.6) arc (180:0:1.6);
  \draw[
    line width=1.4,
    -Latex,
  ]
    (180:1.6) arc (180:0:1.6);
  \draw[
    line width=1.4,
    -Latex,
  ]
    (0:1.6) arc (0:-180:1.6);
\end{scope}
  
  \draw[
    line width=12,
    white
  ]
    (180:1.6) arc (180:360:1.6);
  \draw[
    line width=1.4,
    -Latex
  ]
    (180:1.6) arc (180:360:1.6);

\end{tikzpicture}
}
&
\scalebox{1.6}{
$\mapsto$
}
&
\adjustbox{
  scale=.5,
  raise=-1.5cm
}{
\begin{tikzpicture}

\begin{scope}[yscale=-1]
  \halfcircle{0}{0}
\end{scope}
\begin{scope}[shift={(6,0)},scale=-1]
  \halfcircleover{0}{0}
\end{scope}
\begin{scope}[shift={(6,0)}, xscale=-1]
\halfcircle{0,0}
\end{scope}
\begin{scope}
\halfcircleover{0}{0}
\end{scope}

\end{tikzpicture}
}
\\
\adjustbox{
  raise=-3.5cm
}{
\begin{tikzpicture}[
  rotate=-90
]
  \draw[
    line width=1.4,
    -Latex
  ]
    (0:1.6) arc (0:180:1.6);

\begin{scope}[shift={(3.2,0)}]
  \draw[
    line width=12,
    white
  ]
    (180:1.6) arc (180:0:1.6);
  \draw[
    line width=1.4,
    -Latex,
  ]
    (180:1.6) arc (180:0:1.6);
  \draw[
    line width=1.4,
    -Latex,
  ]
    (0:1.6) arc (0:-180:1.6);
\end{scope}
  
  \draw[
    line width=12,
    white
  ]
    (180:1.6) arc (180:360:1.6);
  \draw[
    line width=1.4,
    -Latex
  ]
    (180:1.6) arc (180:360:1.6);

\end{tikzpicture}
}

&
$\longmapsto$
&

\adjustbox{
  raise=-7cm,
  scale=.5
}{
\begin{tikzpicture}

\begin{scope}[
  shift={(.1,6)},
  yscale=-1
]
\halfcircle
\end{scope}

\begin{scope}[
  shift={(0,6)},
]
\clip (+4,-4) rectangle 
      (+0,+0);
\halfcircleAdjusted
\end{scope}

\begin{scope}[
  shift={(0,0)},
  yscale=-1,
  xscale=-1
]
\clip (+4,-4) rectangle 
      (+0,+0);
\halfcircleAdjusted
\end{scope}

\begin{scope}[xscale=-1]
\halfcircle
\end{scope}
\begin{scope}[
  xscale=-1,
  yscale=-1
]
\clip (-4,-4) rectangle 
      (+0,+0);
\halfcircleoverAdjusted
\end{scope}

\begin{scope}[
  shift={(.2,6)},
]
\clip (-4,-4) rectangle 
      (-.1,+0);
\halfcircleoverAdjusted
\end{scope}

\end{tikzpicture}
}
\\[-14pt]
&&
\end{tabular}
}

\vspace{.3cm}

Next we want to show that this construction (in Def. \ref{ChargedStringLoopsAsFramedOrientedLinks} of signed interval loops from framed oriented link diagrams) descends to equivalence classes on both sides. For this we need a
 more combinatorial description of equivalence of link diagrams. This is provided by {\it functorial knot theory} \cite{Yetter01} via {\it Shum's theorem} (Prop. \ref{ShumTheorem}, Fig. \ref{FigureRibbonCategory}).

\vspace{.2cm}

\hspace{1.8cm}
\def\tabcolsep{3pt}
\adjustbox{
  margin=0pt,
  scale=.8,
  fbox
}{
\begin{tabular}{c||c||c}
  &&
  \\[-15pt]
  \bf
  \def\arraystretch{.8}
  \begin{tabular}{c}
    Ribbon
    \\
    category
    \\
    axiom
  \end{tabular}
  &
  \begin{minipage}{8.4cm}
    \footnotesize
    {\bf Figure
    \figurenumber:
    \label{FigureRibbonCategory}
    Ribbon category presentation} of framed oriented links according to Shum's theorem (recalled as Prop. \ref{ShumTheorem}). For the slide moves there is also the corresponding mirrored version, which we are not showing just in order to save space.
  \end{minipage}
  &
  \bf
  \def\arraystretch{.8}
  \begin{tabular}{c}
    Framed oriented
    \\
    link diagram
    \\
    move
  \end{tabular}
  \\[-12pt]
  &&
  \\
  \hline
  &&
  \\[-12pt]
  \adjustbox{
    rotate=90,
    raise=-1.8cm
  }{
    \bf Braiding isomorphy
  }
  &
  \adjustbox{
    raise=-2cm
  }{
  \begin{tikzpicture}
    \draw[
      line width=1.2,
      -Latex
    ]
    (+1,-1)
    .. controls (+1,0) and (-1,0)
    ..
    (-1,1);
    \draw[
      color=white,
      line width=8,
    ]
    (-1,-1)
    .. controls (-1,0) and (+1,0)
    ..
    (+1,1);
    \draw[
      line width=1.2,
      -Latex
    ]
    (-1,-1)
    .. controls (-1,0) and (+1,0)
    ..
    (+1,1);

    \begin{scope}[
      shift={(0,1.9)},
      xscale=-1
    ]
    \draw[
      line width=1.2,
    ]
    (+1,-1)
    .. controls (+1,0) and (-1,0)
    ..
    (-1,1);
    \draw[
      color=white,
      line width=8,
    ]
    (-1,-1)
    .. controls (-1,0) and (+1,0)
    ..
    (+1,1);
    \draw[
      line width=1.2,
    ]
    (-1,-1)
    .. controls (-1,0) and (+1,0)
    ..
    (+1,1);
    \end{scope}
  \end{tikzpicture}
  }
\hspace{.1cm}
\begin{tikzpicture}[decoration=snake]
  \draw[decorate,->] (0,0) -- (+.55,0);
  \draw[decorate,->] (0,0) -- (-.55,0);
\end{tikzpicture}
\hspace{.1cm}
\adjustbox{
  raise=-2cm
}{
\begin{tikzpicture}[
   line width=1.2,
]
\draw[
   -Latex
]
  (0,-2) -- (0,.1);
\draw
  (0,0) -- (0,+2);
\end{tikzpicture}
}
\hspace{.1cm}
\begin{tikzpicture}[decoration=snake]
  \draw[decorate,->] (0,0) -- (+.55,0);
  \draw[decorate,->] (0,0) -- (-.55,0);
\end{tikzpicture}
\hspace{.1cm}
  \adjustbox{
    raise=-2cm
  }{
  \begin{tikzpicture}[
    xscale=-1
  ]
    \draw[
      line width=1.2,
      -Latex
    ]
    (+1,-1)
    .. controls (+1,0) and (-1,0)
    ..
    (-1,1);
    \draw[
      color=white,
      line width=8,
    ]
    (-1,-1)
    .. controls (-1,0) and (+1,0)
    ..
    (+1,1);
    \draw[
      line width=1.2,
      -Latex
    ]
    (-1,-1)
    .. controls (-1,0) and (+1,0)
    ..
    (+1,1);

    \begin{scope}[
      shift={(0,1.9)},
      xscale=-1
    ]
    \draw[
      line width=1.2,
    ]
    (+1,-1)
    .. controls (+1,0) and (-1,0)
    ..
    (-1,1);
    \draw[
      color=white,
      line width=8,
    ]
    (-1,-1)
    .. controls (-1,0) and (+1,0)
    ..
    (+1,1);
    \draw[
      line width=1.2,
    ]
    (-1,-1)
    .. controls (-1,0) and (+1,0)
    ..
    (+1,1);
    \end{scope}
  \end{tikzpicture}
  }
  &
  \adjustbox{
    rotate=-90,
    raise=+2.2cm
  }{
    \bf
    \hspace{-4pt}
    2nd Reidemeister move\;
  }
  \\
  &&
  \\[-16pt]
  \hline
  &&
  \\[-24pt]
  \adjustbox{
    rotate=90,
    raise=-1.6cm
  }{
    \bf
    Braiding naturality
  }
  &
\adjustbox{scale=.85}{
\adjustbox{
  raise=-2.5cm
}
{
\begin{tikzpicture}

\begin{scope}[
  xscale=-1
]
  \begin{scope}[
    shift={(-.35,0)}
  ]
  \draw[
    line width=1.2,
    -Latex
  ] 
   (-1.5,-2.2)
   .. controls 
     (-1.5, +.2) and (+1.5,+.1)
   ..
   (+1.5,+2.2);
  \end{scope}
  \begin{scope}[
    shift={(+.35,0)}
  ]
  \draw
    (+.98,1.2) node {
      \scalebox{1.8}{...}
    };
  \draw[
    line width=1.2,
    -Latex
  ] 
   (-1.5,-2.2)
   .. controls 
     (-1.5, -.2) and (+1.5,-.2)
   ..
   (+1.5,+2.2);
  \end{scope}
\end{scope}

  \draw[
    draw=white,
    line width=34
  ]
   (-1.5,-2.2) .. controls
   (-1.5,0) and (+1.5,0) ..
   (+1.5,+2.2);

  \begin{scope}[
    shift={(-.35,0)}
  ]
  \draw[
    line width=1.2,
    -Latex
  ] 
   (-1.5,-2.2)
   .. controls 
     (-1.5, +.2) and (+1.5,+.1)
   ..
   (+1.5,+2.2);
  \end{scope}
  \begin{scope}[
    shift={(+.35,0)}
  ]
  \draw
    (+.98,1.2) node {
      \scalebox{1.8}{...}
    };
  \draw[
    line width=1.2,
    -Latex
  ] 
   (-1.5,-2.2)
   .. controls 
     (-1.5, -.2) and (+1.5,-.2)
   ..
   (+1.5,+2.2);
  \end{scope}

  \begin{scope}[
    shift={(0,.3)}  
  ]
  \draw[
    fill=white,
    rounded corners
  ]
    (-1.95,-1.2) rectangle 
    (-.8,-2);
    \node
      at (-1.4,-1.6)
      {A};
  \end{scope}
  \begin{scope}[
    xscale=-1,
    shift={(0,.3)}
  ]
  \draw[
    fill=white,
    rounded corners,
  ]
    (-1.95,-1.2) rectangle 
    (-.8,-2);
    \node
      at (-1.4,-1.6)
      {B};
  \end{scope}
\end{tikzpicture}
}
\hspace{.1cm}
\begin{tikzpicture}[decoration=snake]
  \draw[decorate,->] (0,0) -- (+.55,0);
  \draw[decorate,->] (0,0) -- (-.55,0);
\end{tikzpicture}
\hspace{.1cm}
\adjustbox{
  raise=-2.5cm
}
{
\begin{tikzpicture}

\begin{scope}[
  xscale=-1
]
  \begin{scope}[
    shift={(-.35,0)}
  ]
  \draw[
    line width=1.2,
    -Latex
  ] 
   (-1.5,-2.2)
   .. controls 
     (-1.5, +.2) and (+1.5,+.1)
   ..
   (+1.5,+2.2);
  \end{scope}
  \begin{scope}[
    shift={(+.35,0)}
  ]
  \draw
    (+.98,1.2) node {
      \scalebox{1.8}{...}
    };
  \draw[
    line width=1.2,
    -Latex
  ] 
   (-1.5,-2.2)
   .. controls 
     (-1.5, -.2) and (+1.5,-.2)
   ..
   (+1.5,+2.2);
  \end{scope}
\end{scope}

  \draw[
    draw=white,
    line width=34
  ]
   (-1.5,-2.2) .. controls
   (-1.5,0) and (+1.5,0) ..
   (+1.5,+2.2);
 
  \begin{scope}[
    shift={(-.35,0)}
  ]
  \draw[
    line width=1.2,
    -Latex
  ] 
   (-1.5,-2.2)
   .. controls 
     (-1.5, +.2) and (+1.5,+.1)
   ..
   (+1.5,+2.2);
  \end{scope}
  \begin{scope}[
    shift={(+.35,0)}
  ]
  \draw
    (+.98,1.2) node {
      \scalebox{1.8}{...}
    };
  \draw[
    line width=1.2,
    -Latex
  ] 
   (-1.5,-2.2)
   .. controls 
     (-1.5, -.2) and (+1.5,-.2)
   ..
   (+1.5,+2.2);
  \end{scope}

  \begin{scope}[
    yscale=-1,
    shift={(0,.3)}
  ]
  \draw[
    fill=white,
    rounded corners
  ]
    (-1.95,-1.2) rectangle 
    (-.8,-2);
    \node
      at (-1.4,-1.6)
      {B};
  \end{scope}
  \begin{scope}[
    xscale=-1,
    yscale=-1,
    shift={(0,.3)}
  ]
  \draw[
    fill=white,
    rounded corners
  ]
    (-1.95,-1.2) rectangle 
    (-.8,-2);
    \node
      at (-1.4,-1.6)
      {A};
  \end{scope}
\end{tikzpicture}
}  
}
  \\
  &&
  \\[-22pt]
  \cline{1-2}
  &&
  \adjustbox{
    rotate=-90,
    raise=.2cm
  }{
    \clap{
      \bf
      sliding moves
      $\Rightarrow$
      3rd Reidemeister move
    }
  }
  \\[-14pt]
  \adjustbox{
    rotate=90,
    raise=-1.9cm
  }{
    \bf
    Braiding hexagon law
  }
  &
  \adjustbox{
    raise=-2cm
  }{
    \begin{tikzpicture}
      \draw[
        line width=1.2,
        -Latex
      ]
      (+1,-2) 
        .. controls 
        (+1,0) and (-1,0) ..
      (-1,+2);
      \begin{scope}[
        shift={(.05,0)}
      ]
      \draw[
        white,
        line width=11,
      ]
      (-1,-2) 
        .. controls 
        (-1,0) and (+1,0) ..
      (+1,+2);
      \end{scope}
      \draw[
        line width=1.2,
        -Latex
      ]
      (-1-.1,-2) 
        .. controls 
        (-1-.1,0) and (+1-.1,0) ..
      (+1-.1,+2);
      \draw[
        line width=1.2,
        -Latex
      ]
      (-1+.1,-2) 
        .. controls 
        (-1+.1,0) and (+1+.1,0) ..
      (+1+.1,+2);
    \end{tikzpicture}
  }
\begin{tikzpicture}[decoration=snake]
  \draw[decorate,->] (0,0) -- (+.55,0);
  \draw[decorate,->] (0,0) -- (-.55,0);
\end{tikzpicture}
\hspace{4pt}
\adjustbox{
  raise=-2cm
}{
  \begin{tikzpicture}[
    xscale=-1
  ]
  \draw[
    line width=1.2,
    -Latex
  ]
  (-1.5,-2) .. controls
  (-1.5,-1) and (0,-1) ..
  (0,0);
  \draw[
    white,
    line width=9,
  ]
  (0,-2) .. controls
  (0,-1) and (-1.5,-1) ..
  (-1.5,0);
  \draw[
    line width=1.2,
    -Latex
  ]
  (0,-2) .. controls
  (0,-1) and (-1.5,-1) ..
  (-1.5,0);
  \draw[
    line width=1.2,
    -Latex
  ] 
  (1.5,-2) -- (1.5,0);
  \begin{scope}[
    shift={(0,-.1)}
  ]
  \draw[
    line width=1.2
  ]
    (0,0) .. controls
    (0,+1) and (1.5,+1) ..
    (1.5,+2);
  \draw[
    line width=8,
    white
  ]
    (1.5,0) .. controls
    (1.5,+1) and (0,+1) ..
    (0,+2);
  \draw[
    line width=1.2
  ]
    (1.5,0) .. controls
    (1.5,+1) and (0,+1) ..
    (0,+2);
  \draw[
    line width=1.2
  ]
  (-1.5,0) -- (-1.5,+2);
  \end{scope}
  \end{tikzpicture}
}
  &
  \\
  &&
  \\[-15pt]
  \hline
  &&
  \\[-11pt]
  \adjustbox{
    rotate=90,
    raise=-1.8cm
  }{
    \bf
    Duality triangle law
  }
  &
  \adjustbox{
    raise=-2cm
  }{
    \begin{tikzpicture}
      \draw[
        line width=1.2,
        -Latex
      ]
      (-.5,-2.5) -- (-.5,0);
      \begin{scope}[
        shift={(0,-.1)}
      ]
     \draw[
       line width=1.2
     ]
     (180:.5) arc (180:0:.5);
     \end{scope}
     \draw[
       line width=1.2,
       -Latex
    ]
       (.5,-.01) -- (.5,-1.1);
     \begin{scope}[
       shift={(1,-1)}
     ]
     \draw[
       line width=1.2
     ]
     (180:.5) arc (180:360:.5);
     \end{scope}
     \draw[
       line width=1.2,
       -Latex
     ]
      (1.5,-1) -- (1.5,0);
     \draw[
       line width=1.2
     ]
      (1.5,-.1) -- (1.5,1.5-.1);
    \end{tikzpicture}
  }
  \hspace{.1cm}
\begin{tikzpicture}[
  decoration=snake,
]
  \draw[decorate,->] (0,0) -- (+.55,0);
  \draw[decorate,->] (0,0) -- (-.55,0);
\end{tikzpicture}
\hspace{.1cm}
\adjustbox{
  raise=-2cm
}{
  \begin{tikzpicture}
    \draw[
      line width=1.2,
      -Latex
    ]
    (0,-2) -- (0,0.1);
    \draw[
      line width=1.2
    ]
    (0,.1-.1) -- (0,+2-.1);
  \end{tikzpicture}
}
  \hspace{.1cm}
\begin{tikzpicture}[decoration=snake]
  \draw[decorate,->] (0,0) -- (+.55,0);
  \draw[decorate,->] (0,0) -- (-.55,0);
\end{tikzpicture}
\hspace{.1cm}
  \adjustbox{
    raise=-2cm
  }{
    \begin{tikzpicture}[
      xscale=-1
    ]
      \draw[
        line width=1.2,
        -Latex
      ]
      (-.5,-2.5) -- (-.5,0);
      \begin{scope}[
        shift={(0,-.1)}
      ]
     \draw[
       line width=1.2
     ]
     (180:.5) arc (180:0:.5);
     \end{scope}
     \draw[
       line width=1.2,
       -Latex
    ]
       (.5,-.01) -- (.5,-1.1);
     \begin{scope}[
       shift={(1,-1)}
     ]
     \draw[
       line width=1.2
     ]
     (180:.5) arc (180:360:.5);
     \end{scope}
     \draw[
       line width=1.2,
       -Latex
     ]
      (1.5,-1) -- (1.5,0);
     \draw[
       line width=1.2
     ]
      (1.5,-.1) -- (1.5,1.5-.1);
    \end{tikzpicture}
  }
  &
  \adjustbox{
    rotate=-90,
    raise=1.8cm
  }{
    \bf zig-zag yank move
  }
  \\[-10pt]
  &&
  \\
  \hline
  &&
  \\[-26pt]
  &&
  \\
  \adjustbox{
    rotate=90,
    raise=-1.8cm
  }{
    \bf Twist ribbon property
  }
  &
\adjustbox{
  raise=-2.2cm
}{
\begin{tikzpicture}[
  rotate=90,
  scale=.9
]

\draw[
  line width=1.2,
  -Latex
]
  (-1.3,-1) --
  (0,-1)
  .. controls 
     (.5,-1) and (1,-.5) ..
  (1,0)
  .. controls 
     (1,.5) and (.5,1) ..
  (0,1)  
  .. controls
    (-.5,1) and (-1,.5)
  .. (-1,0)
  .. controls
    (-1,-.4) and (-.5, -.8) ..
  (0,-.8);
\draw[
  -Latex,
  line width=1.2,
]
 (2.5,-.8)
 .. controls 
   (3,-.8) and (3.5,-.4) ..
 (3.5,0)
 .. controls
   (3.5,.5) and (3,1) ..
 (2.5,1)
  .. controls
    (2,1) and (1.5,.5) ..
  (1.5,0)
  .. controls
    (1.5,-.5) and (2,-1) ..
  (2.5,-1)
  --
  (3.8,-1);
\draw[
  line width=6,
  white
]
 (0,-.8) -- (2.5,-.8);
\draw[
  line width=1.2,
]
 (0-.1,-.8) -- (2.5,-.8);

\end{tikzpicture}
}
\hspace{.1cm}
\begin{tikzpicture}[decoration=snake]
  \draw[decorate,->] (0,0) -- (+.55,0);
  \draw[decorate,->] (0,0) -- (-.55,0);
\end{tikzpicture}
\hspace{.1cm}
\adjustbox{
  raise=-2cm
}{
\begin{tikzpicture}
  \draw[
    line width=1.2,
    -Latex
  ]
  (0,-2.1) -- (0,.2); 
  \draw[
    line width=1.2,
  ]
  (0,0) -- (0,+2.2); 
\end{tikzpicture}
}
\hspace{.1cm}
\begin{tikzpicture}[decoration=snake]
  \draw[decorate,->] (0,0) -- (+.55,0);
  \draw[decorate,->] (0,0) -- (-.55,0);
\end{tikzpicture}
\hspace{.1cm}
\adjustbox{
  raise=-2.3cm
}{
\begin{tikzpicture}[
  rotate=90,
  scale=.9
]

\draw[
  line width=1.2,
]
 (0-.2,-.8) -- (2.5,-.8);

\draw[
  line width=4.5,
  white
]
  (-1.3,-1) --
  (0,-1)
  .. controls 
     (.5,-1) and (1,-.5) ..
  (1,0)
  .. controls 
     (1,.5) and (.5,1) ..
  (0,1)  
  .. controls
    (-.5,1) and (-1,.5)
  .. (-1,0)
  .. controls
    (-1,-.4) and (-.5, -.8) ..
  (0,-.8);

\draw[
  line width=1.2,
  -Latex
]
  (-1.3,-1) --
  (0,-1)
  .. controls 
     (.5,-1) and (1,-.5) ..
  (1,0)
  .. controls 
     (1,.5) and (.5,1) ..
  (0,1)  
  .. controls
    (-.5,1) and (-1,.5)
  .. (-1,0)
  .. controls
    (-1,-.4) and (-.5, -.8) ..
  (0,-.8);

\draw[
  line width=4.5,
  white
]
 (2.5,-.8)
 .. controls 
   (3,-.8) and (3.5,-.4) ..
 (3.5,0)
 .. controls
   (3.5,.5) and (3,1) ..
 (2.5,1)
  .. controls
    (2,1) and (1.5,.5) ..
  (1.5,0)
  .. controls
    (1.5,-.5) and (2,-1) ..
  (2.5,-1)
  --
  (3.5,-1);
\draw[
  -Latex,
  line width=1.2,
]
 (2.5,-.8)
 .. controls 
   (3,-.8) and (3.5,-.4) ..
 (3.5,0)
 .. controls
   (3.5,.5) and (3,1) ..
 (2.5,1)
  .. controls
    (2,1) and (1.5,.5) ..
  (1.5,0)
  .. controls
    (1.5,-.5) and (2,-1) ..
  (2.5,-1)
  --
  (3.8,-1);

\end{tikzpicture}
}
&
\adjustbox{
  rotate=-90,
  raise=2.3cm
}{
  \bf 1st Reidemeister move
}
\end{tabular}
}

\vspace{1mm}

\begin{proposition}[{\bf Shum's Theorem} {\cite{Shum94}}]
  \label{ShumTheorem}
  Framed oriented links {\rm (Def. \ref{FramedOrientedLinks})} are equivalently the endomorphisms of $\varnothing$ in the category of framed oriented tangles, which is the ribbon category (aka tortile category) freely generated by a single object.
\end{proposition}

What this means here is (cf. \cite[\S 9.1]{Yetter01}) that for a function on link diagrams to descend to equivalence classes and hence to be a {\it link invariant}, it is sufficient that it respects (beyond wiggling of edges) the moves shown in Figure \ref{FigureRibbonCategory}, which subsume the Reidemeister moves (Figure \ref{FigureReidemeister}) but also contains zig-zag yank moves to account for diagram isotopy, combinatorially. 
 While we do not need more than
 Fig. \ref{FigureRibbonCategory} here,
 the interested reader may find more background on it in \cite{Selinger11} (going back to \cite[Prop. 2.7]{JoyalStreet93} for the case of the 3rd Reidemeister move).

\smallskip 
This gives us a good handle on equivalence classes of framed link diagrams. On the other side, we find the corresponding equivalences of loops of signed intervals as follows:

\newpage 
\begin{example}[\bf Link cobordism]
  The first and third move of signed interval loops from Ex. \ref{RelationsBetweenChargedOpenStringWorldsheets}
  relate diagrams whose pre-images under \eqref{MappingFramedLinksToStringLoops}
  are framed oriented link diagrams as shown in the following moves on the left:
\begin{equation}
\label{SaddleAndVacuumMoves}
\hspace{-.7cm}
\scalebox{0.8}{$
\adjustbox{}{
\hspace{-.8cm}
\def\tabcolsep{10pt}
\begin{tabular}{c||c}
\adjustbox{
  raise=-1cm,
}{
\begin{tikzpicture}[yscale=-1]
\draw[
  line width=1.2,
  ]
  (-1,1)
  edge[bend left=40]
  (-1,-1);
\draw[-Latex]
  [line width=1.2]
  (-1+.38,+0)
  --
  (-1+.38,-.0001);

\begin{scope}[
  xscale=-1
]
\draw[
  line width=1.2,
  ]
  (-1,1)
  edge[bend left=40]
  (-1,-1);
\draw[-Latex]
  [line width=1.2]
  (-1+.38,0)
  --
  (-1+.38,+.0001);
\end{scope}
\end{tikzpicture}
}
\;\;\;
\begin{tikzpicture}[decoration=snake]
\draw[decorate,->]
  (0,0) -- (+.55,0);
\draw[decorate,->]
  (0,0) -- (-.55,0);
\end{tikzpicture}
\adjustbox{
  raise=-1cm
}{
\begin{tikzpicture}[xscale=-1]
\begin{scope}[
  shift={(4.5,0)}
]
\draw[
  line width=1.2,
  ]
  (-1,1)
  edge[bend right=40]
  (+1,+1);
\draw[
  line width=1.2,
  -Latex
]
  (0,+.63)
  --
  (0.0001,+.63);

\begin{scope}[
  yscale=-1
]
\draw[
  line width=1.2,
  ]
  (-1,1)
  edge[bend right=40]
  (+1,+1);
\draw[
  line width=1.2,
  -Latex
]
  (+.0001,+.63)
  --
  (0,+.63);
\end{scope}

\end{scope}

\end{tikzpicture}
}
&
\adjustbox{
  scale=.6,
  raise=-1.4cm
}{
\begin{tikzpicture}
\draw[
  gray,
  line width=30pt,
  draw opacity=.4,
]
 (0,2) 
 --
 (0,-2);

\closedinterval{-.5}{1.6}{1}
\closedinterval{-.5}{.8}{1}
\closedinterval{-.5}{0}{1}
\closedinterval{-.5}{-.8}{1}
\closedinterval{-.5}{-1.6}{1}

\begin{scope}[
  xshift=3.3cm,
  xscale=-1
]
\draw[
  gray,
  line width=30pt,
  draw opacity=.4
]
 (0,2) 
 --
 (0,-2);

\openinterval{-.5}{1.6}{1}
\openinterval{-.5}{.8}{1}
\openinterval{-.5}{0}{1}
\openinterval{-.5}{-.8}{1}
\openinterval{-.5}{-1.6}{1}

\end{scope}
\end{tikzpicture}
}
\;\;\;\;
\adjustbox{scale=.8}{
\begin{tikzpicture}[decoration=snake]
\draw[decorate,->, line width=1]
  (0,0) -- (+0.55,0);
\draw[decorate,->, line width=1]
  (0,0) -- (-0.55,0);
\end{tikzpicture}
}
\adjustbox{
  scale=.6,
  raise=-1.1cm
}{
\begin{tikzpicture}

\begin{scope}
\clip
  (-.6,0) rectangle (5,2);
\draw[
  gray,
  line width=28,
  draw opacity=.4
]
  (1.75,2.6) circle (1.8);
\end{scope}
\closedinterval{-.4}{1.75}{1.2};
\closedinterval{-.1}{1.3}{1.75};
\begin{scope}[
  xshift=3.5cm,
  xscale=-1
]
\openinterval{-.4}{1.75}{1.2};
\openinterval{-.1}{1.3}{1.75};
\end{scope}
\oppositeinterval{.32}{.8}{2.85};
\oppositeinterval{1.65}{.33}{.2};

\begin{scope}[yscale=-1]
\begin{scope}
\clip
  (-.6,0) rectangle (5,2);
\draw[
  gray,
  line width=28,
  draw opacity=.4
]
  (1.75,2.6) circle (1.8);
\end{scope}
\closedinterval{-.4}{1.75}{1.2};
\closedinterval{-.1}{1.3}{1.75};
\begin{scope}[
  xshift=3.5cm,
  xscale=-1
]
\openinterval{-.4}{1.75}{1.2};
\openinterval{-.1}{1.3}{1.75};
\end{scope}
\oppositeinterval{.32}{.8}{2.85};
\oppositeinterval{1.65}{.33}{.2};
\end{scope}
  
\end{tikzpicture}
\hspace{-1.7cm}
}
\\[30pt]
\hline
\\[-20pt]
\adjustbox{
  raise=-1cm,
}{
\begin{tikzpicture}
\draw[
  line width=1.2,
  ]
  (-1,1)
  edge[bend left=40]
  (-1,-1);
\draw[-Latex]
  [line width=1.2]
  (-1+.38,+0)
  --
  (-1+.38,-.0001);

\begin{scope}[
  xscale=-1
]
\draw[
  line width=1.2,
  ]
  (-1,1)
  edge[bend left=40]
  (-1,-1);
\draw[-Latex]
  [line width=1.2]
  (-1+.38,0)
  --
  (-1+.38,+.0001);
\end{scope}
\end{tikzpicture}
}
\;\;\;
\begin{tikzpicture}[decoration=snake]
\draw[decorate,->]
  (0,0) -- (+.55,0);
\draw[decorate,->]
  (0,0) -- (-.55,0);
\end{tikzpicture}
\adjustbox{
  raise=-1cm
}{
\begin{tikzpicture}
\begin{scope}[
  shift={(4.5,0)}
]
\draw[
  line width=1.2,
  ]
  (-1,1)
  edge[bend right=40]
  (+1,+1);
\draw[
  line width=1.2,
  -Latex
]
  (0,+.63)
  --
  (0.0001,+.63);

\begin{scope}[
  yscale=-1
]
\draw[
  line width=1.2,
  ]
  (-1,1)
  edge[bend right=40]
  (+1,+1);
\draw[
  line width=1.2,
  -Latex
]
  (+.0001,+.63)
  --
  (0,+.63);
\end{scope}

\end{scope}

\end{tikzpicture}
}
&
\hspace{-.8cm}
\adjustbox{
  scale=.6,
  raise=-1.4cm
}{
\begin{tikzpicture}
\draw[
  gray,
  line width=30pt,
  draw opacity=.4,
]
 (0,2) 
 --
 (0,-2);

\openinterval{-.5}{1.6}{1}
\openinterval{-.5}{.8}{1}
\openinterval{-.5}{0}{1}
\openinterval{-.5}{-.8}{1}
\openinterval{-.5}{-1.6}{1}

\begin{scope}[
  xshift=3.3cm,
  xscale=-1
]
\draw[
  gray,
  line width=30pt,
  draw opacity=.4
]
 (0,2) 
 --
 (0,-2);

\closedinterval{-.5}{1.6}{1}
\closedinterval{-.5}{.8}{1}
\closedinterval{-.5}{0}{1}
\closedinterval{-.5}{-.8}{1}
\closedinterval{-.5}{-1.6}{1}

\end{scope}
\end{tikzpicture}
}
\;\;\;\;
\adjustbox{scale=.8}{
\begin{tikzpicture}[decoration=snake]
\draw[decorate,->, line width=1]
  (0,0) -- (+0.55,0);
\draw[decorate,->, line width=1]
  (0,0) -- (-0.55,0);
\end{tikzpicture}
}
\hspace{-.7cm}
\adjustbox{
  scale=.6,
  raise=-1.1cm
}{
\begin{tikzpicture}[xscale=-1]

\begin{scope}
\clip
  (-.6,0) rectangle (5,2);
\draw[
  gray,
  line width=28,
  draw opacity=.4
]
  (1.75,2.6) circle (1.8);
\end{scope}
\closedinterval{-.4}{1.75}{1.2};
\closedinterval{-.1}{1.3}{1.75};
\begin{scope}[
  xshift=3.5cm,
  xscale=-1
]
\openinterval{-.4}{1.75}{1.2};
\openinterval{-.1}{1.3}{1.75};
\end{scope}
\oppositeinterval{.32}{.8}{2.85};
\oppositeinterval{1.65}{.33}{.2};

\begin{scope}[yscale=-1]
\begin{scope}
\clip
  (-.6,0) rectangle (5,2);
\draw[
  gray,
  line width=28,
  draw opacity=.4
]
  (1.75,2.6) circle (1.8);
\end{scope}
\closedinterval{-.4}{1.75}{1.2};
\closedinterval{-.1}{1.3}{1.75};
\begin{scope}[
  xshift=3.5cm,
  xscale=-1
]
\openinterval{-.4}{1.75}{1.2};
\openinterval{-.1}{1.3}{1.75};
\end{scope}
\oppositeinterval{.32}{.8}{2.85};
\oppositeinterval{1.65}{.33}{.2};
\end{scope}
  
\end{tikzpicture}
\hspace{-1.7cm}
}
\\
\hline
&
\\[-15pt]
\adjustbox{raise=-.8cm}{
\begin{tikzpicture}
  \draw[
    line width=1.2,
    -Latex
  ]
    (0:1) arc (0:182:1);
  \draw[
    line width=1.2,
    -Latex
  ]
    (180:1) arc (180:362:1);
\end{tikzpicture}
}
\;\;
\begin{tikzpicture}[decoration=snake]
\draw[decorate,->]
  (0,0) -- (+.55,0);
\draw[decorate,->]
  (0,0) -- (-.55,0);
\node[scale=1.2] at (1.8,0) {$\varnothing$};
\end{tikzpicture}
\hspace{.7cm}
&
\adjustbox{
  scale=.35,
  raise=-1.2cm
}{
\begin{tikzpicture}
  \halfcircle{0}{0}
  \begin{scope}[yscale=-1]
  \halfcircle{0}{0}
  \end{scope}
\end{tikzpicture}
}
\;\;
\begin{tikzpicture}[decoration=snake]
\draw[decorate,->]
  (0,0) -- (+.55,0);
\draw[decorate,->]
  (0,0) -- (-.55,0);
\node[scale=1.2] at (1.8,0) {$\varnothing$};
\end{tikzpicture}
\hspace{.7cm}
\end{tabular}
}
$}
\end{equation}
These relations (which go beyond those from Fig. \ref{FigureRibbonCategory} defining basic link diagram equivalence) happen to be known, respectively, as the {\it birth}/{\it death move} and the {\it fusion moves} (\cite[\S 6.3]{Khovanov00}\cite[Fig. 15]{Jacobsson04}, cf. \cite[Fig. 12]{Lobb24}) or {\it oriented saddle point} moves (e.g. \cite[Fig. 16]{Kauffman15}) 
generating (on top of usual link diagram equivalence) the relation of {\it link cobordism}
\footnote{
  Beware that early authors (e.g. \cite{Hosokawa67}\cite{CS80}) say ``link cobordism'' for what is now called ``link concordance'', with cylindrical cobordisms only. In this case, the corresponding equivalence classes of links are non-trivial.
  The modern use of ``link cobordism'' for actual cobordisms considered here seems to originate with \cite[\S 6.3]{Khovanov00}, cf. \cite[Fig 12]{Lobb24}. With this notion, all (framed) links are equivalent to (framed) unknots (Lem. \ref{FramedLinksCobordantToFramedUnknots} below), and hence the broader interest in general link cobordism is instead in characterizing the cobordisms themselves, notably through their associated homomorphism between Khovanov homologies \cite{Jacobsson04}.
}.
\end{example}

\begin{proposition}[\bf Configuration loop classes as link invariants]
\label{ChargedStringLoopsAsFramedLinksOnEquivalenceClasses}
  The map \eqref{MappingFramedLinksToStringLoops} descends to  equivalence classes, here sending framed oriented links {\rm (instead of their representing diagrams)} to elements in the fundamental group of the interval configuration space: 
  \begin{equation}
    \label{FramedLinksToPiOfStrings}
    \begin{tikzcd}
      \mathrm{FrmdOrntdLnk}
      \ar[
        r,
        ->>
      ]
      &
      \pi_1\big(
        \mathrm{Conf}^I_0
        (\mathbb{R}^2)
      \big)
      \,.
    \end{tikzcd}
  \end{equation}
\end{proposition}

\hspace{-1cm}
\def\tabcolsep{10pt}
\begin{tabular}{p{7.5cm}l}
  {\it Proof.}
  By Shum's theorem (Prop. \ref{ShumTheorem}) it is sufficient to see that the moves in Fig. \ref{FigureRibbonCategory} are respected.
  This is clear for the sliding moves and hence for the 1st and 2nd Reidemeister moves; while  
  the zig-zag move is respected by \eqref{StringyZigZagMove}.
  
  What remains to be shown is that also the 1st Reidemeister move is respected.

  For this, it is sufficient to show that the extra moves \eqref{SaddleAndVacuumMoves} imply the 1st Reidemeister move. 
  
  How this is the case is shown on the right. 

  \medskip

  That the map \eqref{FramedLinksToPiOfStrings} thus established is furthermore surjective is implied by the following analysis, culminating in  Thm. \ref{ChargedOpenStringLoopsClassifiedByCrossingNumber} below.
  &
\def\tabcolsep{10pt}
\def\arraystretch{.9}
\adjustbox{
  margin=3pt,
  fbox,
  raise=-3.7cm, 
  scale=0.68
}{
\hspace{-.4cm}
\begin{tabular}{c|c}
\begin{tikzpicture}[
  rotate=2
]

\draw[
  line width=1.2,
  -Latex
]
  (-1.3,-1) --
  (0,-1)
  .. controls 
     (.5,-1) and (1,-.5) ..
  (1,0)
  .. controls 
     (1,.5) and (.5,1) ..
  (0,1)  
  .. controls
    (-.5,1) and (-1,.5)
  .. (-1,0)
  .. controls
    (-1,-.4) and (-.5, -.8) ..
  (0,-.8);
\draw[
  -Latex,
  line width=1.2,
]
 (2.5,-.8)
 .. controls 
   (3,-.8) and (3.5,-.4) ..
 (3.5,0)
 .. controls
   (3.5,.5) and (3,1) ..
 (2.5,1)
  .. controls
    (2,1) and (1.5,.5) ..
  (1.5,0)
  .. controls
    (1.5,-.5) and (2,-1) ..
  (2.5,-1)
  --
  (3.8,-1);
\draw[
  line width=4.5,
  white
]
 (0,-.8) -- (2.5,-.8);
\draw[
  line width=1.2,
]
 (0,-.8) -- (2.5,-.8);

\node 
  at (.4,-.64) {
    \scalebox{.7}{$-$}
  };
\node 
  at (2.1,-.62) {
    \scalebox{.7}{$+$}
  };

\end{tikzpicture}
&
\begin{tikzpicture}[
  rotate=2
]

\draw[
  line width=1.2,
]
 (0,-.8) -- (2.5,-.8);

\draw[
  line width=4.5,
  white
]
  (-1.3,-1) --
  (0,-1)
  .. controls 
     (.5,-1) and (1,-.5) ..
  (1,0)
  .. controls 
     (1,.5) and (.5,1) ..
  (0,1)  
  .. controls
    (-.5,1) and (-1,.5)
  .. (-1,0)
  .. controls
    (-1,-.4) and (-.5, -.8) ..
  (0,-.8);

\draw[
  line width=1.2,
  -Latex
]
  (-1.3,-1) --
  (0,-1)
  .. controls 
     (.5,-1) and (1,-.5) ..
  (1,0)
  .. controls 
     (1,.5) and (.5,1) ..
  (0,1)  
  .. controls
    (-.5,1) and (-1,.5)
  .. (-1,0)
  .. controls
    (-1,-.4) and (-.5, -.8) ..
  (0,-.8);

\draw[
  line width=4.5,
  white
]
 (2.5,-.8)
 .. controls 
   (3,-.8) and (3.5,-.4) ..
 (3.5,0)
 .. controls
   (3.5,.5) and (3,1) ..
 (2.5,1)
  .. controls
    (2,1) and (1.5,.5) ..
  (1.5,0)
  .. controls
    (1.5,-.5) and (2,-1) ..
  (2.5,-1)
  --
  (3.5,-1);
\draw[
  -Latex,
  line width=1.2,
]
 (2.5,-.8)
 .. controls 
   (3,-.8) and (3.5,-.4) ..
 (3.5,0)
 .. controls
   (3.5,.5) and (3,1) ..
 (2.5,1)
  .. controls
    (2,1) and (1.5,.5) ..
  (1.5,0)
  .. controls
    (1.5,-.5) and (2,-1) ..
  (2.5,-1)
  --
  (3.8,-1);

\node 
  at (.4,-.64) {
    \scalebox{.7}{$-$}
  };
\node 
  at (2.1,-.62) {
    \scalebox{.7}{$+$}
  };

\end{tikzpicture}
\\
\hspace{.5cm}
 \begin{tikzpicture}[
   rotate=90,
   decoration=snake
 ]
   \draw[decorate, ->]
    (-.01,0) -- (0.54,0);
   \draw[decorate, ->]
    (0.01,0) -- (-0.54,0);
 \end{tikzpicture}
&
\hspace{.5cm}
 \begin{tikzpicture}[
   rotate=90,
   decoration=snake
  ]
   \draw[decorate, ->]
    (-.01,0) -- (0.54,0);
   \draw[decorate, ->]
    (0.01,0) -- (-0.54,0);
 \end{tikzpicture}
\\
\begin{tikzpicture}[
  rotate=2
]

\draw[
  line width=1.2,
  -Latex
]
  (-1.3,-1) --
  (0,-1)
  .. controls 
     (.5,-1) and (1,-.5) ..
  (1,0)
  .. controls 
     (1,.5) and (.5,1) ..
  (0,1)  
  .. controls
    (-.5,1) and (-1,.5)
  .. (-1,0)
  .. controls
    (-1,-.4) and (-.5, -.8) ..
  (0,-.8);
\draw[
  -Latex,
  line width=1.2,
]
 (2.5,-.8)
 .. controls 
   (3,-.8) and (3.5,-.4) ..
 (3.5,0)
 .. controls
   (3.5,.5) and (3,1) ..
 (2.5,1)
  .. controls
    (2,1) and (1.5,.5) ..
  (1.5,0)
  .. controls
    (1.5,-.5) and (2,-1) ..
  (2.5,-1)
  --
  (3.8,-1);
\draw[
  line width=4.5,
  white
]
 (0,-.8) -- (2.5,-.8);
\draw[
  line width=1.2,
]
 (0,-.8) -- (2.5,-.8);

\draw[white,fill=white]
  (0.3,-.5) rectangle (2.3,+.5);

\draw[line width=1.2]
  (.828,.52)
  .. controls
    (.828+.3,.1) and 
  (1.673-.3,.1) ..
  (1.673,.52);

\begin{scope}[
  yscale=-1
]
\draw[line width=1.2]
  (.828,.52)
  .. controls
    (.828+.3,.1) and 
  (1.673-.3,.1) ..
  (1.673,.52);
\end{scope}

\end{tikzpicture}
&
\begin{tikzpicture}[
  rotate=2
]

\draw[
  line width=1.2,
]
 (0,-.8) -- (2.5,-.8);

\draw[
  line width=4.5,
  white
]
  (-1.3,-1) --
  (0,-1)
  .. controls 
     (.5,-1) and (1,-.5) ..
  (1,0)
  .. controls 
     (1,.5) and (.5,1) ..
  (0,1)  
  .. controls
    (-.5,1) and (-1,.5)
  .. (-1,0)
  .. controls
    (-1,-.4) and (-.5, -.8) ..
  (0,-.8);
\draw[
  line width=1.2,
  -Latex
]
  (-1.3,-1) --
  (0,-1)
  .. controls 
     (.5,-1) and (1,-.5) ..
  (1,0)
  .. controls 
     (1,.5) and (.5,1) ..
  (0,1)  
  .. controls
    (-.5,1) and (-1,.5)
  .. (-1,0)
  .. controls
    (-1,-.4) and (-.5, -.8) ..
  (0,-.8);

\draw[
  line width=4.5,
  white
]
 (2.5,-.8)
 .. controls 
   (3,-.8) and (3.5,-.4) ..
 (3.5,0)
 .. controls
   (3.5,.5) and (3,1) ..
 (2.5,1)
  .. controls
    (2,1) and (1.5,.5) ..
  (1.5,0)
  .. controls
    (1.5,-.5) and (2,-1) ..
  (2.5,-1)
  --
  (3.5,-1);
\draw[
  -Latex,
  line width=1.2,
]
 (2.5,-.8)
 .. controls 
   (3,-.8) and (3.5,-.4) ..
 (3.5,0)
 .. controls
   (3.5,.5) and (3,1) ..
 (2.5,1)
  .. controls
    (2,1) and (1.5,.5) ..
  (1.5,0)
  .. controls
    (1.5,-.5) and (2,-1) ..
  (2.5,-1)
  --
  (3.8,-1);

\draw[white,fill=white]
  (0.3,-.5) rectangle (2.3,+.5);

\draw[line width=1.2]
  (.828,.52)
  .. controls
    (.828+.3,.1) and 
  (1.673-.3,.1) ..
  (1.673,.52);

\begin{scope}[
  yscale=-1
]
\draw[line width=1.2]
  (.828,.52)
  .. controls
    (.828+.3,.1) and 
  (1.673-.3,.1) ..
  (1.673,.52);
\end{scope}

\end{tikzpicture}
\\
\hspace{.5cm}
 \begin{tikzpicture}[
   rotate=90,
   decoration=snake
 ]
   \draw[decorate, ->]
    (-.01,0) -- (0.54,0);
   \draw[decorate, ->]
    (0.01,0) -- (-0.54,0);
 \end{tikzpicture}
&
\hspace{.5cm}
 \begin{tikzpicture}[
   rotate=90,
   decoration=snake
 ]
   \draw[decorate, ->]
    (-.01,0) -- (0.54,0);
   \draw[decorate, ->]
    (0.01,0) -- (-0.54,0);
 \end{tikzpicture}
\\
\;
\begin{tikzpicture}[
  rotate=2
]
  \draw[
    line width=1.2,
    -Latex
  ]
  (-1.9,0) -- (3.2,0);
\end{tikzpicture}
&
\;
\begin{tikzpicture}[
  rotate=2
]
  \draw[
    line width=1.2,
    -Latex
  ]
  (-1.9,0) -- (3.2,0);
\end{tikzpicture}
\end{tabular}
\hspace{-.5cm}
}
\end{tabular}

{$\;$ \hspace{16.2cm} $\Box$}

\begin{example}[\bf Group of stringy images of framed unknots]
  \label{GroupOfStringyImagesOfFRamedUnknots}
  The images of the framed unknots under \eqref{FramedLinksToPiOfStrings} constitute an integer subgroup $\mathbb{Z} \,\subset\, \mathbb{Z} \simeq \pi_1\big(\mathrm{Conf}^I_0(\mathbb{R}^2)\big)$ (cf. Prop. \ref{TheFundamenralGroupOfStringConfigurationSpace}) whose group operation corresponds to the addition of total crossing/framing number (Def. \ref{LinkingNumber}).
  For instance, the following is the move corresponding to the equation $1 + 1 = 2$ in this subgroup:
\begin{center}
\adjustbox{
  scale=.35,
  raise=-6.4cm
}{
\begin{tikzpicture}

\begin{scope}[scale=-1]
\halfcircleover{0}{0}
\end{scope}

\draw[
  line width=43,
  white
]
  (-3,-.01) -- (-3,2.8);
\draw[
  gray,
  draw opacity=.4,
  line width=30
]
  (-3,-2.8) -- (-3,2.8);

\draw[
  gray,
  draw opacity=.4,
  line width=30
]
  (-9,-2.8) -- 
  (-9,2.8);

\begin{scope}[xscale=-1]
\halfcircleover{0}{0}
\end{scope}

\begin{scope}[
  shift={(-6,2.8)},
  yscale=-1
]
\halfcircleover{0}{0}
\end{scope}

\begin{scope}[
  shift={(-6,-2.8)}
]
\halfcircleover{0}{0}
\end{scope}

\foreach \n in {0,...,6} {
  \openinterval{-9.5}{-2.4+\n*.8}{1}
}

\foreach \n in {0,...,2} {
  \closedinterval{-3.5}{.5+\n*.9}{1}
}

\begin{scope}[shift={(0,-13.5)}]

\begin{scope}[scale=-1]
\halfcircleover{0}{0}
\end{scope}

\draw[
  line width=43,
  white
]
  (-3,-.01) -- (-3,2.8);
\draw[
  gray,
  draw opacity=.4,
  line width=30
]
  (-3,-2.8) -- (-3,2.8);

\draw[
  gray,
  draw opacity=.4,
  line width=30
]
  (-9,-2.8) -- 
  (-9,2.8);

\begin{scope}[xscale=-1]
\halfcircleover{0}{0}
\end{scope}

\begin{scope}[
  shift={(-6,2.8)},
  yscale=-1
]
\halfcircleover{0}{0}
\end{scope}

\begin{scope}[
  shift={(-6,-2.8)}
]
\halfcircleover{0}{0}
\end{scope}

\foreach \n in {0,...,6} {
  \openinterval{-9.5}{-2.4+\n*.8}{1}
}

\foreach \n in {0,...,2} {
  \closedinterval{-3.5}{.5+\n*.9}{1}
}

\end{scope}
  
\end{tikzpicture}
}
\hspace{.3cm}
\adjustbox{
  raise=-1.6cm
}{
\begin{tikzpicture}[decoration=snake]
  \draw[decorate,->] (0,0) -- (+.55,0);
  \draw[decorate,->] (0,0) -- (-.55,0);
\end{tikzpicture}
}
\hspace{+.3cm}
\adjustbox{
  scale=.35,
  raise=-5.2cm
}{
\begin{tikzpicture}
\begin{scope}[scale=-1]
\halfcircleover{0}{0}
\end{scope}

\draw[
  line width=43,
  white
]
  (-3,-.01) -- (-3,2.8);
\draw[
  gray,
  draw opacity=.4,
  line width=30
]
  (-3,-5) -- (-3,2.8);

\draw[
  gray,
  draw opacity=.4,
  line width=30
]
  (-9,-12.1) -- 
  (-9,2.8);

\begin{scope}[xscale=-1]
\halfcircleover{0}{0}
\end{scope}

\begin{scope}[
  shift={(-6,2.8)},
  yscale=-1
]
\halfcircleover{0}{0}
\end{scope}

\begin{scope}[
  shift={(0,-8)},
  scale=-1
]
\halfcircleover{0}{0}
\end{scope}

\draw[
  line width=47,
  white
]
  (-3,-5) -- (-3,-8);
\draw[
  gray,
  draw opacity=.4,
  line width=30
]
  (-3,-5) -- (-3,-12);

\begin{scope}[
  shift={(0,-8)},
  xscale=-1
]
\halfcircleover{0}{0}
\end{scope}

\begin{scope}[
  shift={(-6,-12)},
]
\halfcircleover{0}{0}
\end{scope}

\foreach \n in {-11,...,5} {
  \openinterval{-9.5}{-2.01+\n*.86}{1}
}

\foreach \n in {0,...,2} {
  \closedinterval{-3.5}{.5+\n*.9}{1}
}

\foreach \n in {0,...,5} {
  \closedinterval{-3.5}{-7.5+\n*.9}{1}
}

\foreach \n in {0,...,0} {
  \closedinterval{-3.5}{-11.4+\n*.9}{1}
}

\end{tikzpicture}
}
\end{center}

\end{example}

In fact, this subgroup inclusion is surjective \eqref{UnitFramedUnknotIsTheGenerator}, hence exhausts the full fundamental group, by the following further analysis.

\begin{lemma}[\bf Framed links are cobordant to framed unknots]
  \label{FramedLinksCobordantToFramedUnknots}
  Every framed oriented link is related by the cobordism moves \eqref{SaddleAndVacuumMoves} to a framed oriented unknot.
\end{lemma}
\begin{proof}
  Using the zig-zag move 
  \eqref{StringyZigZagMove}
  and the 
  saddle move \eqref{SaddleAndVacuumMoves}, every crossing may be turned into an avoided crossing of a straight edge with a twisted edge, like this:
  \begin{equation}
  \label{AvoidCrossing}
  \scalebox{0.65}{$
\adjustbox
{
  scale=.8,
  raise=-1.4cm
}{
\begin{tikzpicture}

\draw[
  line width=1.2,
  -Latex
]
  (-1.6,-1.6) -- (+1.6,+1.6);

\draw[
  line width=15,
  white
]
  (+1.6,-1.6) -- (-1.6,+1.6);
\draw[
  line width=1.2,
  -Latex
]
  (+1.6,-1.6) -- (-1.6,+1.6);
  
\end{tikzpicture}}
\;\;\;\;\;\;
\begin{tikzpicture}[decoration=snake]
  \draw[decorate,->]
    (0,0) -- (+.55,0);
  \draw[decorate,->]
    (0,0) -- (-.55,0);
\end{tikzpicture}
\;\;\;\;\;\;
\adjustbox{
  scale=.8,
  raise=-1.3cm
}
{
\begin{tikzpicture}

\draw[
  line width=1.2,
  -Latex
]
  (-1.7,-1.7) -- (0,0);

\draw[
  -Latex,
  line width=1.2,
]
  (+1.7,-1.7) -- (-1.7,1.7);

\draw[
  white,fill=white
]
  (-1.5,-1.5) rectangle (1.5,1.5);

\draw[line width=1.2]
  (-1.5,-1.5)
   .. controls
     (-.8,-.8) and (-.8,+.8) ..
  (-1.5,+1.5);

\begin{scope}[shift={(-.42,0)}]
\draw[
  line width=1.25
]
  (45:.3) arc (45:360-45:.3);
\end{scope}

\draw[line width=1.2]
  (+1.6,-1.6) -- (-.2,.2);
\draw[line width=1.2, -Latex]
  (+.2,+.2) -- (1.6,1.6);

\draw[white,fill=white]
  (-1.1,-.2) rectangle (-.5,+.2);
\clip
  (-1.1,-.2) rectangle (-.5,+.2);

\draw[line width=1.2]
  (-.82,.24) circle (.17);

\draw[line width=1.2]
  (-.81,-.24) circle (.18);
\end{tikzpicture}
}
\;\;\;\;\;\;
\begin{tikzpicture}[decoration=snake]
  \draw[decorate,->]
    (0,0) -- (+.55,0);
  \draw[decorate,->]
    (0,0) -- (-.55,0);
\end{tikzpicture}
\;\;\;\;\;\;
\adjustbox{
  scale=.8,
  raise=-1.3cm
}
{
\begin{tikzpicture}

\draw[
  line width=1.2,
  -Latex
]
  (-1.7,-1.7) -- (0,0);

\draw[
  line width=1.2,
  -Latex, 
]
  (+1.7,-1.7) -- (-1.7,+1.7);

\draw[
  white,fill=white
]
  (-1.5,-1.5) rectangle (1.5,1.5);

\draw[line width=1.2]
  (-1.5,-1.5)
   .. controls
     (-.8,-.8) and (-.8,+.8) ..
  (-1.5,+1.5);

\begin{scope}[shift={(-.42,0)}]
\draw[
  line width=1.25
]
  (45:.3) arc (45:360-45:.3);
\end{scope}

\draw[line width=1.2]
  (+1.7,-1.7) -- (-.2,.2);
\draw[line width=1.2, -Latex]
  (+.2,+.2) -- (1.7,1.7);
 
\end{tikzpicture}
}
$}
\end{equation}
Applying such a move to all crossings of a given link diagram yields a framed unlink. Then forming the connected sum of its connected components (as in Ex. \ref{GroupOfStringyImagesOfFRamedUnknots}) yields a framed unknot.
\end{proof}

\begin{example}[\bf Framed links turned into framed unknots]
\label{FramedLinksTurnedIntoFramedUnknots}
The Hopf link becomes the unknot with framing $\pm 2$ by applying the saddle move either on the right or in the middle, depending on the given orientations:
\begin{equation}
\adjustbox{scale=0.8}{
\adjustbox{
  scale=.5,
  raise=-1.9cm
}{
\begin{tikzpicture}

\begin{scope}[yscale=-1]
  \halfcircle{0,0};
\end{scope}
\begin{scope}[
  shift={(3.4,.2)},
  xscale=-1
]
\halfcircle{0}{0};
\end{scope}
\begin{scope}[shift={(0,0)}]
\halfcircleover{0}{0};
\end{scope}
\begin{scope}[
  yscale=-1,
  shift={(3.4,-.2)},
  xscale=-1
]
\halfcircleover{0}{0};
\end{scope}
  
\end{tikzpicture}
} 
\begin{tikzpicture}[decoration=snake]
  \draw[decorate,->]
    (0,0) -- (+.55, 0);
  \draw[decorate,->]
    (0,0) -- (-.55, 0);
\end{tikzpicture}
\adjustbox{
  scale=.5,
  raise=-1.8cm
}{
\begin{tikzpicture}

\begin{scope}[yscale=-1]
  \halfcircle{0,0};
\end{scope}
\begin{scope}[
  shift={(3.4,.2)},
  xscale=-1
]
\halfcircle{0}{0};
\end{scope}
\begin{scope}[shift={(0,0)}]
\halfcircleover{0}{0};
\end{scope}
\begin{scope}[
  yscale=-1,
  shift={(3.4,-.2)},
  xscale=-1
]
\halfcircleover{0}{0};
\end{scope}

\draw[
  white, fill=white
] 
  (1.8,-1.35) rectangle (7.5,1.55);
\clip  
  (1.8,-1.35) rectangle (7.5,1.55);

\draw[
  gray,
  line width=24,
  draw opacity=.4
]
 (4.33,3.41) circle (2.6);

\draw[
  gray,
  line width=22,
  draw opacity=.4
]
 (4.33,-3.3) circle (2.6);

\closedinterval{2.25}{1.25}{1.7};
\openinterval{4.15}{1.25}{2.25};
\oppositeinterval{2.86}{.75}{2.7};
\oppositeinterval{4.2}{.38}{.2};

\begin{scope}[
  shift={(0,.1)},
  yscale=-1
]
  \closedinterval{2.35}{1.2}{1.85};
  \openinterval{4.4}{1.2}{1.95};
  \oppositeinterval{2.98}{.75}{2.7};
  \oppositeinterval{4.2}{.38}{.2};
\end{scope}

\end{tikzpicture}
}
}
\end{equation}

\begin{equation}
\adjustbox{scale=0.8}{
\adjustbox{
  scale=.5,
  raise=-1.8cm
}{
\begin{tikzpicture}

\begin{scope}[yscale=-1]
  \halfcircle{0,0};
\end{scope}
\begin{scope}[
  shift={(3.4,.2)}
]
\halfcircle{0}{0};
\end{scope}
\begin{scope}[shift={(0,0)}]
\halfcircleover{0}{0};
\end{scope}
\begin{scope}[
  yscale=-1,
  shift={(3.4,-.2)}
]
\halfcircleover{0}{0};
\end{scope}

\end{tikzpicture}
}
\begin{tikzpicture}[decoration=snake]
  \draw[decorate,->] (0,0) -- (.55,0);
  \draw[decorate,->] (0,0) -- (-.55,0);
\end{tikzpicture}
\adjustbox{
  scale=.5,
  raise=-1.8cm
}{
\begin{tikzpicture}

\begin{scope}[yscale=-1]
  \halfcircle{0,0};
\end{scope}
\begin{scope}[
  shift={(3.4,.2)}
]
\halfcircle{0}{0};
\end{scope}
\begin{scope}[shift={(0,0)}]
\halfcircleover{0}{0};
\end{scope}
\begin{scope}[
  yscale=-1,
  shift={(3.4,-.2)}
]
\halfcircleover{0}{0};
\end{scope}

\draw[white,fill=white] 
  (-.6,-1.4) rectangle (4,1.56);

\clip
  (-.6,-1.4) rectangle (4,1.56);

\draw[
  gray,
  line width=35,
  draw opacity=.4
]
  (1.75,-1.35) circle (.84);

\interval{.4}{-.95}{2.75};
\interval{.75}{-.3}{2};

\begin{scope}[
  shift={(-.15,.25)},
  yscale=-1
]
\draw[
  gray,
  line width=35,
  draw opacity=.4
]
  (1.75,-1.35) circle (.84);
\interval{.4}{-.95}{2.75};
\interval{.75}{-.3}{2};

\end{scope}

\end{tikzpicture}
}
}
\end{equation}

If we understand the stringy moves applied already to the corresponding framed link diagrams, 
then we may draw the above example more succinctly as
\begin{equation}
\label{HopfLinktoUnknot}
\adjustbox{scale=0.65}{
\adjustbox{raise=-1cm}{
\begin{tikzpicture}[
  scale=1
]

\begin{scope}[
  shift={(1,0)}
]
\draw[line width=2, -Latex]
  (0:1) arc (0:180:1);
\end{scope}

\draw[line width=7,white]
  (0:1) arc (0:180:1);
\draw[line width=2, -Latex]
  (0:1) arc (0:180:1);

\draw[line width=2, -Latex]
  (180:1) arc (180:360:1);

\begin{scope}[shift={(1,0)}]
\draw[line width=7, white]
  (180:1) arc (180:360:1);
\draw[line width=2, -Latex]
  (180:1) arc (180:360:1);
\end{scope}

\node[gray]
  at (.5,.64) {\color{red} 
    \scalebox{.9}{$-$}
  };
\node[gray]
  at (.5,-.64) {\color{red}
    \scalebox{.9}{$-$}
  };
  
\end{tikzpicture}}
\hspace{.2cm}
 \begin{tikzpicture}[decoration=snake]
   \draw[decorate, ->]
    (-.01,0) -- (0.54,0);
   \draw[decorate, ->]
    (0.01,0) -- (-0.54,0);
 \end{tikzpicture}
\hspace{.0cm}
\adjustbox{raise=-1cm}{
\begin{tikzpicture}

\begin{scope}[shift={(1,0)}]
\draw[line width=2, -Latex]
  (0:1) arc (0:180:1);
\end{scope}

\draw[line width=7,white]
  (0:1) arc (0:180:1);
\draw[line width=2, -Latex]
  (0:1) arc (0:180:1);

\draw[line width=2, -Latex]
  (180:1) arc (180:360:1);

\begin{scope}[shift={(1,0)}]
\draw[line width=7, white]
  (180:1) arc (180:360:1);
\draw[line width=2, -Latex]
  (180:1) arc (180:360:1);
\end{scope}

\draw[white,fill=white]
 (-.3,-.55) rectangle 
 (1.2,+.55);

\draw[line width=2]
  (.067,.55) 
  .. controls
    (.05-.2, .1) and 
    (.86+.2, .1) ..
  (+.84,.55);

\begin{scope}[
  shift={(.1,0)},
  yscale=-1
]
\draw[line width=2]
  (.067,.55) 
  .. controls
    (.05-.2, .1) and 
    (.86+.2, .1) ..
  (+.84,.55);
\end{scope}
 
\end{tikzpicture}}
}
\end{equation}
Further examples in this notation are the following:
The trefoil knot becomes
\begin{equation}
\label{TrefoilLinkToUnknot}
\adjustbox{scale=0.65}{
\adjustbox{
  raise=-2.6cm,
  scale=.8
}{
\begin{tikzpicture}
\foreach \n in {0,1,2} {
\begin{scope}[
  rotate=\n*120-4
]
\draw[
  line width=2,
  -Latex
]
 (0,-1)
   .. controls
   (-1,.2) and (-2,2) ..
 (0,2)
   .. controls
   (1,2) and (1,1) ..
  (.9,.7);
\end{scope}

\node[darkgreen]
  at (\n*120+31:.7) {
    \scalebox{1}{$+$}
  };

};

\end{tikzpicture}
}
\hspace{-.5cm}
 \begin{tikzpicture}[decoration=snake]
   \draw[decorate, ->]
    (-.01,0) -- (0.54,0);
   \draw[decorate, ->]
    (0.01,0) -- (-0.54,0);
 \end{tikzpicture}
\adjustbox{
  raise=-2.6cm,
  scale=.8
}{
\begin{tikzpicture}
\foreach \n in {0,1,2} {
\begin{scope}[
  rotate=\n*120-4
]
\draw[
  line width=2,
  -Latex
]
 (0,-1)
   .. controls
   (-1,.2) and (-2,2) ..
 (0,2)
   .. controls
   (1,2) and (1,1) ..
  (.9,.7);
\end{scope}
};

\draw[white,fill=white]
  (-.9,-.7)
  rectangle
  (.9,.2);

\draw[line width=2]
  (-.79,.21) 
    .. controls
    (-.7,-.25) and (+.7,-.25) ..
  (.79,.21);

\draw[line width=2]
  (-.27,-.7) 
  .. controls
    (-.4,-.2) and 
    (+.4,-.2) ..
  (+.27,-.7);

\end{tikzpicture}
}
\hspace{-.5cm}
 \begin{tikzpicture}[decoration=snake]
   \draw[decorate, ->]
    (-.01,0) -- (0.54,0);
   \draw[decorate, ->]
    (0.01,0) -- (-0.54,0);
 \end{tikzpicture}
\adjustbox{
  raise=-2.6cm,
  scale=.8
}{
\begin{tikzpicture}
\foreach \n in {0,1,2} {
\begin{scope}[
  rotate=\n*120-4
]
\draw[
  line width=2,
  -Latex
]
 (0,-1)
   .. controls
   (-1,.2) and (-2,2) ..
 (0,2)
   .. controls
   (1,2) and (1,1) ..
  (.9,.7);
\end{scope}
};

\draw[white,fill=white]
  (-.9,-.7)
  rectangle
  (.9,.2);

\draw[line width=2]
  (-.79,.21) 
    .. controls
    (-.7,-.25) and (+.7,-.25) ..
  (.79,.21);

\draw[line width=2]
  (-.27,-.7) 
  .. controls
    (-.4,-.2) and 
    (+.4,-.2) ..
  (+.27,-.7);

\draw[white,fill=white]
 (-.5,.8) 
 rectangle (+.5,-.2);

\draw[line width=2]
  (-.5,.58)
  .. controls
    (-.0,.7) and (-.0,-.26)
    ..
  (-.5,-.06);

\draw[line width=2]
  (+.5,.58)
  .. controls
    (+.0,.7) and (+.0,-.26)
    ..
  (+.5,-.06);

\end{tikzpicture}
}
}
\end{equation}

\vspace{-.4cm}
\noindent
and the figure-eight knot becomes

\vspace{-.3cm}
\begin{equation}
\label{FigureEightNotToUnknot}
\adjustbox{scale=0.65}{
\adjustbox{
  raise=-1cm,
  scale=1.3
}{
\begin{tikzpicture}

\node at (0,0) {
  \includegraphics[width=2.2cm]{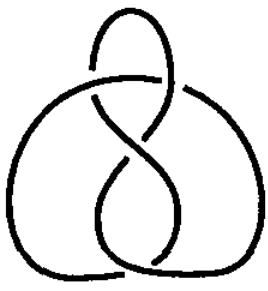}
};

\draw[line width=1.2pt,-Latex]
  (-.276,-.8) -- 
  (-.274,-.8-.02);
\draw[line width=1.2pt,-Latex]
  (.35, -.5) --
  (.35, -.47);
\draw[line width=1.2pt,-Latex]
  (-.09,1.09) --
  (-.09+.05, 1.13);
\draw[line width=1.2pt,-Latex]
  (-.18,.54) --
  (-.19-.01, .54);

\node[red]
  at (0,-1.2) {
    \scalebox{.7}{$-$}
  };
\node[red]
  at (.3,0) {
    \scalebox{.7}{$-$}
  };
\node[darkgreen]
  at (.5,.64) {
    \scalebox{.7}{$+$}
  };
\node[darkgreen]
  at (-.57,.64) {
    \scalebox{.7}{$+$}
  };

\end{tikzpicture}
}
\hspace{0cm}
 \begin{tikzpicture}[decoration=snake]
   \draw[decorate, ->]
    (-.01,0) -- (0.54,0);
   \draw[decorate, ->]
    (0.01,0) -- (-0.54,0);
 \end{tikzpicture}
\hspace{.0cm}
\adjustbox{
  scale=1.3,
  raise=-1.2cm
}{
\begin{tikzpicture}

\node at (0,0) {
  \includegraphics[width=2.2cm]{FigureEightKnot.png}
};

\draw[line width=1.2pt,-Latex]
  (-.09,1.09) --
  (-.09+.05, 1.13);
\draw[line width=1.2pt,-Latex]
  (-.18,.54) --
  (-.19-.01, .54);

\draw[white,fill=white]
  (-.4,-.8) rectangle (.5,-.3);

\draw[line width=1.5]
  (-.27,-.81)
    ..
    controls
    (-.32, -.6) and (.4,-.6) ..
  (.3,-.81);

\begin{scope}[
  shift={(0,-1.1)},
  yscale=-1
]
\draw[line width=1.5]
  (-.19,-.81)
    ..
    controls
    (-.3, -.6) and (.4,-.6) ..
  (.28,-.81);
\end{scope}

\end{tikzpicture}
}
\hspace{.2cm}
 \begin{tikzpicture}[decoration=snake]
   \draw[decorate, ->]
    (-.01,0) -- (0.54,0);
   \draw[decorate, ->]
    (0.01,0) -- (-0.54,0);
 \end{tikzpicture}
\hspace{.0cm}
\adjustbox{
  scale=1.3,
  raise=-1.2cm
}{
\begin{tikzpicture}

\node at (0,0) {
  \includegraphics[width=2.2cm]{FigureEightKnot.png}
};

\draw[white,fill=white]
  (-.4,-.8) rectangle (.5,-.3);

\draw[line width=1.5]
  (-.27,-.81)
    ..
    controls
    (-.32, -.6) and (.4,-.6) ..
  (.3,-.81);

\begin{scope}[
  shift={(0,-1.1)},
  yscale=-1
]
\draw[line width=1.5]
  (-.19,-.81)
    ..
    controls
    (-.3, -.6) and (.4,-.6) ..
  (.28,-.81);
\end{scope}

\draw[white,fill=white]
  (-.2,.4) rectangle (+.15,1.4);

\draw[line width=1.5]
 (-.22,1.01) 
 .. controls 
   (0,1.1) and (0,.55) ..
 (-.22,.54);

\begin{scope}[
  shift={(-.03,0)},
  xscale=-1
]
\draw[line width=1.5]
 (-.22,1) 
 .. controls 
   (0,1.1) and (0,.55) ..
 (-.22,.54);
\end{scope}

\draw[
  line width=1.2,
  -Latex
]
  (-1.01,-.6) --
  (-1.01,-.6-.01);

\end{tikzpicture}
}
}
\end{equation}
\end{example}

Using all this we finally obtain our main result in this section:

\begin{theorem}[\bf Interval configuration loops classified by crossing number]
\label{ChargedOpenStringLoopsClassifiedByCrossingNumber}
  The map \eqref{FramedLinksToPiOfStrings} from framed oriented links to the fundamental group of the configuration space of signed intervals in the plane is, under the latter's identification with the integers \eqref{FundamentalGroupOfStringConfigurationSpace}, given by sending a link $L$ to its total crossing number/writhe $\# L$ \eqref{TotalCrossingNumber}:
  \begin{equation}
    \label{LinksToWrithe}
    \begin{tikzcd}[
      sep=5pt,
    ]
      \mathrm{FrmdOrntdLnk}
      \ar[
        r,
        ->>
      ]
      &[+8pt]
      \pi_0\big(
        \Omega\,
        \mathrm{Conf}^I_0(\mathbb{R}^2)
      \big)
      \ar[
        r,
        phantom,
        "{ \simeq }"
      ]
      &
      \mathbb{Z}
      \\[-4pt]
      L
      \ar[
        rr,
        |->,
        shorten=10pt
      ]
      &&
      \# L
      \mathrlap{\,.}
    \end{tikzcd}
  \end{equation}
\end{theorem}
\begin{proof}
  By Lem. \ref{FramedLinksCobordantToFramedUnknots},
  the image of $L$ is equivalently a framed unknot via the saddle moves \eqref{SaddleAndVacuumMoves}.
  Since all framed unknots are multiples of the unit-framed unknot, by Ex. \ref{GroupOfStringyImagesOfFRamedUnknots}, this exhibits 
  the unit framed unknot as the generator
  \begin{equation}
  \label{UnitFramedUnknotIsTheGenerator}
  \scalebox{0.7}{$
\left[
\adjustbox{
  scale=.5,
  raise=-3.3cm
}
{
\begin{tikzpicture}

\begin{scope}[scale=-1]
\halfcircleover{0}{0}
\end{scope}

\draw[
  line width=43,
  white
]
  (-3,-.01) -- (-3,2.8);
\draw[
  gray,
  draw opacity=.4,
  line width=30
]
  (-3,-2.8) -- (-3,2.8);

\draw[
  gray,
  draw opacity=.4,
  line width=30
]
  (-9,-2.8) -- 
  (-9,2.8);

\begin{scope}[xscale=-1]
\halfcircleover{0}{0}
\end{scope}

\begin{scope}[
  shift={(-6,2.8)},
  yscale=-1
]
\halfcircleover{0}{0}
\end{scope}

\begin{scope}[
  shift={(-6,-2.8)}
]
\halfcircleover{0}{0}
\end{scope}

\foreach \n in {0,...,6} {
  \openinterval{-9.5}{-2.4+\n*.8}{1}
}

\foreach \n in {0,...,2} {
  \closedinterval{-3.5}{.5+\n*.9}{1}
}
  
\end{tikzpicture}
}
\right]
$}
\;\;
=
\;\;
1
\quad \in\;
\mathbb{Z}
\;\simeq\;
\pi_1\big(
  \mathrm{Conf}^I(\mathbb{R}^2)
\big)
\,.
\end{equation}
(which hence corresponds to the Hopf fibration under the identification of Prop. \ref{TheFundamenralGroupOfStringConfigurationSpace}).

Moreover, since the saddle move \eqref{AvoidCrossing} used in Lem. \ref{FramedLinksCobordantToFramedUnknots}
  manifestly preserves writhe $\#$ \eqref{TotalCrossingNumber}, the writhe of the resulting unknot (being its framing number) is that of $L$ (cf. Ex. \ref{FramedLinksTurnedIntoFramedUnknots}), and hence
  it represents the $\#(L)$-fold multiple of the generator \eqref{UnitFramedUnknotIsTheGenerator}.
\end{proof}

\begin{remark}[\bf Relation to discriminants of configurations]
\label{RelationToDiscriminant}
A statement similar to Thm. \ref{ChargedOpenStringLoopsClassifiedByCrossingNumber} appears as \cite[Thm. 5.2]{MoravaRolfsen23} following \cite[Lem. 3.6]{GorinLin69}
(we are grateful to Jack Morava for pointing this out): There is a locally trivial fibration
$$
  \Delta_k
  \,:\,
  \begin{tikzcd}
  \mathrm{Conf}_k(\mathbb{C})
  \ar[r]
  &
  \mathbb{C}^\times
  \,,
  \end{tikzcd}
$$
called the {\it discriminant} and given by the product of $\mathbb{C}$-coordinate differences of distinct points in a  configuration, which is such that under passage to $\pi_1$ it becomes the writhe function \eqref{TotalCrossingNumber}
$$
  \mathllap{
  \#_k
  \,=\,
  }
  \pi_1(\Delta_k)
  \,:\,
  \begin{tikzcd}
  \mathrm{Br}_k
  \ar[r]
  &
  \mathbb{Z}
  \end{tikzcd}
$$
on plain braids (counting their signed number of crossings).
We may thus think of the discriminant as giving a complex-analytic formula for the ``restriction'' from framed links to plain braids of our assignment \eqref{LinksToWrithe} (in that every framed link can be obtained from adding framing to a plain braid and then closing it by connecting endpoints).
\end{remark}

\begin{remark}[\bf Loops based in the $n$-component]\label{LoopsBasedInTheNComponent}
$\,$
  
\noindent {\bf (i)}  Since the group completed configuration space $\mathbb{G}\mathrm{Conf}(\mathbb{R}^2)$ is, by construction, a topological group, it follows abstractly that all its connected components 
are (weak homotopy) equivalent, hence so are the connected components 
\eqref{ConnectedComponentOfConfigurationSpaceOfSignedIntervals} of the interval configuration space $\mathrm{Conf}^I(\mathbb{R}^2)$, by Prop. \ref{OkuyamaTheorem}, and hence so are also the loop spaces based on any of these connected components:
\begin{equation}
 \label{EquivalenceOfBasedLoopSpaces}
    \underset{n,n' \in \mathbb{Z}}{\forall}
    \;\;\;\;\;\;\;
    \Omega \,
    \mathrm{Conf}^I_n(\mathbb{R}^2) 
    \;\;
    \weakHomotopyEquivalence
    \;\;
    \Omega \,
    \mathrm{Conf}^I_{n'}
    (\mathbb{R}^2) 
    \,.
\end{equation}

\noindent {\bf (ii)}  
Concretely, we may now exhibit this equivalence in terms of the interpretation of loops in $\mathrm{Conf}^I_0(\mathbb{R}^2)$ as framed links that we have established. Or rather, this interpretation applies to the loops in the $0$-component, while loops in the $n$-component may be understood more generally as {\it braids} on $n$ strands interlinked with any number of links, as illustrated in Fig. \ref{FigureL}:

\medskip 
\begin{tabular}{p{4.1cm}l}
  \footnotesize  
  {\bf Figure \figurenumber:
  \label{FigureL}}
  A  loop 
  in $\mathrm{Conf}^I_(\mathbb{R}^2)$
  based in the $n = 3$ component.
&
\hspace{.4cm}
    \adjustbox{
      raise=-2.5cm, 
      scale=0.75
    }{
\begin{tikzpicture}

\begin{scope}[shift={(-2.36,0)}]
\draw[dashed, line width=2]
  (0,0) ellipse (.6 and 1.8);
\draw[white,fill=white]
  (.2,   1.44) rectangle 
  (.2+1,-1.5);
\end{scope}

\begin{scope}[shift={(2-2.36,0)}]
\draw[dashed, line width=2]
  (0,0) ellipse (.6 and 1.8);
\draw[white,fill=white]
  (.2,   1.44) rectangle 
  (.2+1,-1.5);
\end{scope}

\begin{scope}[shift={(4-2.36,0)}]
\draw[dashed, line width=2]
  (0,0) ellipse (.6 and 1.8);
\draw[white,fill=white]
  (.2,   1.44) rectangle 
  (.2+1,-1.5);
\end{scope}

\draw[
  line width=7,
  white
]
  (0,-1.5)
  .. controls (0,0) and (2,0) ..
  (2,+1.5);
\draw[
  line width=1.8,
  -Latex
]
  (0,-1.5)
  .. controls (0,0) and (2,0) ..
  (2,+1.5);

\begin{scope}[shift={(0,-.03)}]
\draw[
  line width=8,
  white
]
  (2,-1.5) 
    .. controls (2,0) and (-2,0) ..
  (-2,+1.5);
\end{scope}
\begin{scope}
\clip
  (-1.5,.8) rectangle (0,-.8);
\begin{scope}[shift={(0,-.07)}]
\draw[
  line width=8,
  white
]
  (2,-1.5) 
    .. controls (2,0) and (-2,0) ..
  (-2,+1.5);
\end{scope}
\end{scope}

\draw[
  line width=1.8,
  -Latex
]
  (2,-1.5) 
    .. controls (2,0) and (-2,0) ..
  (-2,+1.5);

\begin{scope}[shift={(.1,0)}]
\draw[
  line width=8,
  white
]
  (-2,-1.5)
    .. controls (-2,0) and (0,0) ..
  (0,+1.5);
\end{scope}
\draw[
  line width=1.8,
  -Latex
]
  (-2,-1.5)
    .. controls (-2,0) and (0,0) ..
  (0,+1.5);

\begin{scope}[shift={(.8,.65)}]
\begin{scope}[shift={(0,-.04)}]
\draw[line width=8, white]
  (403:1.3 and .3) arc 
  (403:472:1.4 and .3);
\end{scope}
\draw[line width=1.8]
  (403:1.3 and .3) arc 
  (403:472:1.4 and .3);
\draw[line width=7, white]
  (141:1.3 and .3) arc 
  (141:390:1.4 and .3);
\draw[line width=1.8]
  (141:1.3 and .3) arc 
  (141:390:1.4 and .3);
\draw[line width=1.8,-Latex]
  (250:1.3 and .3) arc 
  (250:264:1.4 and .3);
\end{scope}

\begin{scope}[shift={(-2,-.5)}, rotate=-40]

\begin{scope}[shift={(.23,0)}]
\draw[line width=5, white]
  (180:.35) arc 
  (180:274:.35); 
\draw[line width=1.8]
  (180:.35) arc 
  (180:274:.35); 
\draw[line width=1.8]
  (315:.35) arc 
  (315:360:.35); 
\end{scope}

\begin{scope}[shift={(-.23,0)}]
\draw[line width=5,white]
  (0:.35) arc 
  (0:180:.35); 
\draw[line width=1.8]
  (0:.35) arc 
  (0:180:.35); 
\end{scope}

\begin{scope}[shift={(.23,0)}]
\draw[line width=5, white]
  (0:.35) arc 
  (0:180:.35); 
\draw[line width=1.8]
  (0:.35) arc 
  (0:180:.35); 
\end{scope}

\begin{scope}[shift={(-.23,0)}]
\draw[line width=5, white]
  (180:.35) arc 
  (180:360:.35); 
\draw[line width=1.8]
  (180:.35) arc 
  (180:360:.35); 
\end{scope}

\end{scope}

\end{tikzpicture}      
    }
\adjustbox{
  raise=-15pt
}{
$
    \;\;\;
    \in
    \;\;\;
    \Omega
    \,
    \mathrm{Conf}^I_3(\mathbb{R}^2)
    \,.
$
}
\end{tabular}  
\vspace{.2cm} 

\noindent {\bf (iii)} To see the equivalences $\Omega \, \mathrm{Conf}^I_n(\mathbb{R}^2) \,\weakHomotopyEquivalence\, \Omega\, \mathrm{Conf}^I_0(\mathbb{R}^2)$, and thereby also all the others 
\eqref{EquivalenceOfBasedLoopSpaces},
in terms of such 
``framed link-braids'' (Fig. \ref{FigureL}) being equivalently framed links with un-braids, and hence equivalently just framed links,
observe that the saddle move from Lem. \ref{FramedLinksCobordantToFramedUnknots} in the following symmetrized form
\begin{equation}
\hspace{-.3cm}
\adjustbox{scale=0.7}{
\def\tabcolsep{9pt}
\begin{tabular}{ccccccccc}
&&&&
\\[-8pt]
\adjustbox
{
  scale=.7,
  raise=-1.4cm
}{
\hspace{-10pt}
\begin{tikzpicture}

\draw[
  line width=1.8,
  -Latex
]
  (-1.6,-1.6) -- (+1.6,+1.6);

\draw[
  line width=15,
  white
]
  (+1.6,-1.6) -- (-1.6,+1.6);
\draw[
  line width=1.8,
  -Latex
]
  (+1.6,-1.6) -- (-1.6,+1.6);
  
\end{tikzpicture}}
&
\hspace{-10pt}
\adjustbox{
  raise=-5pt,
  scale=.8
}{
\begin{tikzpicture}[decoration=snake]
  \draw[decorate, ->]
    (0,0) -- (+.55,0);
  \draw[decorate, ->]
    (0,0) -- (-.55,0);
\end{tikzpicture}
}
\hspace{-10pt}
&
\adjustbox{
  scale=.7,
  raise=-1.3cm
}
{
\begin{tikzpicture}

\draw[
  line width=1.8,
  -Latex
]
  (-1.7,-1.7) -- (0,0);

\draw[
  -Latex,
  line width=1.8,
]
  (+1.7,-1.7) -- (-1.7,1.7);

\draw[
  white,fill=white
]
  (-1.5,-1.5) rectangle (1.5,1.5);

\draw[line width=1.8]
  (-1.5,-1.5)
   .. controls
     (-.8,-.8) and (-.8,+.8) ..
  (-1.5,+1.5);

\begin{scope}[shift={(-.42,0)}]
\draw[
  line width=1.8
]
  (45:.3) arc (45:360-45:.3);
\end{scope}

\draw[line width=1.8]
  (+1.6,-1.6) -- (-.2,.2);
\draw[line width=1.8, -Latex]
  (+.2,+.2) -- (1.6,1.6);

\draw[white,fill=white]
  (-1.1,-.2) rectangle (-.5,+.2);
\clip
  (-1.1,-.2) rectangle (-.5,+.2);

\draw[line width=1.8]
  (-.82,.24) circle (.17);

\draw[line width=1.8]
  (-.81,-.24) circle (.18);
\end{tikzpicture}
}
&
\hspace{-10pt}
\adjustbox{
  raise=-5pt,
  scale=.8
}{
\begin{tikzpicture}[decoration=snake]
  \draw[decorate, ->]
    (0,0) -- (+.55,0);
  \draw[decorate, ->]
    (0,0) -- (-.55,0);
\end{tikzpicture}
}
\hspace{-10pt}
&
\adjustbox{
  scale=.7,
  raise=-1.3cm
}
{
\begin{tikzpicture}

\draw[
  line width=1.8,
  -Latex
]
  (-1.7,-1.7) -- (0,0);

\draw[
  line width=1.8,
  -Latex, 
]
  (+1.7,-1.7) -- (-1.7,+1.7);

\draw[
  white,fill=white
]
  (-1.5,-1.5) rectangle (1.5,1.5);

\draw[line width=1.8]
  (-1.5,-1.5)
   .. controls
     (-.8,-.8) and (-.8,+.8) ..
  (-1.5,+1.5);

\begin{scope}[shift={(-.42,0)}]
\draw[
  line width=1.85
]
  (45:.3) arc (45:360-45:.3);
\end{scope}

\draw[line width=1.8]
  (+1.7,-1.7) -- (-.2,.2);
\draw[line width=1.8, -Latex]
  (+.2,+.2) -- (1.7,1.7);
 
\end{tikzpicture}
}
&
\hspace{-13pt}
\adjustbox{
  raise=-5pt,
  scale=.8
}{
\begin{tikzpicture}[decoration=snake]
  \draw[decorate, ->]
    (0,0) -- (+.55,0);
  \draw[decorate, ->]
    (0,0) -- (-.55,0);
\end{tikzpicture}
}
\hspace{-13pt}
&
\adjustbox{
  scale=.7,
  raise=-1.3cm
}
{
\begin{tikzpicture}

\draw[
  line width=1.8,
]
  (-1.7,-1.7) -- (-1.5,-1.5);
\draw[
  line width=1.8,
  -Latex
]
  (-1.5,+1.5) -- (-1.7,+1.7);

\draw[
  line width=1.8,
]
  (+1.7,-1.7) -- (+1.5,-1.5);
\draw[line width=1.8, -Latex]
  (+1.5,+1.5) -- (1.7,1.7);

\draw[
  line width=1.8,
  -Latex
]
  (+.25,-.25) -- (-.25,+.25);

\draw[line width=1.8]
  (-1.5,-1.5)
   .. controls
     (-.8,-.8) and (-.8,+.8) ..
  (-1.5,+1.5);

\begin{scope}[shift={(-.42,0)}]
\draw[
  line width=1.85
]
  (45:.3) arc (45:360-45:.3);
\end{scope}

\begin{scope}[
  shift={(+.42,0)},
  xscale=-1
]
\draw[
  line width=1.85
]
  (45:.3) arc (45:360-45:.3);
\end{scope}

\begin{scope}[
  xscale=-1
]

\draw[line width=1.8]
  (-1.5,-1.5)
   .. controls
     (-.8,-.8) and (-.8,+.8) ..
  (-1.5,+1.5);

\draw[white,fill=white]
  (-1.1,-.2) rectangle (-.5,+.2);
\clip
  (-1.1,-.2) rectangle (-.5,+.2);

\draw[line width=1.8]
  (-.82,.24) circle (.17);

\draw[line width=1.8]
  (-.81,-.24) circle (.18);

\end{scope}

\end{tikzpicture}
}
&
\hspace{-13pt}
\adjustbox{
  raise=-5pt,
  scale=.8
}{
\begin{tikzpicture}[decoration=snake]
  \draw[decorate, ->]
    (0,0) -- (+.55,0);
  \draw[decorate, ->]
    (0,0) -- (-.55,0);
\end{tikzpicture}
}
\hspace{-13pt}
&
\adjustbox{
  scale=.7,
  raise=-1.3cm
}
{
\begin{tikzpicture}

\draw[
  line width=1.8,
]
  (-1.7,-1.7) -- (-1.5,-1.5);
\draw[
  line width=1.8,
  -Latex
]
  (-1.5,+1.5) -- (-1.7,+1.7);

\draw[
  line width=1.8,
]
  (+1.7,-1.7) -- (+1.5,-1.5);
\draw[line width=1.8, -Latex]
  (+1.5,+1.5) -- (1.7,1.7);

\draw[
  line width=1.8,
  -Latex
]
  (+.25,-.25) -- (-.25,+.25);

\draw[line width=1.8]
  (-1.5,-1.5)
   .. controls
     (-.8,-.8) and (-.8,+.8) ..
  (-1.5,+1.5);

\begin{scope}[shift={(-.42,0)}]
\draw[
  line width=1.85
]
  (45:.3) arc (45:360-45:.3);
\end{scope}

\begin{scope}[
  shift={(+.42,0)},
  xscale=-1
]
\draw[
  line width=1.85
]
  (45:.3) arc (45:360-45:.3);
\end{scope}

\begin{scope}[
  xscale=-1
]

\draw[line width=1.8]
  (-1.5,-1.5)
   .. controls
     (-.8,-.8) and (-.8,+.8) ..
  (-1.5,+1.5);

\end{scope}
 
\end{tikzpicture}
}
\\[-8pt]
&&&&
\end{tabular}
}
\end{equation}
\vspace{.1cm}

\noindent
un-braids any braid-link at the cost of picking up a corresponding collection of further framed link components, e.g.:

\begin{equation}
\adjustbox{
  raise=-1.7cm, scale=0.7
}{
\begin{tikzpicture}

\begin{scope}[shift={(-2.36,0)}]
\draw[dashed, line width=2]
  (0,0) ellipse (.6 and 1.8);
\draw[white,fill=white]
  (.2,   1.44) rectangle 
  (.2+1,-1.5);
\end{scope}

\begin{scope}[shift={(2-2.36,0)}]
\draw[dashed, line width=2]
  (0,0) ellipse (.6 and 1.8);
\draw[white,fill=white]
  (.2,   1.44) rectangle 
  (.2+1,-1.5);
\end{scope}

\begin{scope}[shift={(4-2.36,0)}]
\draw[dashed, line width=2]
  (0,0) ellipse (.6 and 1.8);
\draw[white,fill=white]
  (.2,   1.44) rectangle 
  (.2+1,-1.5);
\end{scope}

\draw[
  line width=7,
  white
]
  (0,-1.5)
  .. controls (0,0) and (2,0) ..
  (2,+1.5);
\draw[
  line width=1.8,
  -Latex
]
  (0,-1.5)
  .. controls (0,0) and (2,0) ..
  (2,+1.5);

\begin{scope}[shift={(0,-.03)}]
\draw[
  line width=8,
  white
]
  (2,-1.5) 
    .. controls (2,0) and (-2,0) ..
  (-2,+1.5);
\end{scope}
\begin{scope}
\clip
  (-1.5,.8) rectangle (0,-.8);
\begin{scope}[shift={(0,-.07)}]
\draw[
  line width=8,
  white
]
  (2,-1.5) 
    .. controls (2,0) and (-2,0) ..
  (-2,+1.5);
\end{scope}
\end{scope}

\draw[
  line width=1.8,
  -Latex
]
  (2,-1.5) 
    .. controls (2,0) and (-2,0) ..
  (-2,+1.5);

\begin{scope}[shift={(.1,0)}]
\draw[
  line width=8,
  white
]
  (-2,-1.5)
    .. controls (-2,0) and (0,0) ..
  (0,+1.5);
\end{scope}
\draw[
  line width=1.8,
  -Latex
]
  (-2,-1.5)
    .. controls (-2,0) and (0,0) ..
  (0,+1.5);

\end{tikzpicture}
\hspace{-.2cm}
\adjustbox{raise=1.5cm}{
\begin{tikzpicture}[decoration=snake]
  \draw[decorate, ->]
    (0,0) -- (+.55,0);
  \draw[decorate, ->]
    (0,0) -- (-.55,0);
\end{tikzpicture}
\hspace{20pt}
}
\begin{tikzpicture}

\begin{scope}[shift={(-2.36,0)}]
\draw[dashed, line width=2]
  (0,0) ellipse (.6 and 1.8);
\draw[white,fill=white]
  (.2,   1.44) rectangle 
  (.2+1,-1.5);
\end{scope}

\begin{scope}[shift={(2-2.36,0)}]
\draw[dashed, line width=2]
  (0,0) ellipse (.6 and 1.8);
\draw[white,fill=white]
  (.2,   1.44) rectangle 
  (.2+1,-1.5);
\end{scope}

\begin{scope}[shift={(4-2.36,0)}]
\draw[dashed, line width=2]
  (0,0) ellipse (.6 and 1.8);
\draw[white,fill=white]
  (.2,   1.44) rectangle 
  (.2+1,-1.5);
\end{scope}

\draw[
  line width=1.8,
  -Latex
]
  (-2,-1.5) -- (-2,+1.5);
\draw[
  line width=1.8,
  -Latex
]
  (0,-1.5) -- (0,+1.5);
\draw[
  line width=1.8,
  -Latex
]
  (+2,-1.5) -- (+2,+1.5);

\begin{scope}[
  shift={(3.7,-.7)}
]
\begin{scope}[shift={(-.42,0)}]
\draw[
  line width=1.8
]
  (45:.3) arc (45:360-45:.3);
\end{scope}

\begin{scope}[
  shift={(+.42,0)},
  xscale=-1
]
\draw[
  line width=1.8
]
  (45:.3) arc (45:360-45:.3);
\end{scope}

\draw[line width=1.8, -Latex]
  (+.22,-.22) -- (-.27,+.27);

\end{scope}

\begin{scope}[
  shift={(3.7,+.5)},
  xscale=-1
]
\begin{scope}[shift={(-.42,0)}]
\draw[
  line width=1.8
]
  (45:.3) arc (45:360-45:.3);
\end{scope}

\begin{scope}[
  shift={(+.42,0)},
  xscale=-1
]
\draw[
  line width=1.8
]
  (45:.3) arc (45:360-45:.3);
\end{scope}

\draw[line width=1.8, -Latex]
  (+.22,-.22) -- (-.27,+.27);

\end{scope}

\end{tikzpicture}
}
\end{equation}
\end{remark}

\medskip

\section{Quantum States and Braiding}
\label{QuantumObservables}

In consequence and elaboration of Thm. \ref{ChargedOpenStringLoopsClassifiedByCrossingNumber}, we show here that pure quantum states on the group algebra of any of the groups
\begin{equation}
  \pi^3(S^2) 
    \;\simeq\; 
  \pi_1 \mathbb{G}\mathrm{Map}^\ast(S^2, S^2)   \;\simeq\; 
  \pi_1 \, \mathrm{Conf}^I(\mathbb{R}^2)
\end{equation}
are naturally identified with the ``Wilson loop observables'' of abelian Chern-Simons theory, automatically with the correct expected regularization (cf. Rem. \ref{RegularizationOfWilsonLoops}), and assign {\it braiding} phases to links in the way corresponding to worldlines of abelian anyonic particles/solitons. We close in \S\ref{Conclusion} with outlook on what this means for physics.

\smallskip 
\noindent
{\bf Quantum states.}
In quantum probability theory (cf. \cite[\S 2.4]{Strocchi08}), given a complex star-algebra thought of as an algebra of quantum observables and hence here to be denoted 
$$
  \mathrm{QObs}
  \;\in\;
  \mathbb{C}\mathrm{Alg}
  \,, 
  \hspace{.8cm}
  \begin{tikzcd}[sep=0pt]
    \mathrm{QObs}
    \ar[rr, "{(-)^\ast}"]
    &&
    \mathrm{QObs}
    \\
    c \, \mathcal{O} 
      &\longmapsto& 
    \overline{c} \,\mathcal{O}^\ast
    \\
    \mathcal{O}_1 
    \mathcal{O}_2
    &\longmapsto&
    \mathcal{O}_2^\ast 
    \mathcal{O}_1^\ast
  \end{tikzcd}
  \;\;\;
  \mbox{
    for
  }
  \;
  \begin{array}{l}
    c \in \mathbb{C}
    \\
    \mathcal{O}_i \in \mathrm{QObs}
    \,,
  \end{array}
$$
the corresponding {\it quantum states} 
are 
(cf. \cite[\S I.1.1]{Meyer95}\cite[\S 7]{Warner10}\cite[Def. 2.4]{Landsman17}, exposition in \cite[p. 6]{Gleason11})
embodied by the expectation values that they induce, which are linear forms $\rho$ on observables subject to (i) reality, (ii) semi-positivity, and (3.) normalization:
\begin{equation}
  \label{QuantumStates}
  \hspace{-.2cm} 
  \mathrm{QStates}
  \;\;
  :=
  \;\;
  \bigg\{\!\!
  \begin{tikzcd}
    \rho
    \;:\;
    \mathrm{QObs}
    \ar[
      r,
      "{\color{darkgreen} \mathrm{linear} }"{swap}
    ]
    &
    \mathbb{C}
  \end{tikzcd}
  \;\Big\vert\;
  \underset{
    \mathcal{O}
    \in \mathrm{QObs}
  }{\forall}
  \Big(
    \underset{
      \mathclap{
        \raisebox{-1pt}{
          \scalebox{.7}{
            \color{darkblue}
            \bf
            reality
          }
        }
      }
    }{
    \rho\big(
      \mathcal{O}^\ast
    \big)
    \,=\,
    \rho(\mathcal{O})^\ast
    }
    ,\;
  \underset{
      \mathclap{
        \raisebox{-1pt}{
          \scalebox{.7}{
            \color{darkblue}
            \bf
            (semi-)positivity
          }
        }
      }  
  }{
  \rho(\mathcal{O}^\ast 
    \!\cdot\! 
  \mathcal{O})
  \;\geq 0\;
  \in
  \mathbb{R}
  \hookrightarrow
  \mathbb{C}
  }
  \Big)
  \,,
  \underset{
      \mathclap{
        \raisebox{-1pt}{
          \scalebox{.7}{
            \color{darkblue}
            \bf
            normalization
          }
        }
      }  
  }{
  \;\;
  \rho(1) = 1
  }\;\,
  \bigg\}
  \,.
\end{equation}
This subsumes all {\it mixed} states (``density matrices''), for which the expectation values are {\it not required} to preserve the algebra product. Among them, the {\it pure} states (those which form a Hilbert space of states) are characterized as not being convex combinations of other states.

Those (expectation values of) states that do preserve the algebra structure called {\it multiplicative states}:
\begin{equation}
  \label{MultiplicativeStates}
  \mbox{$
  \rho \,\in\, \mathrm{QStates}
  $ is multiplicative}
  \hspace{.8cm}
  :\Leftrightarrow
  \hspace{.8cm}
  \underset{
    \mathcal{O},
    \mathcal{O}'
    \in \mathrm{QObs}
  }{\forall}
  \;\;\;
  \rho\big(
    \mathcal{O}
    \cdot
    \mathcal{O}'
  \big)
  \;\;
  =
  \;\;
  \rho\big(\mathcal{O}\big)
  \,
  \rho\big(\mathcal{O}'\big)
  \,.
\end{equation}
We will need the following fact:
\begin{lemma}[{\bf Multiplicative states are pure} {(e.g. \cite[Ex. 13.3-4]{Zhu93}\cite[Lem. 7.20-21]{Warner10})}]
\label{MultiplicativeStatesArePure}
  Every multiplicative state \eqref{MultiplicativeStates} is pure; and on central observables the multiplicative states coincide with the pure states.
\end{lemma}

\medskip

\noindent
{\bf Quantum states of framed link cobordism.}
We now study the case where the algebra $\mathrm{QObs}$ happens to be the group algebra $\mathbb{C}[-]$ (cf. \cite[\S 3.4]{FultonHarris91}) of homotopy classes of loops in the signed interval configuration space (Def. \ref{ConfigurationsOfChargedStrings}):
\begin{equation}
  \mathrm{QObs}
  \;\;
  :=
  \;\;
  \mathbb{C}\Big[
    \pi_0\big(\Omega
    \,\mathrm{Conf}^I_0(\mathbb{R}^2)\big)
  \Big]
\end{equation}
(we explain in \S\ref{Conclusion}, following \cite{SS25-FQHViaAlgTop}, how this arises as the algebra of topological quantum observables on effective flux through certain 2D quantum materials),
whose elements are equivalently compactly supported functions 
\begin{equation}
  \label{QuantumObservablesAsFunctions}
  \begin{tikzcd}
    \mathcal{O}
    \;:\;
    \pi_0
    \big(
      \Omega \, 
      \mathrm{Conf}^I_0(\mathbb{R}^2)
    \big)
    \ar[r]
    &
    \mathbb{C}
    \,.
  \end{tikzcd}
\end{equation}

Now, by Thm. \ref{ChargedOpenStringLoopsClassifiedByCrossingNumber}, these quantum observables detect the total crossing number/writhe $\#L$ (Def. \ref{LinkingNumber}) of the framed links $L$ which these loops equivalently form by Prop. \ref{ChargedStringLoopsAsFramedLinksOnEquivalenceClasses}. A choice of $\mathbb{C}$-linear basis of the underlying vector space is hence given by:
\begin{equation}
  \label{BasisObservablesInDegreeZero}
  \mathrm{QObs}
  \;\;
  \simeq_{{}_{\mathbb{C}}}
  \;\;
  \mathbb{C}
  \Big\langle
  \mathcal{O}_w
  \;:\;
  [L]
  \;\mapsto\;
  \delta\big(\# L,\,  w\big)
  \Big\rangle_{w \in \mathbb{Z}}
\end{equation}
(with $\delta(-,-)$ being the Kronecker symbol)
on which the product 
and star-operation is readily found to be:
\begin{equation}
  \label{PontrjaginProductOfDegZeroObs}
  \mathcal{O}_w
  \cdot
  \mathcal{O}_{w'}
  \;\;
  =
  \;\;
  \mathcal{O}_{w + w'}
  \,,
  \hspace{1cm}
  \big(
    \mathcal{O}_w
  \big)^\ast
  \;=\;
  \mathcal{O}_{-w}
  \,.
\end{equation}

\begin{proposition}[\bf The pure quantum states]
  \label{TheCovarianceFreeQuantumStates}
  The (expectation values of) pure quantum states \eqref{QuantumStates} 
  on $\mathrm{QObs}$
  \eqref{BasisObservablesInDegreeZero}
  are precisely the linear maps of the form   
  \begin{equation}
    \label{CovarianceFreeQuantumStates}
    \begin{tikzcd}[sep=0pt]
      \mathrm{QObs}
      \ar[
        rr,
        "{ \rho_w }"
      ]
      &&
      \mathbb{C}
      \\[-2pt]
      \mathcal{O}_w
      &\longmapsto&
      \zeta^n
      \mathrlap{
        \,\defneq\,
        \exp\big(
        \tfrac{2 \pi \mathrm{i}}{k}
        \, w
       \big)  
     }
    \end{tikzcd}
    \hspace{3cm}
    \mbox{\rm for any}
    \;\;\;
    k \in \mathbb{R} \setminus \{0\}
    \,.
  \end{equation}
\end{proposition}
\begin{proof}
With \eqref{PontrjaginProductOfDegZeroObs} and by Lem. \ref{MultiplicativeStatesArePure},
a pure state $\rho$ on the commutative observables $\mathrm{QObs}$ restricts to and is fixed by a group homomorphism 
$$
  \rho\big(
    \mathcal{O}_{w+w'}
  \big)
  \;=\;
  \rho\big(
    \mathcal{O}_w
    \cdot
    \mathcal{O}_{w'}
  \big)
  \;=\;
  \rho\big(
    \mathcal{O}_w
  \big)
  \,
  \rho\big(
    \mathcal{O}_{w'}
  \big)
$$
 from the additive group of integers to the multiplicative group of non-vanishing (due to the normalization condition) complex numbers, hence:

\vspace{-.5cm}
\begin{equation}
  \label{GroupHomomorphicQuantumState}
  \begin{tikzcd}[sep=0pt]
    \mathbb{Z}
    \ar[rr]
    &&
    \mathbb{C}^\times
    \\
    n &\longmapsto& 
    \rho\big( \mathcal{O}_w \big)
    \mathrlap{
      \;\, = \;
      \rho(\mathcal{O}_1)^{w}
      \,.
    }
  \end{tikzcd}
\end{equation}
Moreover, using also the reality condition \eqref{QuantumStates} gives that $\rho(\mathcal{O}_w)$  is unitary
\begin{equation}
  \rho\big(
    \mathcal{O}_w
  \big)^\ast
  \;
  \underset{
    \mathclap{
      \raisebox{-1pt}{
        \scalebox{.7}{
          \eqref{QuantumStates}
        }
      }
    }
  }{
  \;=\;
  }
  \;
  \rho\big(
    \mathcal{O}_w^\ast
  \big)
  \;
  \underset{
    \mathclap{
      \raisebox{-1pt}{
        \scalebox{.7}{
          \eqref{PontrjaginProductOfDegZeroObs}
        }
      }
    }  
  }{
  \;=\;
  }
  \;
  \rho\big(
    \mathcal{O}_{-w}
  \big)
  \;
  \underset{
    \mathclap{
      \raisebox{-1pt}{
        \scalebox{.7}{          \eqref{GroupHomomorphicQuantumState}
        }
      }
    }    
  }{
    \;=\;
  }
  \;
  \rho\big(
    \mathcal{O}_{w}
  \big)^{-1}
\end{equation}
and hence of the claimed form \eqref{CovarianceFreeQuantumStates}.

Conversely it just remains to observe that every map of the form \eqref{CovarianceFreeQuantumStates} really is (the expectation value of) a quantum state \eqref{QuantumStates}, which is immediate.
\end{proof}

The combination of Prop. \ref{TheCovarianceFreeQuantumStates}
with Thm. \ref{ChargedOpenStringLoopsClassifiedByCrossingNumber} now gives:
\begin{corollary}[\bf Pure states on framed links]
\label{CovarianceFreeTopologicalQuantumStatesAsFunctions}
When 
regarded as functions on framed links via Prop. \ref{ChargedStringLoopsAsFramedLinksOnEquivalenceClasses}, the pure quantum states $\rho_k$ are given by
\begin{equation}
  \label{WaveFunctions}
  \begin{tikzcd}[
    row sep=0pt,
    column sep=0pt
  ]
    \Omega
    \, 
    \mathrm{Conf}^I_0(\mathbb{R}^2)
    \ar[
      rr,
      ->>
    ]
    &&
    \pi_0
    \big(
    \Omega
    \, 
    \mathrm{Conf}^I_0(\mathbb{R}^2)
    \big)
    \ar[
      rr,
      "{ \rho_k }"
    ]
    &&
    \mathbb{C}
    \\[-2pt]
    L 
      &\longmapsto&
    \mathcal{O}_{{}_{\# L}}
      &\mapsto&
    \exp\big(
      \tfrac{2 \pi \mathrm{i}}{k}
      \,
      \# L
    \big)
    \mathrlap{\,.}
  \end{tikzcd}
\end{equation}
\end{corollary}

This is remarkable because it coincides with the known form of quantum states/observables of abelian Chern-Simons theory:
\begin{remark}[\bf Identification with quantum observables of abelian Chern-Simons theory]
\label{IdentificationWithAbelianChernSimons}
$\,$

\noindent {\bf (i)} For Chern-Simons theory with abelian gauge group it is widely understood by appeal to path-integral arguments  (\cite[p. 363]{Witten89}\cite[p. 169]{FroehlichKing89} following \cite{Polyakov88}) that 
\begin{itemize}[
  leftmargin=.7cm,
  topsep=1pt,
  itemsep=2pt
]
\item the quantum states of the gauge field are labeled by a {\it level}  \footnote{In abelian Chern-Simons theory with non-compact gauge group, the level may indeed be any non-zero real number (cf. \cite[p. 169]{FroehlichKing89}), just as in \eqref{CovarianceFreeQuantumStates}. The \emph{level quantization} 
in $\mathrm{U}(1)$ (spin) Chern-Simons theory,
constraining $k$ to be a (half) integer,  arises from our cohomotopical analysis when considering quantum states not just on the plane $\Sigma^2 \,\defneq\,\mathbb{R}^2$ as considered here (cf. \S\ref{Conclusion}) but also on the torus $\Sigma^2 \,\defneq\, \mathbb{R}^2/\mathbb{Z}^2$: This is discussed in \cite[(8) \& \S3.4]{SS25-FQHViaAlgTop}.
} 
$k \in \mathbb{R} \setminus \{0\}$,

\item the quantum observables are labeled by framed links $L$, 

often considered as equipped with labels (charges) $q_i$ on their $i$th connected component $L_i$
\end{itemize}
and the expectation value of these observables in these states is the charge-weighted exponentiated framing- and linking numbers (Def. \ref{LinkingNumber}) as follows
(\cite[p. 363]{Witten89}, cf. review e.g. in \cite[(5.1)]{MPW19}):
\begin{equation}
  \label{CSWilsonLoopObservable}
  \def\arraystretch{2}
  \begin{array}{ccl}
  \mathrm{W}_{\!k}(L)
  &=&
  \exp\bigg(\!
    \frac{2\pi \mathrm{i}}{k}
    \Big(
    \sum_{i}
    q_i^2
    \,
    \mathrm{frm}(L_i)
    +
    \sum_{i,j}
    q_i q_j
    \,
    \mathrm{lnk}(L_i, L_j)
    \Big)
  \!\bigg)
  \,.
  \end{array}
\end{equation}

\noindent {\bf (ii)}  However, with the charges $q_i$ being integers, we may equivalently replace a $q_i$-charged component $L_i$ with $q_i$ unit-charged parallel copies of $L_i$, and hence assume without loss of generality that $\forall_i \;\;q_i = 1$. With this, we may observe that the Chern-Simons expectation values \eqref{CSWilsonLoopObservable} coincide exactly with our pure topological quantum states \eqref{WaveFunctions}:
$$
  W_k(L)
   \;=\;
  \exp\bigg(\!
    \tfrac{2\pi \mathrm{i}}{k}
    \Big(
    \sum_{i}
    \mathrm{frm}(L_i)
    +
    \sum_{i,j}
    \mathrm{lnk}(L_i, L_j)
    \Big)
  \! \bigg)
  \;
  \underset{
    \mathclap{
      \adjustbox{
        raise=-2pt,
        scale=.7
      }{
        \eqref{TotalCrossingNumber}
      }
    }
  }{
    =
  }
  \;
  \exp\Big(
    \tfrac{2\pi \mathrm{i}}{k}
    \#(L)
  \Big).
$$
\end{remark}
\begin{remark}[\bf Regularization of Wilson loop observables]
\label{RegularizationOfWilsonLoops}
In the Chern-Simons field theory literature, going back to \cite{Polyakov88}\cite{Witten89}, the all-important 
framing of links is imposed in an  {\it ad hoc} manner in order to work around an ill-defined expression appearing from the path-integral 
arguments (going back to \cite[p. 326]{Polyakov88}). In contrast, our approach uses only well-defined constructions, and the framing emerges naturally from the differential topology of cohomotopy via the Pontrjagin theorem.
\end{remark}

\begin{remark}[\bf Comparison to the literature on Hopfions]
  The construction and arguments in \S\ref{ConfigurationSpace} are closely related to the traditional analysis of ``Hopf-charged solitons'' (or {\it Hopfions}) on $\mathbb{R}^3$ in Lagrangian field theories like the Skyrme-Fadeev model (\cite{FadeevNiemi97}, review in \cite{Fadeev02}\cite[\S 9.11]{MantonSutcliffe04}) or the 3D $\mathbb{C}P^1$ sigma-model (cf. \cite{RaduTchrakianYang13}), as can be seen from Prop. \ref{TheFundamenralGroupOfStringConfigurationSpace} and its proof. 
  Indeed, the authors of \cite{WilzcekZee83} already proposed (by at least implicit appeal to the Pontrjagin theorem) to identify the worldlines of such Hopfions on $\mathbb{R}^{1,2}$ with that of anyons in much the way that we find is the case here.
  What seems not to have been noticed before in previous literature is the finer identification of Thm. \ref{ChargedOpenStringLoopsClassifiedByCrossingNumber} using the finer theorems by Segal-Okuyama \cite{Okuyama05} 
  and the resulting identification with Chern-Simons Wilson loop observables (Rem. \ref{IdentificationWithAbelianChernSimons}), which, we suggest, fully nails down this identification.
\end{remark}

\section{Conclusion and Outlook}
\label{Conclusion}

\noindent
{\bf Summary of Results.}
We have naturally identified the fundamental group of the mapping space $\mathrm{Map}^\ast(S^2, S^2)$ with cobordism classes of framed oriented links (Prop. \ref{ChargedStringLoopsAsFramedLinksOnEquivalenceClasses}), by combining a classical theorem of Segal, relating iterated loop spaces to group completed configuration spaces of points, with a more recent and maybe previously underappreciated result by Okuyama, modeling the latter by configurations of intervals. We found that, under this identification, the multiple of the Hopf generator  corresponds to the sum $\#L$ of the linking number and framing number of the framed links $L$ (Thm. \ref{ChargedOpenStringLoopsClassifiedByCrossingNumber}), which we observed happens to be the properly regularized ``Wilson loop observable'' of abelian Chern-Simons theory (Rem. \ref{IdentificationWithAbelianChernSimons}, \ref{RegularizationOfWilsonLoops}). Finally we captured this relation to quantum invariants more systematically by showing how the main result implies that these Wilson loop obseravbles are the expectation values of the pure quantum states (in the sense of quantum probability theory) on the group algebra of cobordism classes of framed links (Cor. \ref{CovarianceFreeTopologicalQuantumStatesAsFunctions}), given as such by assigning one fixed complex phase factor to each crossing of strands in the framed link diagram.

\bigskip

Regarding framed links as wordlines of (solitonic) particles moving in a plane,
such an assignment of complex quantum phases to braidings is exactly the characteristic property of (abelian) {\it anyonic} quantum particles. Therefore we close with a brief outlook on how this is not just a coincidence but a profound relation between anyons and cohomotopy (due to \cite{SS25-FQHViaAlgTop}, to which we refer for extensive details and referencing).

\medskip

\noindent
{\bf Exotic quantization law for FQH flux.}
In quantum materials known as {\it fractional quantum Hall systems} (FQH), a strong magnetic field transversally penetrates a gas of electrons confined to an effectively 2-dimensional plane 
\begin{equation}
  \Sigma^2 \simeq \mathbb{R}^2
  \mathrlap{\,.}
\end{equation}
Around critical values of the total magnetic flux through $\Sigma^2$, a subtle  effect of the strong coupling of the electrons makes any further magnetic {\it flux quantum} through $\Sigma^2$ induce an exotic kind of vortex in the electron gas (called a {\it quasi-hole}) which as such becomes anyonic, in that each braiding of worldlines of such surplus flux- quanta/quasi-holes makes the quantum state of the system pick up a fixed {\it braiding phase} $\zeta$, just as seen in the above discussion.

While experimental observation of this phenomenon has been consistently reported in recent years by various groups and in various materials,
the theoretical understanding of the effect remains shallow, this being a general problem with strongly coupled quantum systems.

Concretely, in the absence of a strongly coupled electrons, the standard theory of electromagnetism predicts 
\cite{SS23-Obs}
that the topological quantum observables of the flux quanta through $\Sigma^2$ form the group algebra
\begin{equation}
  \label{GroupAlgebraForElectromagnetism}
  \mathbb{C}\Big[
    \pi_1
    \,
    \mathrm{Map}^\ast\big(
      \Sigma^2_{\cpt}
      ,\,
      \mathcolor{purple}{\mathbb{C}P^\infty}
    \big)
  \Big]
  \;\simeq\;
  \mathbb{C}\Big[
    \pi_0
    \,
    \mathrm{Map}^\ast\big(
      S^2
      ,\,
      S^1
    \big)
  \Big]  
  \;\simeq\;
  \mathbb{C}
  \,,
\end{equation}
 where the infinite projective space $\mathbb{C}P^\infty$ is a model for the {\it classifying space} $B \mathrm{U}(1) \,\sim\, K(\mathbb{Z},2)$ of charges of an abelian Maxwell gauge field, reflecting that ordinary magnetic flux is quantized in ordinary 2-cohomology.

This algebra \eqref{GroupAlgebraForElectromagnetism} being trivial means that ordinary magnetic flux through a plane shows no anyonic effect --- as expected. But the way this is deduced hereby from nothing but the algebraic topology of the the classifying space of the electromagnetic field suggests that the effect of the strong electron coupling, which is experimentally seen to substantially deform the situation, might be mathematically modeled by a deformation of the classifying space, reflecting an effective {\it exotic flux quantization} law \cite{SS24-Flux}
for the surplus magnetic field, due to the strong electron coupling.

The results which we found in this article show that exactly this is the case: If we postulate that the surplus FQH flux is effectively quantized in 2-cohomotopy instead of in ordinary 2-cohomology, and hence if we replace the ordinary classifying space by just its 2-skeleton, the 2-sphere,
\begin{equation}
  S^2
  \;\simeq\;
  \mathbb{C}P^1
  \xhookrightarrow{\;}
  \mathbb{C}P^\infty
  \,,
\end{equation}
then the algebra of topological quantum observables becomes 
$$
  \mathbb{C}
  \Big[
  \pi_1
  \, 
  \mathrm{Map}^\ast\big(
    \Sigma^2_{\cpt}
    ,\,
    \mathcolor{purple}{\mathbb{C}P^1}
  \big)
  \Big]
  \;\simeq\;
  \mathbb{C}
  \Big[
  \pi_1
  \, 
  \mathrm{Map}^\ast\big(
    S^2
    ,\,
    S^2
  \big)
  \Big]
  \;\simeq\;
  \mathbb{C}\big[\mathbb{Z}\big]
  \,,
$$
generated by the Hopf fibration,
and the results of \S\ref{QuantumObservables} show that  its pure quantum states are naturally identified with the abelian Chern-Simons Wilson loop observables assigning a fixed {\it braiding phase} to each crossing of soliton worldlines -- just as seen in experiment.

Specifically, the loops in the 0-component $\mathrm{Conf}^I_0(\mathbb{R}^2)$ that we have been dealing with correspond this way to ``vacuum-to-vacuum'' anyon processes, which are precisely those envisioned in application of anyon braiding to topological quantum computing, see Fig. \ref{AnyonLinksInTQC}.

This novel re-derivation of anyonic braiding phenomena from the algebraic topology of exotic flux quantization in 2-Cohomotopy immediately generalizes to more general surface geometries $\Sigma^2$ (disks, tori, annuli, ...) where it reproduces subtle effects expected in the literature, and makes new predictions. This is further discussed in \cite{SS25-FQHViaAlgTop}.

\vspace{.2cm}

\hypertarget{FigureC}{}
\begin{minipage}{9.6cm}
 \footnotesize {\bf Figure \figurenumber. \label{AnyonLinksInTQC}}
 In the traditional picture of anyon braiding processes implementing topological quantum computations (e.g.
 \cite[Fig. 17]{Kauffman02} \cite[Fig. 2]{FKLW03}\cite[p. 10]{NSSFD08}\cite[Fig. 2]{Rowell16}\cite[Fig. 2]{DMNW17} \cite[Fig. 3]{RowellWang18} \cite[Fig. 1]{Rowell22}), the computation is:
 
 \begin{itemize}[
   leftmargin=.8cm,
   topsep=1pt,
   itemsep=2pt
 ]
 \item[(i)] initialized by creating anyon/anti-anyon pairs out of the vacuum $\varnothing$, 
 
 \item[(ii)] executed by adiabatically braiding their worldlines, 
 
 \item[(iii)] read-out by annihilating the anyons again into the vacuum $\varnothing$.
 \end{itemize}

 This means that the computation is encoded by a {\it link diagram} and that its result is the corresponding Wilson loop observable, just as here we naturally find realized here (for the case of abelian anyons).
\end{minipage}
\;\;
\adjustbox{
  raise=-2.5cm,
  scale=.7
}{
\begin{tikzpicture}

\begin{scope}[xscale=-1]
\draw[
  line width=.27cm,
  gray,
]
  (0,0) 
    .. controls (0,1) and (1.7,1) .. 
  (1.7,2)
    .. controls (1.7,3) and (0,3) ..
  (0,4);
\draw[
  line width=.22cm,
  white,
]
  (0,0) 
    .. controls (0,1) and (1.7,1) .. 
  (1.7,2)
    .. controls (1.7,3) and (0,3) ..
  (0,4);
\end{scope}
\draw[line width=.25cm]
  (0,0) 
    .. controls (0,1) and (1.7,1) .. 
  (1.7,2)
    .. controls (1.7,3) and (0,3) ..
  (0,4);
\begin{scope}[xscale=-1]
\draw[
  line width=.22cm,
  white,
  draw opacity=.5
]
  (0,0) 
    .. controls (0,1) and (1.7,1) .. 
  (1.7,2)
    .. controls (1.7,3) and (0,3) ..
  (0,4);
\end{scope}


\begin{scope}[
  shift={(2,0)},
  xscale=-1
]
\draw[line width=.4cm, white]
  (0,0) 
    .. controls (0,1) and (1.7,1) .. 
  (1.7,2);
\begin{scope}[xscale=-1]
\draw[
  line width=.27cm,
  gray,
]
  (0,0) 
    .. controls (0,1) and (1.7,1) .. 
  (1.7,2)
    .. controls (1.7,3) and (0,3) ..
  (0,4);
\draw[
  line width=.22cm,
  white,
]
  (0,0) 
    .. controls (0,1) and (1.7,1) .. 
  (1.7,2)
    .. controls (1.7,3) and (0,3) ..
  (0,4);
\end{scope}
\draw[line width=.25cm]
  (0,0) 
    .. controls (0,1) and (1.7,1) .. 
  (1.7,2)
    .. controls (1.7,3) and (0,3) ..
  (0,4);
\begin{scope}[xscale=-1]
\draw[
  line width=.22cm,
  white,
  draw opacity=.5
]
  (0,0) 
    .. controls (0,1) and (1.7,1) .. 
  (1.7,2)
    .. controls (1.7,3) and (0,3) ..
  (0,4);
\end{scope}
\end{scope}


\begin{scope}
\clip
  (0,2.5)
  rectangle (2,3.3);

\draw[line width=.4cm, white]
  (0,0) 
    .. controls (0,1) and (1.7,1) .. 
  (1.7,2)
    .. controls (1.7,3) and (0,3) ..
  (0,4);
\draw[line width=.25cm]
  (0,0) 
    .. controls (0,1) and (1.7,1) .. 
  (1.7,2)
    .. controls (1.7,3) and (0,3) ..
  (0,4);
\end{scope}

\node at (0,-.25) {$\varnothing$};
\node at (2,-.25) {$\varnothing$};
\node at (0,4.2) {$\varnothing$};
\node at (2,4.2) {$\varnothing$};

\end{tikzpicture}
}

\medskip

\noindent
{\bf Relation to high energy physics.}
Finally, the inclined reader may wonder whether there is deeper rationale for deforming, in certain situations, the traditional classifying space $\mathbb{C}P^\infty \simeq B \mathrm{U}(1)$ for electromagnetic charge by one of its skeleta $\mathbb{C}P^n \subset \mathbb{C}P^\infty$ --- apart from the above curiosity that for $n = 1$ this captures the nature of FQH anyons, by our main theorem here. And indeed (exposition in \cite{SS25-Srni}): In discussion of ``UV-completion'' of particle physics one may consider ``geometrically engineering'' gauge fields on {\it branes}, specifically on M5-branes \cite{GSS24-FluxOnM5} probing orbi-singularities \cite{SS26-Orb}. A careful analysis of the admissible quantization laws \cite{SS24-Flux} for the {\it higher} flux on these 5-branes suggests that its classifying space should indeed be $\mathbb{C}P^3$, which on an A-type orbi-singularity reduces to $\mathbb{C}P^1 \,\simeq\, S^2$ \cite{SS25-TEC}\cite{SS25-Seifert}. This is how we originally discovered the phenonema presented here. 

\vspace{1cm}





\medskip

\begin{thebibliography}{c}

\bibitem[Ad94]{Adams94}
C. C. Adams, 
{\it\color{darkblue}The Knot Book -- An elementary Introduction to the Mathematical Theory of Knots}, W. H. Freedman and Co. (1994), [\href{https://bookstore.ams.org/KNOT}{\tt ISBN:978-0-8218-3678-1}].








\bibitem[Be21]{Benedetti21}
R. Benedetti, 
{\it\color{darkblue}Lectures on Differential Topology}, Graduate Studies in Mathematics {\bf 218}, AMS (2021), 
[\href{https://bookstore.ams.org/cdn-1631100821985/gsm-218}{\tt ISBN:978-1-4704-6674-9}],
[\href{https://arxiv.org/abs/1907.10297}{\tt arXiv:1907.10297}]. 




\bibitem[Br93]{Bredon93}
G. Bredon, 
{\it\color{darkblue}Topology and Geometry}, Graduate Texts in Mathematics {\bf 139}, Springer, Berlin (1993),  \newline 
[\href{https://link.springer.com/book/10.1007/978-1-4757-6848-0}{\tt doi:10.1007/978-1-4757-6848-0}].


\bibitem[CS80]{CS80}
S. E. Cappell and J. L. Shaneson,
{\it\color{darkblue}Link cobordism}, Commentarii Mathematici Helvetici {\bf 55} (1980), 20–49, \newline 
[\href{https://doi.org/10.1007/BF02566673}{\tt doi:10.1007/BF02566673}].


\bibitem[CW81]{CarusoWaner81}
J. Caruso and S. Waner, 
{\it\color{darkblue}An Approximation to $\Omega^n \Sigma^n X$}, Trans. Amer.
Math. Soc. {\bf 265} 1 (1981), 147-162,  
[\href{https://doi.org/10.2307/1998487}{\tt doi:10.2307/1998487}].





\bibitem[CF63]{CrowellFox63}
R. H. Crowell and R. H. Fox,
{\it\color{darkblue}Introduction to knot theory}, 
Graduate Texts in Mathematics {\bf 57}, Springer (1963),
[\href{https://doi.org/10.1007/978-1-4612-9935-6}{\tt doi:10.1007/978-1-4612-9935-6}].


\bibitem[DMNW17]{DMNW17}
A. Di Pierro, R. Mengoni, R. Nagarajan, and D. Windridge,
{\it\color{darkblue}Hamming Distance Kernelisation via Topological Quantum Computation}, in: {\it  Theory and Practice of Natural Computing. TPNC 2017}, Lecture Notes in Computer Science {\bf 10687}, Springer (2017), 269-280, 
[\href{https://doi.org/10.1007/978-3-319-71069-3_21}{\tt doi:10.1007/978-3-319-71069-3\_21}].










\bibitem[EHI20]{EHI20}
M. Elhamdadi, M. Hajij, and K. Istvan,
{\it\color{darkblue}Framed Knots}, 
Math. Intelligencer {\bf 42} (2020), 7–22, \newline
[\href{https://doi.org/10.1007/s00283-020-09990-0}{\tt doi:10.1007/s00283-020-09990-0}],
[\href{https://arxiv.org/abs/1910.10257}{\tt arXiv:1910.10257}].

\bibitem[Fa02]{Fadeev02}
L. D. Faddeev, 
{\it\color{darkblue}Knotted solitons}, Proceedings of the ICM Beijing 2002 {\bf 1}, Higher Education Press (2002), 235-244,
[\href{https://arxiv.org/abs/math-ph/0212079}{\tt arXiv:math-ph/0212079}].

\bibitem[FaNi97]{FadeevNiemi97}
L. D. Faddeev and A. J. Niemi, 
{\it\color{darkblue}Knots and Particles}, 
Nature {\bf 387} (1997), 58–61,
[\href{https://doi.org/10.1038/387058a0}{\tt doi;10.1038/387058a0}],
[\href{https://arxiv.org/abs/hep-th/9610193}{\tt arXiv:hep-th/9610193}].



\bibitem[FH01]{FadellHusseini01}
E. Fadell and S. Husseini, 
{\it\color{darkblue}Geometry and topology of configuration spaces}, Springer Monographs in Mathematics (2001), 
[\href{https://link.springer.com/book/10.1007/978-3-642-56446-8}{\tt doi:10.1007/978-3-642-56446-8}].









\bibitem[FSS23]{FSS23Char}
D. Fiorenza, H. Sati, and U. Schreiber,
{\it\color{darkblue}The Character map in Nonabelian Cohomology --- Twisted, Differential and Generalized},
World Scientific, Singapore (2023),
[\href{https://doi.org/10.1142/13422}{\tt doi:10.1142/13422}],
[\href{https://arxiv.org/abs/2009.11909}{\tt arXiv:2009.11909}].





\bibitem[FKLW03]{FKLW03}
M. Freedman, A. Kitaev, M. Larsen, and Z. Wang, 
{\it\color{darkblue}Topological quantum computation}, 
Bull. Amer. Math. Soc. {\bf 40} (2003), 31-38,
[\href{https://doi.org/10.1090/S0273-0979-02-00964-3}{\tt doi:10.1090/S0273-0979-02-00964-3}],
[\href{https://arxiv.org/abs/quant-ph/0101025}{\tt arXiv:quant-ph/0101025}].


\bibitem[FK89]{FroehlichKing89}
J. Fr{\"o}hlich and C. King, 
{\it\color{darkblue}The Chern-Simons theory and knot polynomials}, 
Commun. Math. Phys. {\bf 126} 1 (1989), 167-199,
[\href{https://doi.org/10.1007/BF02124336}{\tt doi:10.1007/BF02124336}],
[\href{https://projecteuclid.org/euclid.cmp/1104179728}{\tt euclid:cmp/1104179728}].


\bibitem[FH91]{FultonHarris91}
W. Fulton and J. Harris, 
{\it\color{darkblue}Representation Theory: A First Course}, 
Springer (1991), \newline 
[\href{https://link.springer.com/book/10.1007/978-1-4612-0979-9}{\tt doi:10.1007/978-1-4612-0979-9}].




\bibitem[GSS24]{GSS24-FluxOnM5}
G. Giotopoulos, H. Sati, and U. Schreiber,
{\it\color{darkblue}Flux Quantization on M5-Branes},
JHEP {\bf 2024} 140 (2024)
[\href{https://arxiv.org/abs/2406.11304}{\tt arXiv:2406.11304}], [\href{https://doi.org/10.1007/JHEP10(2024)140}{\tt doi:10.1007/JHEP10(2024)140}].







\bibitem[Gl11]{Gleason11}
J. Gleason, 
{\it\color{darkblue}From classical to quantum: The algebraic approach}, lecture at \href{http://www.math.uchicago.edu/~may/VIGRE/}{VIGRE} \href{http://www.math.uchicago.edu/~may/VIGRE/VIGREREU2011.html}{\it REU 2011}, Chicago (2011),
[\href{https://ncatlab.org/nlab/files/GleasonAlgebraic.pdf}{\tt ncatlab.org/nlab/files/GleasonAlgebraic.pdf}].




\bibitem[GL69]{GorinLin69}
E. A. Gorin and V. Ya. Lin, 
{\it\color{darkblue}Algebraic equations with continuous coefficients and some problems of the algebraic theory of braids}, 
Math. USSR-Sb. {\bf 7} 4 (1969), 569–596,
[\href{https://doi.org/10.1070/SM1969v007n04ABEH001104}{\tt doi:10.1070/SM1969v007n04ABEH001104}],
\newline 
[\href{https://www.mathnet.ru/eng/sm3572}{\tt mathnet:sm3572}].













\bibitem[Ho68]{Hosokawa67}
F. Hosokawa, 
{\it\color{darkblue}A Concept of Cobordism between Links}, Ann. Math. 
{\bf 86} 2 (1967), 362-373, \newline 
[\href{https://doi.org/10.2307/1970693}{\tt doi:10.2307/1970693}].



\bibitem[Hu59]{STHu59}
S.-T. Hu, 
{\it\color{darkblue}Homotopy Theory}, 
Academic Press (1959),
[\href{https://www.maths.ed.ac.uk/~v1ranick/papers/hu2.pdf}{\tt maths.ed.ac.uk/$\sim$v1ranick/papers/hu2.pdf}]




\bibitem[HSS19]{HSS19}
J. Huerta, H. Sati, and U. Schreiber,
{\it\color{darkblue}Real ADE-equivariant (co)homotopy and Super M-branes}, 
Commun. Math. Phys. {\bf 371} (2019), 425–524,
[\href{https://doi.org/10.1007/s00220-019-03442-3}{\tt doi:10.1007/s00220-019-03442-3}],
[\href{https://arxiv.org/abs/1805.05987}{\tt arXiv:1805.05987}].


\bibitem[Jac04]{Jacobsson04}
M. Jacobsson, 
{\it\color{darkblue}An invariant of link cobordisms from Khovanov homology}, 
Algebr. Geom. Topol. {\bf 4} (2004), 1211-1251,
[\href{https://doi.org/10.2140/agt.2004.4.1211}{\tt doi:10.2140/agt.2004.4.1211}],
[\href{https://arxiv.org/abs/math/0206303}{\tt arXiv:math/0206303}].




\bibitem[Ja84]{James84}
I. M. James, 
{\it\color{darkblue}General Topology and Homotopy Theory}, 
Springer (1984),
[\href{https://doi.org/10.1007/978-1-4613-8283-6}{\tt doi:10.1007/978-1-4613-8283-6}].

\bibitem[JS93]{JoyalStreet93}
A. Joyal and R. Street,
{\it\color{darkblue}Braided tensor categories}, 
Adv. Math. {\bf 102} (1993), 20-78,
\newline 
[\href{https://doi.org/10.1006/aima.1993.1055}{\tt doi:10.1006/aima.1993.1055}].


\bibitem[Ka24]{Kallel24}
S. Kallel, 
{\it\color{darkblue}Configuration spaces of points: A user’s guide}, Encyclopedia of Mathematical Physics 2nd ed. (2024), 98-135, 
 [\href{https://doi.org/10.1016/B978-0-323-95703-8.00211-1}{\tt doi:10.1016/B978-0-323-95703-8.00211-1}],
 [\href{https://arxiv.org/abs/2407.11092}{\tt arXiv:2407.11092}].


\bibitem[Kau02]{Kauffman02}
L. H. Kauffman, 
{\it\color{darkblue}Quantum Topology and Quantum Computing}, in: {\it Quantum Computation: A Grand Mathematical Challenge for the Twenty-First Century and the Millennium}, Proceedings of Symposia in Applied Mathematics {\bf 58}, AMS (2002), 
[\href{https://doi.org/10.1090/psapm/058}{\tt doi:10.1090/psapm/058}].

\bibitem[Kau15]{Kauffman15}
L. H. Kauffman,  
{\it\color{darkblue}Virtual Knot Cobordism}, in: {\it New Ideas in Low Dimensional Topology},  335-377, World Scientific (2015),
[\href{https://doi.org/10.1142/9789814630627_0009}{\tt doi:10.1142/9789814630627\_0009}],
[\href{https://arxiv.org/abs/1409.0324}{\tt arXiv:1409.0324}].


\bibitem[Kh00]{Khovanov00}
M. Khovanov, 
{\it\color{darkblue}A categorification of the Jones polynomial}, Duke Math. J. {\bf 101} (2000), 359-426, \newline  
[\href{https://projecteuclid.org/journals/duke-mathematical-journal/volume-101/issue-3/A-categorification-of-the-Jones-polynomial/10.1215/S0012-7094-00-10131-7.full}{\tt doi:10.1215/S0012-7094-00-10131-7}],
[\href{https://arxiv.org/abs/math/9908171}{\tt arXiv:math/9908171}].




\bibitem[Ko96]{Kochman96}
S. Kochman,
{\it\color{darkblue}Bordism, Stable Homotopy and Adams Spectral Sequences}, 
Fields Institute Monographs {\bf 7}
American Mathematical Society (1996),
[\href{https://bookstore.ams.org/fim-7}{\tt fim-7}].

\bibitem[Kos93]{Kosinski93}
A. Kosinski, 
{\it\color{darkblue}Differential manifolds}, 
Academic Press (1993),
[\href{https://www.sciencedirect.com/bookseries/pure-and-applied-mathematics/vol/138/suppl/C}{\tt ISBN:978-0-12-421850-5}].


\bibitem[La17]{Landsman17}
K. Landsman, 
{\it\color{darkblue}Foundations of quantum theory -- From classical concepts to Operator algebras}, Springer Open (2017), 
[\href{https://link.springer.com/book/10.1007/978-3-319-51777-3}{\tt doi:10.1007/978-3-319-51777-3}].

\bibitem[Lo24]{Lobb24}
A. Lobb,
{\it\color{darkblue}A feeling for Khovanov homology}, Notices Amer. Math. Soc. {\bf 71} 5 (2024), 
[\href{https://doi.org/10.1090/noti2928}{\tt doi:10.1090/noti2928}].



\bibitem[Ly03]{Lyons03}
D. W. Lyons,
{\it\color{darkblue}An Elementary Introduction to the Hopf Fibration}, 
Mathematics Magazine {\bf 76} 2 (2003), 87-98, 
[\href{https://doi.org/10.2307/3219300}{\tt doi:10.2307/3219300}],
[\href{https://www.jstor.org/stable/3219300}{\tt jstor:3219300}],
[\href{https://arxiv.org/abs/2212.01642}{\tt arXiv:2212.01642}].



\bibitem[MS04]{MantonSutcliffe04}
N. Manton and P. Sutcliffe, 
{\it\color{darkblue}Topological Solitons}, 
Cambridge University Press (2004),
\newline
[\href{https://doi.org/10.1017/CBO9780511617034}{\tt doi:10.1017/CBO9780511617034}].


\bibitem[McD75]{McDuff75}
D. McDuff, 
{\it\color{darkblue}Configuration spaces of positive and negative particles}, 
Topology {\bf 14} 1 (1975), 91-107, \newline 
[\href{https://doi.org/10.1016/0040-9383(75)90038-5}{\tt doi:10.1016/0040-9383(75)90038-5}].

\bibitem[Mey95]{Meyer95}
P.-A. Meyer, 
{\it\color{darkblue}Quantum Probability for Probabilists}, Lecture Notes in Mathematics {\bf 1538}, Springer (1995),
[\href{https://link.springer.com/book/10.1007/BFb0084701}{\tt doi:10.1007/BFb0084701}].


\bibitem[MPW19]{MPW19}
M. Mezei, S. S. Pufu, and Y. Wang,
{\it  \color{darkblue} Chern-Simons theory from M5-branes and calibrated M2-branes}, J. High Energ. Phys. {\bf 2019}  (2019) 165, 
[\href{https://doi.org/10.1007/JHEP08(2019)165}{\tt doi:10.1007/JHEP08(2019)165}],
[\href{https://arxiv.org/abs/1812.07572}{\tt arXiv:1812.07572}].




\bibitem[Mi97]{Milnor97}
J. Milnor, 
{\it\color{darkblue}Topology from the differentiable viewpoint}, 
Princeton University Press (1997), \newline 
[\href{https://press.princeton.edu/books/paperback/9780691048338/topology-from-the-differentiable-viewpoint}{\tt ISBN:9780691048338}].


\bibitem[MR23]{MoravaRolfsen23}
J. Morava and D. Rolfsen, 
{\it\color{darkblue}Toward the group completion of the Burau representation}, 
Trans. Amer. Math. Soc. {\bf 376} (2023), 1845-1865,
[\href{https://doi.org/10.1090/tran/8796}{\tt doi:10.1090/tran/8796}],
[\href{https://arxiv.org/abs/1809.01994}{\tt arXiv:1809.01994}].











\bibitem[NSSFD08]{NSSFD08}
C. Nayak, S. H. Simon, A. Stern, M. Freedman,
and S. Das Sarma, 
{\it\color{darkblue}Non-Abelian Anyons and Topological Quantum Computation}, 
Rev. Mod. Phys. {\bf 80} (2008) 1083,
[\href{https://doi.org/10.1103/RevModPhys.80.1083}{\tt doi:10.1103/RevModPhys.80.1083}],
\newline 
[\href{https://arxiv.org/abs/0707.1889}{\tt arXiv:0707.1889}].



\bibitem[Oh01]{Ohtsuki01}
T. Ohtsuki, 
{\it\color{darkblue}Quantum Invariants -- A Study of Knots, 3-Manifolds, and Their Sets}, World Scientific (2001),  \newline 
[\href{https://doi.org/10.1142/4746}{\tt doi:10.1142/4746}].


\bibitem[Ok05]{Okuyama05}
S. Okuyama, 
{\it\color{darkblue}The space of intervals in a Euclidean space}, 
Algebr. Geom. Topol. {\bf 5} (2005), 1555-1572, \newline 
[\href{https://doi.org/10.2140/agt.2005.5.1555}{\tt doi:10.2140/agt.2005.5.1555}],
[\href{https://arxiv.org/abs/math/0511645}{\tt arXiv:math/0511645}].

\bibitem[Ok18]{Okuyama18}
S. Okuyama,
{\it\color{darkblue}Configuration space of intervals with partially summable labels}, talk at Shinshu University (2018).




\bibitem[Po10]{Polyak10}
M. Polyak,
{\it\color{darkblue}Minimal generating sets of Reidemeister moves}, 
Quantum Topol. {\bf 1} (2010), 399–411,
\newline 
[\href{https://doi.org/10.4171/QT/10}{\tt doi:10.4171/QT/10}],
[\href{https://arxiv.org/abs/0908.3127}{\tt arXiv:0908.3127}].



\bibitem[Pol88]{Polyakov88}
A. M. Polyakov, 
{\it\color{darkblue}Fermi-Bose transmutation induced by gauge fields}, Mod. Phys. Lett. A {\bf 03} 03 (1988), 325- 328, 
[\href{https://doi.org/10.1142/S0217732388000398}{\tt doi:10.1142/S0217732388000398}].

\bibitem[Pon38]{Pontrjagin38}
L. Pontrjagin,
{\it\color{darkblue}Classification of continuous maps of a complex into a sphere, Communication I},
Doklady Akademii Nauk SSSR {\bf 19} 3  (1938), 147-149, 
[\href{https://www.taylorfrancis.com/chapters/edit/10.1201/9780367813758-13/classification-continuous-mappings-complex-sphere-communication-gamkrelidze}{\tt doi:10.1201/9780367813758-13}]. 

\bibitem[Pon55]{Pontrjagin55}
L. Pontrjagin, 
{\it\color{darkblue}Smooth manifolds and their applications in Homotopy theory}, 
Trudy Mat. Inst. im Steklov {\bf 45} Izdat. Akad. Nauk. USSR, Moscow (1955), 
AMS Translation Series {\bf 2} 11 (1959), \newline 
[\href{https://www.worldscientific.com/doi/abs/10.1142/9789812772107_0001}{\tt doi:10.1142/9789812772107\_0001}].


\bibitem[RTY13]{RaduTchrakianYang13}
E. Radu, D. H. Tchrakian, and Y. Yang,
{\it\color{darkblue}Abelian Hopfions of the $\mathbb{CP}^n$ model on $\mathbb{R}^{2n+1}$ and a fractionally powered topological lower bound}, 
Nucl. Phys. B {\bf 875}  (2013), 388-407,
[\href{https://arxiv.org/abs/1305.4784}{\tt arXiv:1305.4784}], \newline 
[\href{https://doi.org/10.1016/j.nuclphysb.2013.07.006}{\tt doi:10.1016/j.nuclphysb.2013.07.006}].





\bibitem[Ro16]{Rowell16}
E. C. Rowell, 
{\it\color{darkblue}An Invitation to the Mathematics of Topological Quantum Computation}, 
J. Phys. Conf. Ser. {\bf 698} (2016) 012012, 
[\href{https://iopscience.iop.org/article/10.1088/1742-6596/698/1/012012}{\tt doi:10.1088/1742-6596/698/1/012012}].

\bibitem[Ro22]{Rowell22}
E. C. Rowell, 
{\it\color{darkblue}Braids, Motions and Topological Quantum Computing},
[\href{https://arxiv.org/abs/2208.11762}{\tt arXiv:2208.11762}].

\bibitem[RW18]{RowellWang18}
E. C. Rowell and Z. Wang, 
{\it\color{darkblue}Mathematics of Topological Quantum Computing}, Bull. Amer. Math. Soc. {\bf 55} (2018), 183-238,
[\href{https://doi.org/10.1090/bull/1605}{\tt doi:10.1090/bull/1605}],
[\href{https://arxiv.org/abs/1705.06206}{\tt arXiv:1705.06206}].








\bibitem[SS23]{SS23-Mf}
H. Sati and U. Schreiber, 
{\it\color{darkblue}M/F-Theory as Mf-Theory}, Rev. Math. Phys. {\bf 35} (2023) 2350028, \newline 
[\href{https://doi.org/10.1142/S0129055X23500289}{\tt doi:10.1142/S0129055X23500289}],
[\href{https://arxiv.org/abs/2103.01877}{\tt arXiv:2103.01877}].




\bibitem[SS24]{SS23-Obs}
H. Sati and U. Schreiber,
{\it\color{darkblue}Quantum Observables of Quantized Fluxes},
Ann. Henri Poincar{\'e} (2024),
\newline 
[\href{https://doi.org/10.1007/s00023-024-01517-z}{\tt doi:10.1007/s00023-024-01517-z}],
[\href{https://arxiv.org/abs/2312.13037}{\tt arXiv:2312.13037}]. 



\bibitem[SS25a]{SS24-Flux}
H. Sati and U. Schreiber,
{\it\color{darkblue}Flux quantization},
Encyclopedia of Mathematical Physics 2nd ed. 
{\bf 4}, Elsevier (2025), 281-324,
[\href{doi:10.1016/B978-0-323-95703-8.00078-1}{\tt doi:10.1016/B978-0-323-95703-8.00078-1}],
[\href{https://arxiv.org/abs/2402.18473}{\tt arXiv:2402.18473}].


\bibitem[SS25b]{SS25-Seifert}
H. Sati and U. Schreiber:
{\it \color{darkblue} Anyons on M5-Probes of Seifert 3-Orbifolds via Flux-Quantization},
Lett. Math. Phys. {\bf 115} (2025) 36,
[\href{https://arxiv.org/abs/2411.16852}{\tt arXiv:2411.16852}],
[\href{https://doi.org/10.1007/s11005-025-01918-z}{\tt doi:10.1007/s11005-025-01918-z}].

\bibitem[SS25c]{SS25-Srni}
H. Sati and U. Schreiber,
{\it\color{darkblue} Engineering of Anyons on M5-Probes
via Flux-Quantization},
SciPost Physics Lecture Notes (2025, in press),
[\href{https://arxiv.org/abs/2501.17927}{\tt arXiv:2501.17927}]
[\href{https://scipost.org/submissions/2501.17927v1}{\tt scipost.org/submissions/2501.17927v1}].

\bibitem[SS25d]{SS25-TEC}
H. Sati and U. Schreiber:
{\it\color{darkblue}The Character Map in Twisted Equivariant Nonabelian Cohomology},
in: {\it Applied Algebraic Topology}, special issue of
Beijing Journal of Pure and Applied Mathematics (2025, in press)
[\href{https://arxiv.org/abs/2011.06533}{\tt arXiv:2011.06533}].


\bibitem[SS25e]{SS25-FQHViaAlgTop}
H. Sati and U. Schreiber,
{\it\color{darkblue}Fractional Quantum Hall anyons via the Algebraic Topology of exotic Flux Quanta},
[\href{https://arxiv.org/abs/2505.22144}{\tt arXiv:2505.22144}].

\bibitem[SS26]{SS26-Orb}
H. Sati and U. Schreiber,
{\it\color{darkblue}Geometric Orbifold Cohomology},
CRC Press (2026, in press)
\newline
[\href{https://ncatlab.org/schreiber/show/Geometric+Orbifold+Cohomology}{\tt ncatlab.org/schreiber/show/OrbCoh}], 
[\href{https://arxiv.org/abs/2008.01101}{\tt arXiv:2008.01101}],





\bibitem[Se73]{Segal73}
G. Segal, 
{\it\color{darkblue}Configuration-spaces and iterated loop-spaces}, 
Invent. Math. {\bf 21} (1973), 213-221, \newline 
[\href{https://doi.org/10.1007/BF01390197}{\tt doi:10.1007/BF01390197}].

\bibitem[Sel11]{Selinger11}
P. Selinger, 
{\it \color{darkblue} A survey of graphical languages for monoidal categories}, 
Springer Lecture Notes in Physics {\bf 813} (2011), 289-355, 
 [\href{https://doi.org/10.1007/978-3-642-12821-9_4}{\tt doi:10.1007/978-3-642-12821-9\_4}],
[\href{https://arxiv.org/abs/0908.3347}{\tt arXiv:0908.3347}].

\bibitem[Sh94]{Shum94}
M. C. Shum, 
{\it\color{darkblue}Tortile tensor categories}, J. Pure Appl. Algebra {\bf 93} (1994), 57-110, \newline 
[\href{https://doi.org/10.1016/0022-4049(92)00039-T}{\tt doi:10.1016/0022-4049(92)00039-T}].


\bibitem[So23]{Sossinsky23}
A. B. Sossinsky, 
{\it \color{darkblue} Knots, Links and Their Invariants: An Elementary Course in Contemporary Knot Theory}, 
The Student Mathematical Library {\bf 101},
AMS (2023), 
[\href{https://doi.org/10.1090/stml/101}{\tt doi:10.1090/stml/101}].



\bibitem[Sp49]{Spanier49}
E. Spanier, 
{\it\color{darkblue}Borsuk’s Cohomotopy Groups}, 
Ann. Math. {\bf 50}  (1949), 203-245, 
[\href{http://www.jstor.org/stable/1969362}{\tt jstor:1969362}].





\bibitem[St08]{Strocchi08}
F. Strocchi, 
{\it \color{darkblue} An introduction to the mathematical structure of quantum mechanics}, 
Advanced Series in Mathematical Physics {\bf 28}, 
World Scientific (2008),
[\href{https://doi.org/10.1142/7038}{\tt doi:10.1142/7038}].


\bibitem[St13]{Strocchi13}
F. Strocchi, 
{\it\color{darkblue}An Introduction to Non-Perturbative Foundations of Quantum Field Theory}, 
Oxford University Press (2013),
[\href{https://doi.org/10.1093/acprof:oso/9780199671571.001.0001}{\tt doi:10.1093/acprof:oso/9780199671571.001.0001}].



\bibitem[Ve12]{Veneziano12}
G. Veneziano,
{\it \color{darkblue} Rise and fall of the hadronic string}, Chapter 2 in: 
{\it The Birth of String Theory},
Cambridge University Press (2012), 
[\href{https://doi.org/10.1017/CBO9780511977725.004}{\tt doi:10.1017/CBO9780511977725.004}].


\bibitem[Wa10]{Warner10}
G. Warner, 
{\it\color{darkblue}$C^\ast$-Algebras}, EPrint Collection, University of Washington (2010), 
[\href{http://hdl.handle.net/1773/16302}{\tt hdl:1773/16302}].



\bibitem[WZ83]{WilzcekZee83}
F. Wilczek and A. Zee,
{\it\color{darkblue}Linking Numbers, Spin, and Statistics of Solitons}, 
Phys. Rev. Lett. {\bf 51} (1983) 2250,
[\href{https://doi.org/10.1103/PhysRevLett.51.2250}{\tt doi:10.1103/PhysRevLett.51.2250}].

\bibitem[Wi20]{Williams20}
L. Williams, 
{\it\color{darkblue}Configuration Spaces for the Working Undergraduate}, 
Rose-Hulman Undergrad. Math. J. {\bf 21} 1 (2020) 8, 
[\href{https://scholar.rose-hulman.edu/rhumj/vol21/iss1/8}{\tt rhumj:vol21/iss1/8}],
[\href{https://arxiv.org/abs/1911.11186}{\tt arXiv:1911.11186}].



\bibitem[Wi89]{Witten89}
E. Witten, 
{\it\color{darkblue}Quantum Field Theory and the Jones Polynomial}, Commun. Math. Phys. {\bf 121} 3 (1989), 351-399,
[\href{https://doi.org/10.1007/BF01217730}{\tt doi:10.1007/BF01217730}].





\bibitem[Ye01]{Yetter01}
D. N. Yetter,
{\it\color{darkblue}Functorial Knot Theory -- Categories of Tangles, Coherence, Categorical Deformations, and Topological Invariants}, 
Series on Knots and Everything {\bf 26}, 
World Scientific (2001),
[\href{https://doi.org/10.1142/4542}{\tt doi:10.1142/4542}].





\bibitem[Zh93]{Zhu93}
K. Zhu, 
{\it\color{darkblue}An Introduction to Operator Algebras}, CRC Press (1993),
[\href{https://www.routledge.com/An-Introduction-to-Operator-Algebras/Zhu/p/book/9780849378751}{\tt ISBN:9780849378751}].

\end{thebibliography}
\end{document}